\documentclass[twoside]{article}

\usepackage[marginparwidth=70pt]{geometry}
\usepackage{amssymb,latexsym,amsmath,amsthm,alltt}
\usepackage{stmaryrd}
\usepackage{bm}
\usepackage{proof,fancyhdr,currfile}
\usepackage{tikz-cd}
\usepackage{url}
\usepackage{todonotes}
\usepackage{datetime2}
\usepackage{mathtools}
\usepackage{xspace}
\usepackage[normalem]{ulem}
\usepackage{hyperref}

\newtheorem{thm}{Theorem}[section]
\newtheorem{lem}[thm]{Lemma}
\newtheorem{cor}[thm]{Corollary}

\theoremstyle{definition}
\newtheorem{defn}[thm]{Definition}

\theoremstyle{remark}
\newtheorem{rem}[thm]{Remark}
\newtheorem{expl}[thm]{Example}

\newcommand{\bvar}{{\textbf{bvar}}\xspace}
\newcommand{\fvar}{{\textbf{fvar}}\xspace}
\newcommand{\sth}{\theta}

\newcommand{\mea}{\llbracket}
\newcommand{\ning}{\rrbracket}
\newcommand{\me}[1]{\mea{#1}\ning}
\newcommand{\bracq}[1]{[{#1}]_i}

\newcommand{\brseq}[2]{\bracq{\seq{{#1}}{{#2}}}}
\newcommand{\braseq}[2]{\bracq{\seq{{#1}}{{#2}}}}
\newcommand{\idseq}[2]{\ensuremath {\langle\!\langle {#1}|{#2}\rangle\!\rangle}}

\newcommand{\mee}{\mea{\,}\ning}  
\newcommand{\fv}[1]{\ensuremath{{\sf fv}({#1})}} 

\newcommand{\imp}{\rightarrow}
\newcommand{\arr}{\!\rightarrow\!}

\newcommand{\NULL}{{\ensuremath \Box}}
\newcommand{\pmi}{\leftarrow}

\newcommand{\wimp}{\,{\Rightarrow}\,}

\newcommand{\cfpf}{\ensuremath{
\:\vdash^{^{\!\!\!_{\mbox{\scriptsize \sf cf}}\!}}}}

\newcommand{\seq}[2]{{#1}\vdash {#2}}
\newcommand{\cfseq}[2]{{#1}\cfpf {#2}}

\newcommand{\seqi}[2]{{#1}\vdash_i {#2}}
\newcommand{\cfseqi}[2]{{#1}\cfpf_i {#2}}

\newcommand{\seqx}[3]{{#2}\vdash_{#1} {#3}}
\newcommand{\cfseqx}[3]{{#2}\cfpf_{#1} {#3}}

\newcommand{\cg}{\mbox{${\cal G}$}}
\newcommand{\cv}{\mbox{${\cal V}$}}
\newcommand{\cl}{\mbox{${\cal L}$}}

\newcommand{\nf}[1]{{\bf nf(}{#1}{\bf )}}

\newcommand{\fa}{\ensuremath{\mathfrak A}}
\newcommand{\fb}{\ensuremath{\mathfrak B}}

\newcommand{\fd}{\mbox{${\mathfrak D} $}}
\newcommand{\sfd}{\mbox{${\mathfrak D}_{\sf s}$}}

\newcommand{\case}[1]{\vspace{0.2cm}\noindent\underline{#1}:\par\vspace{0.1cm}}

\newcommand{\sem}[2]{{\me{#1}_{#2}}}

\newcommand{\nat}{\mathbb{N}} 

\newcommand{\nn}{\ensuremath {\mathbb N }} 

\newcommand{\bbp}{\ensuremath {\mathbb P}}

\newcommand{\bool}{\ensuremath {\mathbb B }}

\newcommand{\nullset}{\varnothing}        

\newcommand{\warn}[1]{}

\newcommand{\cat}[1]{\mathbb{#1}}

\newcommand{\ids}{\ensuremath{{\mathit id}}}

\newcommand{\oml}{\text{$\Omega_\ell$}}  
\newcommand{\omfib}{\ensuremath{\widetilde{\Omega}}}

\newcommand{\fri}{\ensuremath{\mathfrak I}}

\newcommand{\clat}{c$\mathbb{L}$AT} 

\newcommand{\ang}[1]{\langle {#1} \rangle}   

\newcommand{\cc}{\mbox{${\cal C}\,$}}

\newcommand{\ct}{\mbox{${\cal T}\,$}}

\newcommand{\bd}{{\bf d}}         %
\newlength{\strikewidth}

\newcommand{\sfrac}[2]{\frac{\textstyle{#1}}{\textstyle{#2}}}

\newcommand{\sequent}[2]{{#1}\,\vdash\,{#2}}
\newcommand{\horseshoe}{\supset}
\newcommand{\hoe}{\horseshoe}

\newcommand{\dee}{{\sf D}}
\newcommand{\gee}{{\sf g}}
\newcommand{\aee}{{\sf a}}
\newcommand{\ree}{{\sf r}}
\newcommand{\pee}{{\sf p}}

\newcommand{\tdee}[2]{\dee_{{#1},{#2}}}
\newcommand{\deeip}{\tdee{i}{\pee}}
\newcommand{\deep}{\dee_\pee}

\newcommand{\deeig}{\tdee{i}{\gee}}

\newcommand{\deeih}{\tdee{i}{\hee}}
\newcommand{\deeia}{\tdee{i}{\aee}}
\newcommand{\deeir}{\tdee{i}{\ree}}

\newcommand{\deegg}{\dee_{\gee}}

\newcommand{\kforces}{\Vdash}  
\newcommand{\mforces}{\ |\!\!\!\!\!\models} 

\newcommand{\app}{{\sf App}}

\newcommand{\const}{\mbox{\sf Const}}

\newcommand{\inc}{\mbox{\sf Inc}}

\newcommand{\aph}{\varphi}

\newcommand{\thoint}[3]{\hoint_{#1}({#2},{#3})}
\newcommand{\linden}[3]{\cl_{#1}({#2},{#3})}

\newcommand{\sthoint}[3]{{(\shoint)_{#1}}({#2},{#3})}

\newcommand{\whoint}{I}  
\newcommand{\ifix}{{\ensuremath{I^{^\ast}\!\!}}}  
\newcommand{\ifixx}{{\ensuremath{I^{^\ast}\!\!\!}}}  

\newcommand{\tua}{\ensuremath{\mathfrak{T}}}  

\newcommand{\whobot}{I_\bot}  
\newcommand{\wthoint}[3]{\whoint_{#1}({#2},{#3})}

\newcommand{\hoforce}[4]{{#1}_{#2}({#3},{#4}) = \top_\Omega}

\newcommand{\psem}[1]{\mea{#1}\ning}

\newcommand{\harr}{\rightarrow}

\newcommand{\sos}{{\Omega}_{\sf s}}

\renewcommand{\owedge}{\wedge_\Omega}

\newcommand{\oleq}{\leq_\Omega}
\newcommand{\sleq}{\leq_{\sf s}}
\newcommand{\tops}{\top_{\sf s}}

\newcommand{\bigland}{\bigwedge}
\newcommand{\biglor}{\bigvee}

\newcommand{\hemod}{\ensuremath{{\mathfrak E}} }

\newcommand{\vars}{\ensuremath{{\cal V\, }} }

\newcommand{\bleadsto}{\ensuremath{\mathbf \leadsto}}
\newcommand{\bleadston}[1]{\ensuremath{\stackrel{#1}{\bleadsto}}}

\newcommand{\bareres}[3]{{#1}\ \stackrel{#2}{\bleadsto} \ {#3}}

\newcommand{\resdots}[3]{{#1}\   \stackrel{}{\bleadsto} \,
        \stackrel{#2}{\cdots} \,  \stackrel{}{\bleadsto} \ {#3}}

\newcommand{\cutrest}[5]{{#1} \, \stackrel{}{\bleadsto} \,
  \stackrel{#2}{\cdots} \, \stackrel{}{\bleadsto} \,  {#3}\, \stackrel{#4}{\bleadsto} \, {#5}
  }

\newcommand{\mfa}{\ensuremath{{\mathfrak A}}}
\newcommand{\mfai}{\ensuremath{{\mathfrak A}_1}}
\newcommand{\mfaii}{\ensuremath{{\mathfrak A}_2}}

\newcommand{\mfb}{\ensuremath{{\mathfrak B}} }
\newcommand{\mfc}{\ensuremath{{\mathfrak C}} }

\newcommand{\univ}{\ensuremath{{\cal U}} }

\newcommand{\leqt}{\leq_{\cal T}}
\newcommand{\eqt}{\sim_{\cal T}}
\newcommand{\wedget}{\wedge_{\cal T}}
\newcommand{\veet}{\vee_{\cal T}}

\newcommand{\srest}[3]{{#1} \, \stackrel{#2}{\bleadsto} \, {#3} }

\newcommand{\nullres}[1]{{#1} \, \stackrel{}{\bleadsto} \,
            \cdots \, \stackrel{}{\bleadsto} \, \NULL}

\newcommand{\nullsres}[2]{{#1} \, {\bleadsto} \,
            \stackrel{{#2}}{\cdots} \, \stackrel{}{\bleadsto} \, \NULL}

\newcommand{\midnullsres}[2]{{#1} \, \bleadsto \,
             \stackrel{{#2}}{\cdots} \, \stackrel{}{\bleadsto} \, \NULL}
\newcommand{\midnullsrest}[4]{{#1} \, \stackrel{#2}{\bleadsto} \, {#3}\,
           \stackrel{}{\bleadsto} \,   \stackrel{#4}{\cdots} \
           \stackrel{}{\bleadsto} \, \NULL}
\newcommand{\midnullsrestarray}[4]{
\begin{array}{ll}
{#1} \,  \stackrel{#2}{\bleadsto} & \\
\hspace*{3em} \, {#3}\, & \stackrel{}{\bleadsto} \,   \stackrel{#4}{\cdots} \
           \stackrel{}{\bleadsto} \, \NULL
\end{array}
}

\newcommand{\midres}[3]{{#1}  \, \bleadsto \, \stackrel{#2}{\cdots}
 \, \bleadsto   \, {#3}}
\newcommand{\midsrest}[5]{{#1} \, \stackrel{#2}{\bleadsto} \, {#3}\,
           \stackrel{}{\bleadsto} \,   \stackrel{#4}{\cdots} \
           \stackrel{}{\bleadsto} \, {#5}}

\newcommand{\starmidres}[3]{{#1}  \, \stto \,
 \stackrel{#2}{\cdots}
 \, \stto   \, {#3}}

\newcommand{\starnullsres}[2]{{#1} \, \stto \,
             \stackrel{{#2}}{\cdots} \, \stto \, \NULL}

\newcommand{\stto}{\ensuremath{\stackrel{*}{\bleadsto}}}

\newcommand{\hhh}{\ensuremath{{\mathcal H}}}
\newcommand{\hee}{{\sf h}}

\newcommand{\hoint}{\mathcal{I}}
\newcommand{\shoint}{{\ensuremath{\mathcal{I}_{\sf s}}}}

\newcommand{\umod}{\mathfrak{U}}

\newcommand{\tmod}{\mathfrak{T}}
\newcommand{\dmod}{\mathfrak{D}}

\newcommand{\tp}[1]{\ensuremath{T^{#1}(I_\bot)}}
\newcommand{\tpn}{\tp{n}}
\newcommand{\tpni}{\tp{n+1}}
\newcommand{\hotpn}[3]{\tpn_{#1}({#2},{#3})}

\newcommand{\elab}{\textsf{elab}}
\newcommand{\down}{{\downarrow}}
\newcommand{\upvec}{{\uparrow}}

\newcommand{\hresgoal}[4]{\langle {#2} ; {#3} \mid {#4} \rangle}
\newcommand{\ipresgoal}[3]{\langle {#1} ; {#2} \mid {#3} \rangle}
\newcommand{\pgstate}[3]{\ipresgoal{#1}{#2}{#3}}

\newcommand{\stipresgoal}[3]{\mfa\otimes \langle {#1} ; {#2} \mid {#3} \rangle
\otimes\mfb }
\newcommand{\stsipresgoal}[4]{\mfa {#4}\otimes \langle {#1} ; {#2}
\mid {#3} \rangle  \otimes\mfb {#4} }

\newcommand{\eresgoal}[1]{\langle \;  \rangle}

\newcommand{\btraca}{\ensuremath{\mathbf \vdash\!\!\!\!\!\!\vdash}}

\newcommand{\resub}[1]{\ensuremath {\ \stackrel{{#1}}{\bleadsto \ }}}
\newcommand{\residots}[1]{\ensuremath {\ {\bleadsto}\stackrel{{#1}}{\cdots}
  \bleadsto\ }}

\newcommand{\vladres}[3]{\mfa ,\ {#1},\ \mfb \
  \stackrel{#2}{\bleadsto} \ \mfa ,\ {#3} ,\ \mfb}

\newcommand{\vladsub}[3]{\mfa ,\ {#1} ,\ \mfb \
  \stackrel{#2}{\bleadsto} \ \mfa{#2} ,\ {#3} ,\ \mfb{#2}}

\newcommand{\vladnsub}[2]{\mfa ,\ {#1} ,\ \mfb \
  \stackrel{#2}{\bleadsto} \ \mfa{#2} ,\ \mfb{#2}}

\newcommand{\sqq}{\sqsubseteq}

\newcommand{\?}{\mbox{\tt  ?-  }}

\newcommand{\pmonoid}{{\sf M}_{\pee}}

\newcommand{\ol}[1]{\overline{{#1}}}
\newcommand{\oland}{\ol{\wedge}}
\newcommand{\olor}{\ol{\vee}}

\newcommand{\olhorseshoe}{\ol{\hoe}}

\newcommand{\olSigma}{\ol{\Sigma}}

\newcommand{\olPi}{\ol{\Pi}}

\newcommand{\ssq}{\subseteq}

\newcommand{\epr}{\ensuremath{\nullset}} 

\newcommand{\mland}{\bm{\wedge}}
\newcommand{\mPi}{\bm{\Pi}}

\newcommand{\dom}{\mathsf{dom}}
\newcommand{\rng}{\mathsf{rng}}
\newcommand{\oc}{\mathsf{oc}}

\newcommand{\firstd}{\mathsf{first}}
\newcommand{\lastd}{\mathsf{last}}
\newcommand{\derstep}[1]{\mathrel{\stackrel{#1}{\leadsto}}} 
\newcommand{\derivation}[1]{\mathrel{\stackrel{#1}{\leadsto \cdots \leadsto}}} 

\newcommand{\ICTT}{\textsf{ICTT}\xspace}
\newcommand{\UCTT}{\textsf{UCTT}\xspace}
\newcommand{\HOHH}{\text{HOHH}\xspace}
\newcommand{\FOHH}{\text{FOHH}\xspace}
\newcommand{\HOH}{\text{HOH}\xspace}  

\renewcommand{\vladres}[3]{\mfa\otimes {#1}\otimes \mfb \ 
  \stackrel{#2}{\bleadsto} \ \mfa\otimes {#3}\otimes \mfb}

\renewcommand{\vladsub}[3]{\mfa\otimes {#1}\otimes \mfb \ 
  \stackrel{#2}{\bleadsto} \ \mfa{#2}\otimes {#3}\otimes \mfb{#2}}

\renewcommand{\vladnsub}[2]{\mfa\otimes {#1}\otimes \mfb \ 
  \stackrel{#2}{\bleadsto} \ \mfa{#2}\otimes \mfb{#2}}

\renewcommand{\hee}{{\sf h}}

\renewcommand{\bool}{o}
\newcommand{\jlfbox}[1]{} 
\newcommand{\gabox}[1]{}
\newcommand{\gaboxcut}[1]{}
\newcommand{\gaboxsubs}[1]{}

\newcommand{\jlbox}[1]{}

\allowdisplaybreaks
\begin{document}

\title{
\vspace*{-4em} {
Uniform Algebras:\\ [0.5em]
{\large Models and constructive Completeness for \\
Full, Simply Typed $\lambda$Prolog}
}
}

\author{Gianluca Amato\\
University of Chieti--Pescara\\
\mbox{\small \tt gianluca.amato@unich.it}
\and
Mary DeMarco\\
Wesleyan University\\
\mbox{\small \tt mdemarco@wesleyan.edu}
\and
James Lipton\\
Wesleyan University\\
\mbox{\small \tt jlipton@wesleyan.edu}
}
\date{May 23, 2024}

\maketitle

\tableofcontents

\newpage

\bibliographystyle{plain}

\begin{abstract}

This paper introduces a model theory for resolution on Higher Order
Hereditarily Harrop formulae (\HOHH), the logic underlying the
$\lambda$Prolog programming language, and proves soundness and
completeness of resolution. The semantics and the proof
of completeness of the formal system is shown in several ways,
suitably adapted to deal with the
impredicativity of higher-order logic, which rules out definitions of
truth based on induction on formula structure. First, we use the least
fixed point of a certain operator on interpretations, 
in the style of Apt and Van Emden, Then a constructive
completeness theorem is given
using a proof theoretic variant of the Lindenbaum algebra, which also
contains a new approach to establishing cut-elimination.

\end{abstract}

\section{Preface}
HOHH is an executable fragment
of an intuitionistic sequent presentation of Church's theory of types,
here called \ICTT, introduced by Miller and Nadathur in
the early 1990s and shown to be an instance of their notion of an
{\em abstract logic programming language\/} in
\cite{uniform}.

A number of fragments of \HOHH have been shown to be suitable for logic
programming,  among them the Higher-Order Horn fragment (\HOH) and
the First-Order Hereditarily Harrop formulas (\FOHH).

An enrichment of \HOHH, with higher-order unification and a certain
degree of type polymorphism, constitutes the framework for the $\lambda$Prolog
programming language. This language allows programming in typed
higher-order intuitionistic logic, with a declarative
treatment of modules via implication and with universally quantified
goals, and the ability to handle abstract syntax via lambda-tree syntax as
well as abstract data types. Recently there has been a lot of work
to embed it, or a variant such as ELPI, 
in proof assistants like {\sf Coq}
\cite{tassi19itp,coq-elpi20web,dunchev15lpar,tassi18coqpl}.

Model theories have been defined for some \HOHH fragments.
Miller \cite{modules} defined an indexed variant of Kripke semantics
in 1989 for the $\forall$-free part of \FOHH. DeMarco extended these
models to full \FOHH and \HOHH in \cite{DeMarco}, and defined an equivalent
Kripke semantics, in 1999. This was the first time a sound and complete
semantics appeared for (intensional) \HOHH with the Apt-Van Emden
bottom-up $T_P$ operator. This paper expands the results in DeMarco's
dissertation \cite{DeMarco} in considering other classes of models as well.

Independently, Kripke semantics for a core \HOHH \emph{with
  constraints} was given in \cite{lipton-nieva-tcs}. Implicitly, this
gives a Kripke model theory for the logic of $\lambda$Prolog by
choosing the right constraint system, with $\beta\eta$-reduction
handled by equational constraint manipulation. The applicability to
$\lambda$Prolog is indirect, as it is based on a constraint-based
sequent calculus and not on
resolution rules. As a result the semantics is not operational, as the
one introduced in this paper (based on Nadathur's \cite{nadathurPP}) is, and
does not have the flexibility of 
our broader class of algebraic models. Lastres
\cite{elastres-phd,LastresMoreno98} produced a denotational semantics
for \FOHH and gave definitions for the higher-order case, in
2002. 
Also, the classical higher-order Horn fragment was shown sound and
complete with respect to a declarative semantics introduced by Wolfram in
\cite{wolf94}, and independently by Blair and Bai in \cite{bb92}, and
Wadge in \cite{wadge91,wadge13}. Interestingly, \cite{wolf94}
builds on the Takahashi--Andrews V-complex
model theory \cite{takahashi,andrews71}
originally introduced to solve the Takeuti conjecture  for
cut-elimination in classical higher-order logic, inducting on types
to deal with the impossibility of inducting on higher-order formulas, which form an
impredicative class.  However, the larger
hereditarily Harrop fragment considered here requires intuitionistic
logic as discussed in the cited papers by Miller and Nadathur,
and, consequently, an intuitionistic
model theory. See \cite{ccct05}
for semantics of \ICTT, the
intuitionistic version of Church's type theory
of which \HOHH is an executable
fragment.

HOHH with uniform proofs is here
shown sound and complete with respect to {\em Uniform Algebras},
an indexed version of typed Omega-set semantics similar to that
introduced by Rasiowa and Sikorski \cite{RasiowaSikorski,TvD} and
Fourman and Scott \cite{FourmanScott}, for
higher-order logic,
and used to model \ICTT in \cite{ccct05}.
A bottom-up semantics is also developed here, based on these structures,
producing a minimal term model as the least fixed point of
an appropriate continuous operator.

The reader should note that \eqref{mforces}
\gabox{(1) and (2) have not been defined yet.}
essentially mimics the AUGMENT rule \cite{uniform} of the uniform proof theory for
$\lambda$Prolog  for which reason we  term it operational semantics.

It is easily shown sound and complete for
resolution, and admits a minimal Herbrand--style model, definable as a
suitable fixed point.
  What is harder is to show that it agrees with full semantic
logical consequence. In other words, that an equivalent model of type
(1) above exists.
The formal difference above between \eqref{kforces} and
\eqref{mforces} reflects, in the semantics,
the difference between  full provability and uniform provability.

The equivalence of the two notions of deduction was shown
syntactically in \cite{uniform}. Semantic equivalence for two
corresponding model theories (one for \HOHH and one for \ICTT)
was shown
directly for the \FOHH fragment in \cite{DeMarco}, and for the \ICTT/\HOHH
logics with constraints in \cite{lipton-nieva-tcs}.


We now proceed to define the syntax and proof theory of \HOHH. In
order to justify the presence of (numerically) indexed variables in the syntax we
 briefly discuss in the Appendix the problem that arises in proof search if we do not
 \gabox{We repeat ourselved at the end of the section}
enforce a device for keeping track of the change of signature that
takes place when solving universally quantified goals  to  avoid the
inconsistency that arises from 
allowing pre-existing logic variables to unify with constants.

In order to reflect this problem \emph{semantically}
we construct a countable hierarchy of universes.
We will need a supply of countably many fresh variables and constants
at each integer level $i$ and at every type $\alpha$, as is further
explained in the development of the \UCTT type theory, below.



\paragraph{Indexed models, operationality and bottom-up semantics}

In \cite{modules}, Miller produced an indexed variant of Kripke models
for a first-order fragment of $\lambda$Prolog in which the
conventional definition of truth in Kripke semantics at a possible
world $p$ in an arbitrary poset:
\begin{equation}
\label{kforces}
p\kforces A\horseshoe B \qquad \mbox{\it iff \quad for every}\quad
q\geq p,\qquad q\kforces A
\wimp q \kforces B
\end{equation}
is replaced by:
\begin{equation}
\label{mforces}
P \mforces D\horseshoe G \qquad\mbox{\it iff }\qquad P,D \mforces G
\end{equation}
where $P$ and $D$ must be program formulae and $G$ a goal.
Since the set of
possible worlds is fixed in advance, and consists of formulae,
this semantics --- as far as programs are concerned ---
has some of the flavor of a
term model.  In this paper, the model theory
is in no sense restricted to term models:
we supply arbitrary carriers for all types, including that
of program formulae. In this way we generalize Miller's semantics
in \cite{modules}, in addition to extending it
to higher order logic with implication and
universal quantification in goals as well.




We now proceed to define the syntax and proof theory of \HOHH. In
order to justify the presence of (numerically) indexed variables in the syntax we
 briefly discuss in the Appendix the problem that arises in proof
 search if we do not enforce a device for keeping track of the change
 of signature that takes place when solving universally quantified goals.




\section{Syntax of the Logic and the Programming Language}
\gabox{Change name of this section}

The fundamental aim of this paper is to define a semantics for an executable
fragment of logic, one that need not extend easily to the full
underlying logic. It is the underlying logic (\ICTT), however, that
gives declarative content to a logic program, and it is for this
reason that we term a model theory for \ICTT, as developed in
\cite{ccct05}, {\em declarative\/} and a model theory for \HOHH \textit{operational}.


Before presenting the calculus we review the definition of \HOHH
programs and goals, upon which \UCTT based.


\subsection{The \ICTT Calculus and \HOHH Programs}
\label{subsec:ICTT}
We recall the definitions of the Intuitionistic fragment of
Church's Theory of Types (\ICTT) upon which the \HOHH is built and programs
defined. 

Church's calculus consists of a collection of higher-order functional
types together with a simply typed lambda-calculus of terms.
Connectives and quantifiers are given explicitly as typed constants
 in the language. The set of base types must include a type of
logical formulas, as well as a type of individuals. Proofs are defined by a
sequent calculus for logical formulas. In more detail:
\begin{description}
    \item[Type expressions] are formed from
a set of {\em base types\/} (which must include at least the type $\iota$ of
individuals and $\bool$ of truth values, or {\it formulae\/}) via
the functional (or arrow)  type constructor:
if $\alpha$ and $\beta$ are types then so is
 $\alpha\arr\beta$. {\bf Church's notation} for functional type
expressions is sometimes more convenient because of its brevity:
$(\beta\alpha)$ denotes the type $\alpha\arr\beta$.
Association, in Church's notation, is to the {\it left\/}
whereas arrow terms associate to the right: $\alpha\beta\gamma$ denotes
$(\alpha\beta)\gamma$, meaning $\gamma\arr\beta\arr\alpha$,
which is the same as  $\gamma\arr(\beta\arr\alpha)$.
We will use his notation interchangeably with arrow notation,
especially with complex type expressions.
  \item[Terms] are always associated with a type expression, and are
built up from typed constants and typed variables (countably many  at each
type $\alpha$):
$x^1_\alpha,x^2_\alpha,\ldots,y^1_\alpha,y^2_\alpha,\ldots$
according to the following definition:
\begin{enumerate}
    \item A variable $x_\alpha$ is a term of type $\alpha$.
  \item   A constant of type $\alpha$ is a term of type $\alpha$.
  \item If $t_1$ is a term of type $\alpha\arr\beta$ and $t_2$ is a
  term of type $\alpha$ then $t_1 t_2$ is a term of type $\beta$,
  called an application.
  \item If $t$ is a term of type $\beta$ then $\lambda x_\alpha.t$ is
  a term of type $\alpha\arr\beta$, called an abstraction.
\end{enumerate}
Terms may be displayed subscripted with their associated type, when
convenient. We omit such subscripts from variables when clear from
context, or when the information is of no interest.

\item[Logical Constants] The following logical constants are available
to define formulas:
\[
 \top_\bool, \bot_\bool, \land_{\bool\bool\bool},
\lor_{\bool\bool\bool}, \horseshoe_{\bool\bool\bool},
\]
as well as
\[
\Sigma_{\bool (\bool\alpha)} \quad \mbox{\small  and }
\quad \Pi_{\bool (\bool\alpha)},
\]
the latter two existing for all types $\alpha$. 

Note that because of currying in Church's type theory, the constants
$\wedge,\vee$ corresponding to the usual binary logical connectives,
have higher types $\bool\bool\bool$,
i.e. $\bool\arr(\bool\arr\bool)$.
\end{description}

The meaning of the logical constants $\land$, $\lor$ and $\horseshoe$ is
evident, while $\Sigma$ and $\Pi$ are used to define the quantifiers
$\exists$ and $\forall$: $\exists x_\alpha. P_\bool$ is an
abbreviation for $\Sigma_{\bool (\bool\alpha)} \lambda x_\alpha
. P_{\bool}$ and $\forall x_\alpha. P_\bool$ is an
abbreviation for $\Pi_{\bool (\bool\alpha)} \lambda x_\alpha
. P_{\bool}$.  Also, $f(x_\alpha)$ is an
abbreviation of $f_{\beta\alpha}  x_\alpha$, for some $\beta$.
We will often write logical formulas in infix form, e.g.\@ $A\hoe B,
A\wedge B$. This should be viewed as informal notation for the true
terms involved: $\hoe_{ooo}AB, \wedge_{ooo}AB$.

\paragraph{Reduction}
Two terms are called \emph{$\alpha$-convertible} if it is possible to obtain
one from the other by renaming bound variables.   
  We recall that a $\beta$-{\em redex\/} is a term  $(\lambda
  x.t ) u$,
formed by application of an abstraction to any compatibly typed
  term. The {\em contractum\/} (or $\beta$-contraction)
 of such a term is defined to be
 any term $t[u/x]$ resulting from the substitution (after renaming of
  bound variables if necessary to avoid capture) of  $u$ for all free
  occurrences of $x$ in $t$.  $\beta$-contraction also refers to the
  binary relation
  between redices and contracta. The transitive closure of this binary
  relation and its extension to terms containing redices in subterms
  is called $\beta$-reduction. Contraction
  of  $\eta$-redices $\lambda x.tx$ to $t$ when $x$ does not occur
  freely in $t$ induces $\eta$-reduction. In this paper we will refer
  to the combination of both reductions, $\beta\eta$-reduction as
  $\lambda$-reduction, and the reflexive, symmetric, transitive
  closure of this relation as $\lambda$-{\bf equivalence.}
  \gabox{I think it would be better to ay that $\lambda$-equivalence also include $\alpha$-conversion}
  If a term cannot be $\lambda$-reduced we say it is in \emph{normal
    form}.


    In Church's Theory, as in Andrews, Miller and Nadathur's work, a
  canonical linear ordering of variables is used to ensure the existence of
  unique  \emph{principal} normal forms%
  \footnote{A term in principal normal form, as defined in
  \cite{church40} is a term in  normal form in which no variable occurs
  both free and bound, and in which bound variables must  be
  consecutive variables in increasing order, with no omissions, save
  those variables which occur freely in the term. In \cite{church40} it
  is shown that for each normal term there is a unique equivalent term in
  principal normal form.}
$\nf{t}$ for  each term $t$, 
  by enforcing a specific choice of bound variables. Since we do not
  want to bother with  $\lambda$-equivalence and $\alpha$-conversion,
  in the following we assume that all  terms are always in principal
  normal form, and automatically converted to principal 
  normal form at the end of every operation.

 \begin{defn}
A {\bf formula} is any term of type $\bool$.
 An atomic formula is a term of type $\bool$ which is not
$\lambda$-equivalent to any term whose head is a logical
constant.
\end{defn}

An atomic formula (in normal form) is either a variable $X_o$ or a term of the
form $p t_1\cdots t_m$ with the $t_i$ (each of type $\alpha_i$)
in normal form and $p$ a non-logical
constant or variable of type $\alpha_1\imp\cdots\alpha_m\imp o$.
If $p$ is a constant, the formula is called \emph{rigid}, otherwise \emph{flex}.

\begin{figure}[t]
  \protect
  \label{fig:ictt}
$$
\infer[\top_R
]
{\sequent{\Gamma}{\top}}
{}
\qquad
\infer[Ax
]
{\sequent{\Gamma, U}{U}}
{}
$$

$$
\infer[\land_L]
{\sequent{\Gamma, B \land C}{A}}
{\sequent{\Gamma, B, C}{A}}
\qquad
\infer[\land_R]
{\sequent{\Gamma}{B \land C}}
{\sequent{\Gamma}{B} & \sequent{\Gamma}{C}}
$$

$$
\infer[\lor_L]
{\sequent{\Gamma, B \lor C}{A}}
{\sequent{\Gamma, B}{A} & \sequent{\Gamma, C}{A}}
\qquad
\infer[\lor_R]
{\sequent{\Gamma}{B_1 \lor B_2}}
{\sequent{\Gamma}{B_i}}
$$

$$
\infer[\horseshoe_L]
{\sequent{\Gamma, B \horseshoe C}{A}}
{\sequent{\Gamma}{B} & \sequent{\Gamma, C}{A}}
\qquad
\infer[\horseshoe_R]
{\sequent{\Gamma}{B \horseshoe C}}
{\sequent{\Gamma, B}{C}}
$$

$$
\infer[\forall_L]
{\sequent{\Gamma, \forall x.P}{A}}
{\sequent{\Gamma, P[t/x]}{A}}
\qquad
\infer[\forall_R\,\ast]
{\sequent{\Gamma}{\forall x.P}}
{\sequent{\Gamma}{P[c/x]}}
$$

$$
\infer[\exists_L\,\ast]
{\sequent{\Gamma, \exists x.P}{A}}
{\sequent{\Gamma, P[c/x]}{A}}
\qquad
\infer[\exists_R]
{\sequent{\Gamma}{\exists x.P}}
{\sequent{\Gamma}{P[t/x]}}
$$

$$
\infer[\bot_R]
{\sequent{\Gamma}{B}}
{\sequent{\Gamma}{\bot}}
$$

\centering
{\small Constants appearing in rules marked with an asterisk are fresh.}
\caption{\label{fig:deduction}Higher-order Sequent Rules of $\ICTT^c$.}
\end{figure}

\subsection{Deduction}
We assume familiarity with Gentzen's sequent calculus, and in
particular the intuitionistic LJ (see e.g.\@ \cite{GLT}).
We present here, in Figure~\ref{fig:deduction},
a sequent system for
\ICTT, similar to that of e.g.\@ \cite{uniform} of Miller
et al. Following them, we modify the $\forall_R$ and
$\exists_L$ rules of the calculus to use fresh constants instead of
fresh eigenvariables. We call the modified calculus $\ICTT^c$. It
is easily seen equivalent to the more conventional version. Since all formulas are in principal normal
form, we avoid the need for Miller and Nadathur's $\lambda$-rule (which allows
replacement of any parts of a sequent by $\lambda$-equivalent
formulas). In the rest of the paper, with \ICTT we denote the $\ICTT^c$
calculus in Figure~\ref{fig:deduction}.

  In the sequent calculus displayed above
all {\em sequents\/} $\sequent{\Gamma }{A}$ consist
  of a {\it set\/} of formulas (antecedents) $\Gamma$ and a single
  formula $A$ (consequent). The symbol $U$ in the identity rule stands
for any atomic formula.
 The rules of inference shown in the
  figure consist of {\em premisses:\/} the zero, one or two sequents
 above the line and {\it conclusions:\/} the lone sequent below the
  line.

\begin{lem}
  Left-weakening, left-contraction and CUT are derived rules in \ICTT.
\end{lem}
The first claim is a very easy consequence, by induction on proof depth,
of the fact that the two initial sequents in \ICTT, the
sequents $\sequent{\Gamma}{\top}$ and $\sequent{\Gamma,U}{U}$ are allowed
for any set of formulas $\Gamma$.
The derivability of CUT in \ICTT is established in \cite{ccct05}.

\subsection{The \UCTT Calculus}
\paragraph*{Programs and Goals}
Let $\hhh$ be the set of {\em positive terms}, the terms which do not
include $\bot$, $\horseshoe$ or $\forall$; $A$ will
refer to an atomic formula over $\hhh$ and $A_r$ to a rigid atomic
formula over $\hhh$.
  Let $G$ refer to a goal formula, $D$ to a definite clause (that is, a single
program formula), and $P$ to a set of program formulae; $P \cup D$ or
$P,D$ when clear from context,
will stand for $P \cup \{D\}$. All programs and goals are assumed in
normal form.

\begin{defn}
\label{def:hoprograms}
HOHH program formulae $D$ and goal formulae $G$
are given by:
\begin{eqnarray*}
D & ::= & A_r \mid G\supset A_r \mid  D\land D \mid  \forall x.D
\\
G & ::=& \top \mid A\mid G\lor G \mid G\land G \mid \exists x.G \mid
D\supset G \mid \forall x.G
\end{eqnarray*}
\end{defn}



%
%

\paragraph*{Bound and free variables}
In order to simplify working with terms and substitutions, the language of the \UCTT 
fragment we now  define will contain countably many symbols to be used
{\em only\/} as bound variables in logical formulas,
and countably many symbols for free variables, also called \emph{logic
variables}. Free and bound variables are disjoint. The definition of principal normal
form is simplified, since no variables may occur both free and
bound. Occasionally we will need to be especially careful about free
and bound variables and our notational conventions. To this end we
make the following definition.
\begin{defn}
  We define the set \bvar of bound variables, and the set \fvar,
  of free, or logic variables. The \bvar set has countably many variables
  of each of the types specified below. 
\end{defn}
\begin{rem}
  \label{rem:vars}
The reader should be aware of certain uses of notation in this paper.
For example, when we define the {\bf instance} resolution rule below,
we replace a goal of the form $\exists x.G$ by the goal $G[t/x]$ where
$t$ is an arbitrary term \emph{all of whose free variables are in}
\fvar. We are replacing the \emph{freely occurring} \bvar
$x$ in $G$ by the term $t$. The result of this replacement is a goal
in which all freely occurring variables are in \fvar. In other words,
variables in \bvar may occur freely, but only in intermediate computations.
\end{rem}

\paragraph*{Subtypes of the formula type}
\gaboxcut{I propose to remove the type system and just use the type $o$, as in the standard \ICTT. Then, 
Positive formulae, rigid atoms, atoms, programs and goals are simply
the terms of the corresponding {\em subtypes\/} of Church's type $o$ of
logical formulae. Certain constants
Definition~\ref{def:hoprograms} gives the allowed formulas in \UCTT.}
and functions are barred from the calculus to make it impossible to
form illegal programs or goals. In particular, all terms appearing in
atomic formulae are positive. \UCTT has, in addition, the following restrictions

\begin{enumerate}

\item The base types are the types $\iota$ of
individuals, $\hee$ of (positive)
logical formulae, $\aee$ of atoms, $\ree$ of rigid
atoms, $\pee$ of programs and $\gee$ of goals, in addition to any
user-defined types.

\item There are inclusions between types as indicated in Figure~\ref{fig:UCTTdiagram}.
Thus, a term $t_\aee$ of atomic type may also be viewed as a term
of goal type.
   One way of formalizing this inclusion is to add constants such as 
   ${\sf Inc}_{\aee\ree}$, $\inc_{\pee\aee}$ ${\sf Inc}_{\gee\aee}$.
We will occasionally display these maps explicitly, but to save the
reader fatigue,  we will eventually  employ them tacitly
in the semantics (as inclusions). Other implicit inclusions are shown as dashed lines
above.
\item All other types are
functional types $\beta\alpha$.
\item The logical constants are
$\top_\gee$, $\top_\hee$, $\land_{\pee\pee\pee}$,
$\land_{\gee\gee\gee}$, $\land_{\hee\hee\hee}$, $\lor_{\gee\gee\gee}$,
$\lor_{\hee\hee\hee}$, $\horseshoe_{\gee\gee\pee}$,
$\horseshoe_{\pee\ree\gee}$, $\Sigma_{\gee(\gee\alpha)}$,
$\Sigma_{\hee(\hee\alpha)}$, $\Pi_{\gee(\gee\alpha)}$,
$\Pi_{\pee(\pee\alpha)}$ and $\Pi_{\hee(\hee\alpha)}$.

\item  Constants
of type $\ree\alpha_1,\cdots\alpha_n$ and both constants and variables of type
$\aee\alpha_1\cdots\alpha_n$ are allowed, for $n\geq 0$,
and $\alpha_i=\hee$ or a non-logical type.
\item There may be
no other constants or variables of types $\pee$ or $\gee$, nor of
functional types involving $\pee$ or $\gee$, other than those which
come from the included types 
$\ree\subseteq\aee\subseteq\gee$.
For example $X_\aee$ is a legal term of type $\gee$,
because it is a legal atom, but $X_{\pee\gee}$ is not, since it could
never have been formed in $\aee$.
\end{enumerate}
Note that because of the typing of $\hoe$ it is impossible
to construct negative formulae in $\hee$ or as subformulae of atomic formulae.

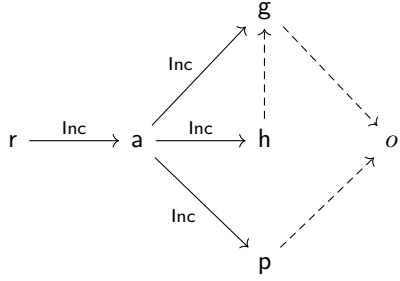
\begin{figure}  
  \centering
  \begin{tikzcd}[sep=8ex,arrow style=tikz,diagrams={>={Latex}}]
      & & \gee \arrow[dr, densely dashed] \\
      \ree \arrow[r, "\textsf{Inc}"] & \aee \arrow[ur, "\textsf{Inc}"] \arrow[r, "\textsf{Inc}"] \arrow[dr, swap, "\textsf{Inc}"]  & \hee \arrow[u, densely dashed] & o\\
      & & \pee \arrow[ur, densely dashed]
  \end{tikzcd}
  \caption{\label{fig:UCTTdiagram}Type Structure of \UCTT}
\end{figure}

\subsection{Levels and Signatures}
\label{hosigs}
\jlbox{Too much bla bla bla... in all the page}



As discussed in the Appendix~\ref{sec:varindex}
we will need to
construct a countable hierarchy of universes (signatures), in order to
avoid the  inconsistency that arises from
allowing pre-existing logic variables to unify with constants
introduced during execution for the sake of universal goals.

Let $\univ_0(\alpha)$ include all constants and variables in \bvar of type $\alpha$, along with
countably many logic variables in \fvar; constants and
variables of $\univ_0$ are labelled with the integer
$0$, indicated with a subscript unless clear from context
or immaterial to the discussion at hand.
The set $\univ_i(\alpha)$ will include
$\univ_{j}(\alpha)$ for all $j<i$ 
along with a new set of countably many logic
variables in \fvar and countably many fresh constants of type $\alpha$,
labelled with the index $i$.

A (positive) term over $\univ_i$ contains only typed constants and variables from $\bigcup_\alpha \univ_i(\alpha)$.  A program over $\univ_i$ is a program whose terms are all over $\univ_i$.

We will often refer to the {\em level\/} of a term, program, goal, set
of goals, or individual or sequences of states, to be defined
below. It always refers to the maximum label of any constant or
logic variable present in the corresponding structure.
If $t$ is a term, we write $\fv{t}$ to refer to the set of free
variables in $t$. If $S$ is a set of terms, then  $\fv{S}$ is the union of
all $\fv{t}$ with $t\in S$.
%

\paragraph{Substitutions}
Substitutions are functions from variables to
terms {\em with finite
support\/}. That is to say, off of a finite set of variables, a
substitution must coincide with the identity function. We will refer
to the set of variables {\em not fixed\/} by a substitution as its
{\it domain\/} or support.
The image of a substitution when restricted to its
set of support will be called its {\em range\/}. As is often done in
the literature, we will sometimes just mean the set of free variables
occurring in this image.

We will call a substitution  {\em legal} if its
range is in the set of {\em positive\/} terms,
it preserves type,
 and is level-preserving,
so that only a term over the universe $\univ_i$ is
substituted for a logic variable of label $i$. The level of a substitution is
the maximum level of all the terms in its range. Since substitutions
will almost always be legal in this paper (which will not stop us from
reminding the reader of their legality from time to time)
we distinguish notationally
between substitutions, represented by Greek letters
such as $\theta,\theta_1,\theta_2,
\ldots, \varphi,\varphi_1,\varphi_2\ldots$,
and {\em instantiations\/}
$Q[t/x]$ which denote the result of replacing all free occurrences of
$x$ in the goal, program, state or state vector $Q$ by the term $t$.
Unless explicitly stated, instantiations may violate the level
restriction imposed on substitutions, but never type preservation.
 If we call an instantiation
legal, we mean levels are preserved.

\paragraph*{Applying a substitution to a term.} If $\theta$ is a
substitution, we will abuse language and denote its lifting to a
function on {\em terms\/} with the same symbol. This lifting is
defined as follows
\begin{equation*}
\theta(t)  = t[\theta(x_1)/x_1, \ldots, \theta(x_n)/x_n]
\end{equation*}
where $x_1,\ldots,x_n$ are precisely the variables in the domain of $\theta$. Note
that no capturing of free variables may occur since we have disjoint sets for free and
bound variables.
We will write the application of a
substitution $\theta$ to a variable (or term)
$u$ as $u\theta$ as an alternative to $\theta(u)$.


Now that we have defined the application of a substitution to a term,
composition of substitutions makes sense.
We will write compositions of substitutions in diagrammatic form:
$\theta_1\theta_2$ is the substitution that maps a variable $x$ to
$\theta_2(\theta_1(x))$.

When a substitution is applied to a set of terms, it is understood to be applied to each term separately. This notation will be used in particular for programs.

In order to avoid many problems with possible variable
name clashes in deductions, we will make use of a particularly
useful kind of substitution.
\begin{defn}
  We will say a substitution is {\em safe\/} or
{\em idempotent\/}
if the free variables in
its range are disjoint from its domain.
\end{defn}
It is easy to see that a substitution $\theta$
is safe precisely when it is {\em idempotent\/} in the usual algebraic
sense, namely when $\theta\theta=\theta$.

\subsection{Proof procedures for \HOHH Formulae}

We now present two non-deterministic transition systems we call
RES(Y) and RES(t), along the lines of \cite{nadathurPP} but with some
significant notational differences, as a proof theory for \UCTT.
These differences are discussed in some detail in Section~\ref{subsec:pp} in the appendix.

Instead of the constrained sequent system that was used in
\cite{uniform} to define uniform proofs
we use an operational proof system so as
better to model the behavior of programs.  Rather than attempt to automate the
$\land_L$ and $\forall_L$ rules of the sequent calculus explicitly, we
decompose the program into simple clauses.


 We begin with a few technical lemmas.
  \begin{lem}
    \label{lem:well-order}
    The following complexity  measure defines a well-ordering on program
    formulas.
    \begin{itemize}
    \item $\delta(A_r) = 0$
    \item $\delta(G\hoe A_r) = 0$
    \item $\delta(D_1 \wedge D_2) = 1 + max\{\delta(D_1), \delta(D_2)\}$
    \item $\delta(\forall x D) = 1 + \delta(D)$
    \end{itemize}
      \end{lem}
  An easy consequence is
  \begin{lem}
    For all program formulas $D$ and terms $t$, $\delta(D[t/x]) = \delta(D)$.    
  \end{lem}
  
\begin{defn}[elaboration of a program]
  \label{hoextprog}
  Given a program formula $D$ we define the set of \emph{clauses} $\elab(D)$, the \emph{elaboration} of $D$, by induction on the complexity of $D$, as follows:
  \begin{itemize}
    \item $\elab(A_r) = \{ A_r \}$
    \item $\elab(G \horseshoe A_r) = \{ G \horseshoe A_r \}$
    \item $\elab(D_1 \wedge D_2) = \elab(D_1) \cup \elab(D_2)$
    \item $\elab(\forall x D) = \{ \forall x D'\mid D' \in
      \elab(D)\}.$
  \end{itemize}
  Moreover, if $P$ is a program, we have $\elab(P) = \bigcup_{D \in P} \elab(D)$.
\end{defn}

Formally, in order to respect the condition that only variables in \fvar may occur freely, the definiton of $\elab(\forall x D)$ should be  $\{ \forall x D'[x/X] \mid D' \in \elab(D[X/x]) \}$, where $X \in \fvar$ does not occur in $D$. However, when the formal definition is too cumbersome, like in this case, we prefer to use a more informal but intuitive notation, using the same symbol to denote both the bound variable in \bvar and the replaced variable in \fvar.

We will need some technical lemmas about $\elab$.

\begin{lem}
  \label{lem:elab-fv-compl}
  Given a program formula $D$ and a clause  $K \in \elab(D)$:
  \begin{itemize}
    \item the set of free variables occuring in $K$ is a subset of the set of free variables occuring in $D$;
    \item the complexity of $K$ is not greater than the complexity of $D$. 
  \end{itemize}
\end{lem}
\begin{proof}
  Immediate by induction on the complexity of $D$.
\end{proof}

\begin{lem}[clause instance]
\label{lem:clause-renaming}
Let $D$ be a program formula and $\theta$ be a substitution, either legal or not. If  $K \in \elab(P)$ then $K\theta\in \elab(P\theta)$.
\end{lem}
\begin{proof}
  The proof is by induction on the complexity of $D$. If $D = G \horseshoe A_r$, then $K = D$ and $\elab(D\theta)=\elab(G\theta \horseshoe A\theta) = \{ G\theta \horseshoe A\theta\} = \{ D\theta \}$. The case when $D$ is a rigid atom is similar to the previous one. If $D = D_1 \wedge D_2$, assume $K \in \elab(D_j)$ for $j \in \{1,2\}$. By inductive hypotehsis, $K\theta \in \elab(D_j\theta)$, hence $K\theta \in \elab(D_1\theta \wedge D_2\theta) =  \elab(D\theta)$. 
  
  Finally, if $D = \forall x \bar D$ and using the more formal
  definition given above, we have that $K = \forall x (\bar D'[x/X])$
  with $\bar D' \in \elab(\bar D[X/x])$ and $X \in \fvar$ does not
  occur in $\bar D$. Now consider a variable $w$ which does not occur
  in $\bar D$, $\bar D'$ and $\theta$. By inductive hypothesis, $\bar
  D'[w/X]  \in \elab(\bar D[X/x][w/X]) = \elab(\bar D[w/x])$. Again,
  by inductive hypothesis,
  \[\bar D'[w/X] \theta \in \elab(\bar
  D[w/x]\theta) = \elab(\bar D\theta[w/x]) .\]
  Therefore, by definition
  of the elaboration of a program,  
  $\forall x (\bar D'[w/X] \theta [x/w]) \in \elab(\forall x (\bar
    D\theta)) =  \elab((\forall x \bar D)\theta) =
  \elab(D\theta)$. Finally, note that
  \[
    \forall x (\bar D'[w/X] \theta [x/w]) = \forall x (\bar
    D'[w/X][x/w])\theta = (\forall x \bar D'[x/X])\theta = K\theta.\qedhere
    \] 
  \end{proof}

  In the following prove, we omit the renaming of a bound variable $x$ with the free variable
  $X$, and we use the same letter to denote both the free and bound variable, although they are
  technically different. The following lemma is a partial converse of the preceding one.
  \begin{lem}
    \label{lem:elab-shift}
    Let $D$ be a program formula, $t$ a term and $K$ a clause.
    If $K \in \elab(D[t/x])$ for any term $t$,
      then there is a clause $\forall x K'$ in $\elab(\forall x D)$ such
      that $K'[t/x] = K$.
    \end{lem}
    \begin{proof}
    We sketch a proof by induction on the program formula $D$. If $K$ is a
    clause in $\elab((G\hoe A_r)[t/x])$ then $K = (G\hoe A_r)[t/x]$
    and obviously $\forall x (G\hoe A_r) \in \elab(\forall x (G\hoe
    A_r))$. The case when $D$ is a rigid atom is similar to the previous one.
    The inductive cases are given by $D_1\wedge D_2$ and
    $\forall x D_0$. In the first case, $K \in \elab((D_1\wedge
    D_2)[t/x])$, hence  in one of $\elab((D_i[t/x])$ for $i= 1$ or $2$. By
    the induction hypothesis, there is a clause $\forall xK'$ in 
    $\elab(\forall x D_i)$ with $K'[t/x] = K$.
    But since $\elab(\forall x D_i) \subseteq
    \elab(\forall x (D_1\wedge D_2))$ we are done.
    Now suppose $K \in \elab(\forall y
    (D[t/x]))$ where since $y$ is in \bvar, $y$ is not free
    in $t$. Then $K \in \{\forall y
    D' \mid D'\in \elab(D[t/x])\}$, so $K = \forall y. D'$
    for some $D'$ in $\elab(D[t/x])$. By the induction hypothesis,
    there is $K' \in \elab(D)$ such that $K'[t/x] = D'$.
   Thus, $\forall x \forall y K'  \in \elab{\forall x D}$ and 
   $(\forall y K') [t/x] = K$.
  \end{proof}

\begin{lem}[clause generalization]
\label{lem:clause-generalization}
Let $D$ a program formula, $x$ a variable, $c$ a costant and $K \in \elab(D[c/x])$. Then there is a clause $\hat{K} \in \elab(D)$ such that $\hat{K}[c/x]=K$.
\end{lem}
\begin{proof}
  Follows by induction on the complexity of $D$, like the previous lemmas.
\end{proof}

\begin{lem}
\label{lem:clause-deriv}
  Suppose $K\in\elab(D)$ where $D$ is a program formula. If $\sequent{\Gamma, K}{G}$ then $\sequent{\Gamma, D}{G}$.
\end{lem}
\begin{proof}
  We proceed by induction on the height of the proof of $\sequent{\Gamma, K}{G}$. We consider several cases, according to the last rule used in this proof.
  
  If the last (and only) rule is $\top_R$, then the lemma holds trivially. 
  For all the other rules, the cases when $K$ is not the principal formula easily follows by induction. However, particular care must be given to the $\forall_R$ and $\exists_L$, due to the freshness conditions.
  \begin{itemize}
    \item if the last rule used is $\forall_R$, then we have
    \[
      \infer[\forall_R]
      {\sequent{\Gamma, K}{\forall x. G}}
      {\sequent{\Gamma, K}{G[c/x]}}
      \]
    where $c$ is a fresh constant. It easy to check \gabox{we need to prove this?} that in the proof of $\sequent{\Gamma, K}{G[c/x]}$ we can replace $c$ with any constant $c'$ to get a proof of $\sequent{\Gamma, K}{G[c'/x]}$. Therefore, consider any $c'$ which does not occur freely in $D$ or $\Gamma$. By inductive hypothesis, we have $\sequent{\Gamma, D}{G[c'/x]}$ and since $c'$ does not occur in $\Gamma, D$, by $\forall_R$ we get $\sequent{\Gamma, D}{\forall x G}$.
    \item if the last rule used is $\exists_L$, the proof proceeds as above.
  \end{itemize}

  When we restrict our attention to the case in which $K$ is the principal formula, only  $\textit{Ax}$ and left rules need to be considered. Moreover,  if the last rule is $\exists_L$, $\vee_L$ or $\wedge_L$, then $K$ cannot be the principal formula.
  \begin{itemize}
    \item If the last rule is $\textit{Ax}$, then $K = A_r$ is a rigid atom, and the proof follows by induction on the complexity of $D$.
    If $D = A_r$, then the result follows by left contraction. If $D = D_1 \wedge D_2$, then  $K \in \elab(D_i)$ for some $i \in \{1,2\}$. By inductive hypothesis, $D_i, K \vdash G$, the result follows from the $\wedge_L$ rule and left weakening. The other cases for $D$ are incompatible with the form of $K$.
    
    \item  If the last rule is $\forall_L$, i.e., $K = \forall x K'$, we have
      \[
      \infer[\forall_L]
      {\sequent{\Gamma, \forall x K'}{G}}
      {\sequent{\Gamma, K' [t/x]}{G}}
      \]
      We now proceed by induction on the complexity of $D$. The cases when $D$ is an atom or an implication are not consistent with $K$. If   $D = \forall x D'$ then $K' \in  \elab(D')$. By Lemma~\ref{lem:clause-renaming} we have that $K'[t/x] \in \elab(D'[t/x])$, hence by inductive hypothesis on the height of the proof, $\sequent{\Gamma, D'[t/x]}{G}$. The result follows from the $\forall_L$ rule. Finally, if $D = D_1 \wedge D_2$ then $\forall x  K' \in \elab(D_j)$ for $j \in \{1,2\}$. By inductive hypothesis on the structure of $D$, we have $\sequent{\Gamma, D_j}{G}$ and the result follows from the $\wedge_L$ rule and left weakening.
    \item
    If the last rule is $\horseshoe_L$, i.e., $K = G' \horseshoe A$, we have
    \[
    \infer[\horseshoe_L]
    {\sequent{\Gamma, G' \horseshoe A}{G}}
    {\sequent{\Gamma}{G'} & \sequent{\Gamma, A}{G}}
    \]
    We now proceed by induction on the complexity of $D$. If $D = G' \horseshoe A$, there is nothing to prove. If $D = D_1 \wedge D_2$, then $G' \horseshoe A \in \elab(D_i)$ for some $i$. By inductive hypothesis, we have a proof of $\Gamma, D_i \vdash G$, and we get the proof of $\Gamma, D \vdash G$ using the $\wedge_L$ rule and left weakening. The other cases are not consistent with the form of $K$. \qedhere
  \end{itemize}
\end{proof}

\begin{defn} Let the signatures be given as above. 

A {\bf state} is a triple
$\hresgoal{\cstar}{i}{P}{G}$ where $i$ is a non-negative integer,
 $P$ is a program and $G$ a goal,
both over the signature $\univ_i$; $i$ is referred to as the {\bf
state index}.

A {\bf state vector}
$\hresgoal{\cc_1}{i_1}{P_1}{G_1}, \dots,
\hresgoal{\cc_i}{i_j}{P_j}{G_j},
\dots, \hresgoal{\cc_n}{i_n}{P_n}{G_n}
$ is a finite sequence of states.  State vectors will be abbreviated to
$\ {\mathfrak A},\ \hresgoal{\cc}{i}{P}{G},\
{\mathfrak B}$, where ${\mathfrak A}$ and ${\mathfrak B}$ denote (possibly
empty) finite sequences of states. The empty or null
state vector is denoted $\NULL$.

\end{defn}

The full structure of state
vectors is required on several accounts:  indices allow
manipulation of universal quantification in goals, as discussed above;
programs must be displayed
because they are susceptible to change during execution,
as they may be augmented to prove conditional goals.
Vectors of multiple states are
necessary for the decomposition and subsequent tracking
of conjunctive
goals, since,
under resolution, the state $\ipresgoal{i}{P}{G_1\wedge G_2}$
spawns two independent states
$\ipresgoal{i}{P}{G_1}, \ipresgoal{i}{P}{G_2}$. These, in turn may
undergo independent updates of $i$ and $P$ depending on the structure
of their respective goals.

When a substitution is applied to a state vector
$\mfa$, it is understood to be applied to the program and to the goal
of each state in the vector.  For example, if $\mfa$ were the vector
$\hresgoal{\cc_1}{i_1}{P_1}{G_1}, \dots,
\hresgoal{\cc_n}{i_n}{P_n}{G_n}$, then $\mfa\sth$ would be
$\hresgoal{\cc_1}{i_1}{P_1\sth}{G_1\sth}, \dots,
\hresgoal{\cc_n}{i_n}{P_n\sth}{G_n\sth}$.

\paragraph{Resolution rules} The resolution reduction rules
for \HOHH include many of those familiar from the first-order
case.  We make minor changes for handling higher-order terms.
We will use the
tensor symbol  $\otimes$ rather than the comma to indicate
concatenation of states and state-vectors,
to make clear the distinction between
formalized conjunction of goals written with conjunctions, (or with
commas in traditional Horn program
code) and its meta-mathematical counterpart on states.

Formally, resolution \mbox{$\bleadsto$}, is a ternary relation on
\[
  {\textit{state vectors} \times \textit{subs} \times \textit{state vectors}},
\]
\gabox{Due to the freshness conditions, which make it depends on past history, formally $\bleadsto$ should  be defined as a relation on derivations and not on state vectors. This is what Gopalan avoid by explicitly keeping the set of used variables in a component of the state.}
but where we write
$\fa\stackrel{\theta}{\bleadsto}\fb$ instead of
$\bleadsto(\fa,\theta,\fb)$. 


Certain instances of resolution below are
annotated with connectives $\forall,\exists,\hoe,\wedge,\vee$. They
denote the identity substitution, although clearly they give more
information, namely the rule invoked.  All resolution rules below
show a transition of a state vector of the form
$\mfa\otimes\ipresgoal{i}{P}{G}\otimes\mfb$ to another state vector.
The distinguished state $\ipresgoal{i}{P}{G}$ in such a transition is
called the {\em selected state\/} of the resolution step, and the rule
used, the
{\em selected rule.\/}

%
%

\begin{description}
\item[null] $\vladnsub{\hresgoal{\cc}{i}{P}{\top}}{}$

deletion of $\top$ goals.

\item[true] $\vladsub{\hresgoal{\cc}{i}{P}{A}}{\sth}
{\hresgoal{\cc}{i}{P\sth}{\top}}$

where $A$ atomic, $A\sth = \top$ and $\theta$ is not the identity.

When $\sth$ is the identity substitution,
this is not a legal resolution step since it makes no progress. In practice,
the null step above would be performed instead.

\item[backchain]
$\vladsub{\hresgoal{\cc}{i}{P}{A'}}{\sth}
{\hresgoal{\cc}{i}{P\sth}{G\delta\sth}}$

where  $A'$ atomic and rigid, $\forall x_1 \cdots \forall x_n (G \horseshoe A) \in \elab(P)$, $\delta =\{w_1/x_1, \ldots, w_n/x_n\}$ is a substitution mapping variables $x_1,\ldots,x_n$ to fresh variables of level $i$ or less, and $\theta$ is a legal substitution such that $A'\sth = A\delta\sth$. Here \emph{fresh} means that the variable has never occured before in the derivation.

\item[axiom]
$\vladnsub{\hresgoal{\cc}{i}{P}{A'}}{\sth}$

where  $A'$ atomic and rigid, $\forall x_1 \cdots \forall x_n A \in \elab(P)$ and $\delta, \theta$ as for the backchain step: this is essentially backchain for clauses without body.

\item[or$_j$]
$\vladres{\hresgoal{\cc}{i}{P}{G_1 \lor G_2}}{\lor}{
\hresgoal{\cc}{i}{P}{G_j}}$ \qquad $(j\in \{1,2\})$

\item[and]
$\vladres{\hresgoal{\cc}{i}{P}{G_1 \land G_2}}{\land}{
\hresgoal{\cc}{i}{P}{G_1}\otimes \hresgoal{\cc}{i}{P}{G_2}}$

\item[instance]
$\vladres{\hresgoal{\cc}{i}{P}{\exists x.G}}{\exists}{
\hresgoal{\cc}{i}{P}{G[t/x]}}$

where $t$ is a term of level $i$ or less.

\item[augment]
$\vladres{\hresgoal{\cc}{i}{P}{D \horseshoe G}}{\horseshoe}{
\hresgoal{\cc}{i}{P \cup D}{G}}$

\item[generic]
$\vladres{\hresgoal{\cc}{i}{P}{\forall x.G}}{\forall}{
\hresgoal{\cc}{i+1}{P}{G[c_{i+1}/x]}}$

where $c_{i+1}$ is a fresh constant of label $i+1$ in ${\mathcal
  U}_{i+1} \setminus {\mathcal U}_i$.
\end{description}

It should be remarked that the main utility of the
 {\bf true} rule is to dispatch
goals of the form $\hresgoal{\cc}{i}{P}{X_o}$
or $\hresgoal{\cc}{i}{P}{X t_1\cdots t_m}$
where the logic variable $X$ is of type
$\alpha_1\imp\cdots\imp\alpha_m\imp o$ and
does not
occur in $P$ or any other state.
Since free variables in goals arise
because at some earlier stage they were existentially quantified after
which the {\bf instance} rule was applied,
this rule, implicitly proposed by Nadathur in a proof procedure
described in
\cite{nadathurPP}, formalizes as a logic programming step the deduction
of the sequent $P\vdash \exists x.x t_1\cdots t_m$ by instantiating $x$ at
$\lambda u_1\cdots\lambda u_m\top$.

\begin{defn}
  A resolution, or derivation, or deduction {\bf d}, is a sequence of resolution
  steps
\[
\fa_1\bleadston{\theta_1}\fa_2\cdots\fa_{n-1}
\bleadston{\theta_{n-1}}\fa_n
\]
Such a sequence is sometimes denoted
\[
\resdots{\fa_1}{\theta}{\fa_n}
\]
with the convention that $\theta$ is the composition
$\theta_1 \cdots \theta_n$ of all the
  intervening substitutions {\em with domain restricted to the
  variables free in \/} $\fa_1$.
   The substitution $\theta$ is called the \emph{computed answer (substitution)} of the derivation.
  A resolution ending in $\NULL$ is called a {\bf successful resolution}.

\end{defn}

Note that a single state vector $\mfa$ may itself be viewed as a derivation with no steps. 
Here are some notational conventions that will be useful below.
\begin{defn}
If $\bd$ is a resolution 
its first and last
state will  be denoted by $\firstd(d)$ and $\lastd(d)$, its length by
$\ell(\bd)$. We write $\fv{\bd}$ to mean the set of all free variables
occurring anywhere in the sequence. 
\end{defn}

In general, many of our results will be stated for derivations with identity computed answers. They may be generalized to the case of arbitrary derivations, but this is not always immediate. If one wants to pursue this direction, it might be convenient to use a domain of existentially quantified substitutions, such $\textit{ISubst}_\sim$ in \cite{AmatoS09-tplp} or \textit{ESubst} in \cite{JacobsL92-jlp}.

For technical reasons, we will often restrict our interest to a particular case of derivations for which not only the computed answer is the identity, but also each intervening substitution has a very particular form.

\begin{defn}[flat derivations]
  A derivation $\fa_1\bleadston{\theta_1}\fa_2\cdots\fa_{n-1}
    \bleadston{\theta_{n-1}}\fa_n$ is called flat when:
  \begin{itemize}
    \item there are no occurrences of the \textbf{true} rule; 
    \item if the $i$-th step is an occurrence of the \textbf{backchain} or \textbf{axiom} rule, then $\theta_i$ is safe and $\dom(\theta_i)$ is restricted to the fresh variables used to rename apart bound variables (the range of $\delta$, see the definition of these two rules ).
  \end{itemize}
  We write \[ \mfa_1 \derivation{[id]} \mfa_2 \] 
  to denote a flat derivation.
\end{defn}


\subsection{Renaming Apart Derivations}

We now focus on the relation between renamings and derivations. We start with a technical lemma which explains how renaming of a  substitution, intended as a syntactic object, may be obtained trough composition.

\begin{lem}
  \label{lem:sub-renaming}
  Given a substitution $\theta = \{t_1/x_1, \ldots, t_n/x_n \}$ and a renaming $\rho$, we have that $\rho^{-1} \theta \rho = \{ t_1\rho / x_1\rho, \ldots, t_n\rho / x_n\rho \}$. Moreover, if $\theta$ is safe than $\rho^{-1} \theta \rho$ is safe, too.
\end{lem}
\begin{proof}
  Let use denote $\rho^{-1} \theta \rho$ by $\eta$. First of all, it is obvious that $x\rho\eta = x \theta \rho$ for each $x \in \vars$. Therefore, for each $i \in \{1,\ldots, n\}$, we have $x_i \rho \theta = t_i \rho \neq x_i \rho$. On the contrary if $x$ is not one of the $x_i$'s, then $x \notin \dom(\theta)$, hence $x\rho\eta = x\rho$. Therefore, $\dom(\eta) = \{ x_1\rho, \ldots, x_n\rho\}$. Finally, the preservation of safety is easy to check.
\end{proof}

\begin{defn}
  A renaming is called \emph{neutral} if it maps variables to variables of the same level.
\end{defn}

The following lemma states that we can apply a neutral renaming to a derivation and obtain another well-formed derivation.
\begin{lem}
  \label{lem:derivation-renaming}
  Given a derivation $\bd: \mfa \derivation{\theta} \mfb$ and a neutral renaming $\rho$, there is a derivation $\bd': \mfa \rho \derivation{\rho^{-1}\theta\rho} \mfb \rho$ obtained from $\bd$ by renaming all state vectors and substitutions with $\rho$. Moreover, if $\bd$ is flat, then $\bd'$ is flat, too.
\end{lem}
\begin{proof}
  Consider $\bd'$ obtained from $\bd$ by renaming all state vectors and substitutions with $\rho$, whereas the renaming of a substitution $\gamma$ is $\rho^{-1}\gamma\rho$. For the moment, assume that $\bd'$ is a derivation, and we prove that its computed substitution is $\rho^{-1}\theta\rho$.

  Assume that $\gamma_1, \ldots, \gamma_n$ are the substitutions
  occurring in $\bd$. Its computed substitution is $\theta = (\gamma_1
  \cdots \gamma_n)_{|X}$ where $X = \fv{\mfa}$. The
  corresponding substitutions occurring in $\bd'$ are
  $\rho^{-1}\gamma_1 \rho, \ldots, \rho^{-1}\gamma_n\rho$, and the
  computed substitution is $\eta = (\rho^{-1}\gamma_1 \cdots
  \gamma_n\rho)_{|Y}$ where $Y=\fv{\mfa\rho} = \rho(X)$. It is
  easy to check that $\eta = \rho^{-1}\theta\rho$. If $y \in \rho(X)$
  the $y \rho^{-1} \in X$, hence $y \eta = y \rho^{-1} \gamma_1 \cdots
  \gamma_n \rho = y \rho^{-1} (\gamma_1 \cdots \gamma_n )_{|X}\rho = y
  \rho^{-1} \theta \rho$. If $y \notin \rho(X)$, then $y \rho^{-1}
  \notin X$, hence $y \eta = y$ and $y \rho^{-1} \theta \rho = y
  \rho^{-1}\rho = y$. 

  
  Note that, if $\bd$ is safe, than $\bd'$ does not contain any occurrence of the \textbf{true} rule, and all the substitutions in the \textbf{backchain} and \textbf{axiom} steps are safe. The condition on the domain of these substitutions will be checked later, completing the proof that $\bd'$ is safe, too.

  We now prove by induction on the lenght of $\bd$ that $\bd'$ is actually a derivation. The case when $\bd$ has length zero is trivial. For the inductive step, we consider different subcases according to the first rule used in the derivation.

  \case{true}
  The deduction $\bd$ has the following shape: 
  \[
    \rlap{$\underbrace{\mfa\otimes\ipresgoal{i}{P}{G}\otimes\mfb
    \derstep{\gamma} \mfa\gamma\otimes\ipresgoal{i}{P\gamma}{\top}\otimes\mfb\gamma}_{\bd_1}$}
    \mfa\otimes\ipresgoal{i}{P}{G}\otimes\mfb
    \derstep{\gamma} 
    \overbrace{\mfa\gamma\otimes\ipresgoal{i}{P\gamma}{\top}\otimes\mfb\gamma
    \derivation{\theta} \mfc}^{\bd_2}
  \]
  while $\bd'$ is
  \[
      \rlap{$\underbrace{\mfa\rho\otimes\ipresgoal{i}{P\rho}{G\rho}\otimes\mfb\rho
      \derstep{\rho^{-1}\gamma\rho}
      \mfa\gamma\rho\otimes\ipresgoal{i}{P\gamma\rho}{\top}\otimes\mfb\gamma\rho}_{\bd'_1}$}
      \mfa\rho\otimes\ipresgoal{i}{P\rho}{G\rho}\otimes\mfb\rho
      \derstep{\rho^{-1}\gamma\rho}
      \overbrace{\mfa\gamma\rho\otimes\ipresgoal{i}{P\gamma\rho}{\top}\otimes\mfb\gamma\rho  
       \derivation{\rho^{-1}\theta\rho} \mfc\rho}^{\bd'_2}
    \]
  By inductive hypothesis $\bd'_2$ is a derivation. Note that, since $\gamma$ is not the identity, the $\rho^{-1}\gamma\rho$ is not the identity subtitution. Therefore $\bd'$ is a derivation, too.

  \case{backchain}
  The proof is similar to the previous case, with 
  \[
    \bd = \mfa\otimes\ipresgoal{i}{P}{A'}\otimes\mfb
    \derstep{\gamma} 
    \mfa\gamma\otimes\ipresgoal{i}{P\gamma}{G\delta\gamma}\otimes\mfb\gamma
    \derivation{\theta} \mfc
  \]
  and 
  \[
      \bd' = \mfa\rho\otimes\ipresgoal{i}{P\rho}{A'\rho}\otimes\mfb\rho
      \derstep{\rho^{-1}\gamma\rho}
      \mfa\gamma\rho\otimes\ipresgoal{i}{P\gamma\rho}{G\delta\gamma\rho}\otimes\mfb\gamma\rho  \derivation{\rho^{-1}\theta\rho} \mfc\rho
  \]
  We need to check that the conditions for the correct application of the \textbf{backchain} rule are satisfied. Note that, since $\bd$ is a derivation, there is a clause $\forall x_1 \cdots \forall x_n (G \horseshoe A) \in \elab(P)$ such that $A \delta \gamma = A' \gamma$, where  $\delta = \{w_1/x_1, \ldots, w_n/x_n\}$ is the renaming-apart substitution. By Lemma~\ref{lem:clause-renaming}, we also have a clause $\forall x_1 \cdots \forall x_n (G\rho \horseshoe A\rho) \in \elab(P\rho)$. Consider $\delta'=\rho^{-1}\delta\rho = \{\rho(w_1) / \rho(x_1), \ldots, \rho(w_n)/\rho(x_n)\}$. Note that, since $w_i$ is fresh in the original derivation, $\rho(w_i)$ is fresh in the new one, hence $\delta'$ is a renaming-apart substitution. Since $\rho^{-1} \gamma \rho$ is an unifier for $A \rho \delta'$ and $A'\rho$, and since $G \rho \delta' \rho^{-1} \gamma \rho = G \delta \gamma \rho$, the conditions for the correct application of the \textbf{backchain} rule are satisfied.

  Moreover, if $\bd$ is flat then $\dom(\gamma) \subseteq \{w_1, \ldots, w_n\}$. It turns out  that $\rho^{-1}\gamma\rho$ is safe and $\dom(\rho^{-1}\gamma\rho) \subseteq \{\rho(w_1), \ldots, \rho(w_n)\} = \dom(\delta')$, which was the missing condition fot ensuring that $\bd'$ is flat. 

  \case{axiom}
  This is a simpler case of the backchain rule.

  \case{instance}
  In this case
  \[
    \bd = \mfa\otimes\ipresgoal{i}{P}{\exists x. G}\otimes\mfb \derstep{\exists} \mfa\otimes\ipresgoal{i}{P}{G[t/x]}\otimes\mfb \derivation{\theta'} \mfc
  \]
  and
  \[
    \bd' = \mfa\rho\otimes\ipresgoal{i}{P\rho}{\exists x. G\rho}\otimes\mfb \derstep{\exists} \mfa\rho\otimes\ipresgoal{i}{P\rho}{G[t/x]\rho}\otimes\mfb\rho \derivation{\rho^{-1}\theta'\rho} \mfc\rho
  \]
  where all but the first step in $\bd'$ are valid by inductive hypothesis. Note that $G[t/x]\rho = G\rho[t\rho/x]$, and since $\rho$ is neutral, $t$ and $t\rho$ have the same level. Hence the first step in $\bd'$ is valid, too.

  \case{null, or, and, augment, generic}
  These cases follow immediately by the inductive hypothesis.
\end{proof}

In general, in all the proofs which follow, we will omit discussing the \textbf{axiom} rule, since it is a simpler case of the \textbf{backchain} rule.

Using the previous result, the following lemma states that substitutions in a derivation
can be chosen not only to be idempotent (as we require) but in such a way as
to avoid any finite set of variables not containing any of the variables of
the first and last state vector.
\begin{lem}
  \label{lem:choicevar}
  Let \bd\ be a derivation $\mfa \derivation{\theta} \mfb$. If $X$ is a finite set of
  variables such that $X\cap (\fv{\mfa} \cup \fv\mfb \cap \fv\theta) = \nullset$, then
  there is a derivation of $\bd': \mfa \derivation{\theta} \mfb$, none of whose
  free variables lie in $X$. Moreover, if $\bd$ is flat, then $\bd'$ is flat, too.
\end{lem}
\begin{proof}
  Let $\rho$ be a neutral renaming which maps variables in $X$ to variables not in $X$ or $\fv{\bd}$. The derivation we are looking for is the one obtained by Lemma~\ref{lem:derivation-renaming}.
\end{proof}

Due to the freshness condition in the \textbf{generic} rule, we sometimes need to rename apart not only variable, but also constants.
 
\begin{defn}[constant replacer]
  A \emph{constant replacer} is a function mapping constants to constants of the same type, which is the identity everywhere but for a finite set of constants. If $X$ is any syntactic object and $\xi$ a constant renamer, $X\xi$ will be the result of replacing every constant $c$ with $\xi(c)$. Similarly to substitution, we define:
  \begin{itemize}
    \item the \emph{domain} of a costant replacer $\xi$ to be the set $\dom(\xi) = \{ c \mid \xi(c) \neq c \}$;
    \item the \emph{range} of a costant replacer $\xi$ to be the set $\rng(\xi) = \xi(\dom(\theta))$;
    \item the constant occuring in $\xi$ to be the set $\oc(\xi) = \dom(\xi) \cup \rng(\xi)$.
  \end{itemize}
  A constant replacer which maps every constants to another one of the same level is called \emph{neutral}, while a bijective constant replacer is called a \emph{constant renamer}. 
\end{defn}

\begin{lem}[clause constant replacing]
  \label{lem:cren-ep}
    Consider a constant replacer $\xi$. If there is a clause $K \in \elab(D)$ then there is a clause $K\xi \in \elab(D\xi)$.
\end{lem}
\begin{proof}
  The proof is immediate by induction on the complexity of $D$.
\end{proof}

\begin{lem}[constant renaming]
  \label{lem:derivation-crenaming}
  Given a derivation $\bd: \mfa \derivation{\theta} \mfb$ and a neutral constant renamer $\xi$, there is a derivation $\bd': \mfa \xi \derivation{\theta\xi} \mfb \xi$ obtained from $\bd$ by applying $\xi$ to all state vectors and substitutions. Moreover, if $\bd$ is flat, then $\bd'$ is flat, too.
\end{lem}

\begin{proof}
  Consider $\bd'$ obtained from $\bd$ by applying $\xi$ to all state vectors and substitutions, where the application of $\xi$ to a substitution $\{t_1/x_1, \ldots, t_n/x_n\}$ is the substitution $\{t_1\xi/x_1, \ldots, t_n\xi/x_n\}$. The proof is similar to the one in Lemma~\ref{lem:derivation-renaming}. 
  
  We first assume that $\bd'$ is a derivation, and prove that its computed substitution is $\theta\xi$. Assume that $\gamma_1, \ldots, \gamma_n$ are the substitutions occurring in $\bd$. Its computed substitution is $\theta = (\gamma_1 \cdots \gamma_n)_{|X}$ where $X = \fv{\mfa}$. The corresponding substitutions occurring in $\bd'$ are $\gamma_1 \xi, \ldots, \gamma_n\xi$, and the computed substitution is $\eta = ((\gamma_1\xi) \cdots (\gamma_n\xi))_{|X}$. It is easy to check that $(\gamma_1 \xi) \cdots (\gamma_n \xi) = (\gamma_1 \cdots \gamma_n) \xi$, whence $\theta = \eta$.  Regarding flatness, since $\xi$ does not act on variables, if $\bd$ is flat, then $\bd'$ is flat, too.

  We now prove that $\bd'$ is a derivation.
  First of all, since $\xi$ is injective, freshness conditions for constants introduced by the \textbf{generic} rule in $\bd'$ are respected. The rest of the proof is by induction on the length of $\bd$. The case when $\bd$ has length zero is trivial. For the inductive step, we consider different cases according to the first rule used in the derivation.

  \case{backchain}
  The derivation $\bd$ has the form
  \[
    \mfa_1 \otimes \ipresgoal{i}{P}{A} \otimes \mfa_2 \derstep{\gamma} 
    \mfa_1 \otimes \ipresgoal{i}{P}{G\delta\gamma} \otimes \mfa_2 \derivation{\theta'} \NULL
  \]
  for some clause $\forall x_1 \cdots \forall x_n(G \horseshoe A') \in \elab(P)$ and $\delta=\{w_1/x_1, \ldots, w_n/x_n\}$ a renaming-apart substitution such that $A' \delta\gamma = A \gamma$. Meanwhile, $\bd'$ is
  \[
    \mfa_1\xi \otimes \ipresgoal{i}{P\xi}{A\xi} \otimes \mfa_2\xi \derstep{\gamma\xi} 
    \mfa_1\xi \otimes \ipresgoal{i}{P\xi}{G\delta \gamma \xi} \otimes \mfa_2\xi \derivation{\theta' \xi} \mfb\xi
  \]
   By Lemma~\ref{lem:cren-ep}, there is a clause $\forall x_1 \cdots \forall x_n(G \xi \horseshoe A'\xi) \in \elab(P\xi)$. Moreover, it is easy to check that if $A' \delta \gamma = A \gamma$ then $A' \delta \gamma \xi = A \gamma \xi$, while $A \gamma\xi = A \xi (\gamma \xi)$ and  $A' \delta \gamma\xi = A' \xi \delta (\gamma \xi)$, i.e., $\gamma\xi$ is an unifier of $A\xi$ and $A'\xi\delta$. Morevoer $G\delta \xi \gamma \xi = G\delta (\gamma \xi)$. Therefore, the first step of derivation $\bd'$ is valid, while validity of the rest of the derivation follows by inductive hypothesis.
  
  \case{true}
  This case is similar to the \textbf{backchain} one.

  \case{null, and, or, augment}
  Immediate by induction hypothesis.

  \case{generic}
  The derivation $\bd$ has the form
  \[
    \mfa_1 \otimes \ipresgoal{i}{P}{\forall x. G} \otimes \mfa_2 \derstep{\forall}\mfa_1 \otimes \ipresgoal{i+1}{P}{G[c/x]} \otimes \mfa_2 \derivation{\theta} \mfb
  \]
  where $c$ is a fresh constant of level $i+1$, while $\bd'$ is 
  \[
    \mfa_1 \xi \otimes \ipresgoal{i}{P\xi}{(\forall x. G)\xi} \otimes \mfa_2\xi
    \derstep{\forall}
    \mfa_1\xi \otimes \ipresgoal{i+1}{P\xi}{G[c/x]\xi} \otimes \mfa_2 \xi
    \derivation{\theta\xi} \mfb\xi
   \]
   Since $G[c/x]\xi = G\xi[\xi(c)/x]$, $\xi(c)$ is of level $i+1$ and by inductive hypotehsis, $\bd'$ is a valid derivation.

  \case{instance}
  This case is similar to the \textbf{generic} one. \qedhere
\end{proof}

Using the previous result, the following lemma states that constants in a derivation can be chosen to avoid any finite set of constants not containing any of the constants of the first and last state vector.
\begin{lem}
  \label{lem:choiceconst}
  Let $\bd$ be a derivation $\mfa \derivation{\theta} \mfb$. If $X$ is a finite set of
  constant not appearing in $\mfa$, $\mfb$ and $\theta$, then there is a derivation of $\bd': \mfa \derivation{\theta} \mfb$, none of whose constants lie in $X$. Moreover, $\bd'$ is flat if $\bd$ is flat.
\end{lem}
\begin{proof}
  Let $\xi$ be a neutral constant renamer which maps variables in $X$ to variables not in $X$ nor in $\bd$. The derivation we are looking for is the one obtained by Lemma~\ref{lem:derivation-crenaming}.
\end{proof}

A particular case of Lemmas~\ref{lem:choicevar} and~\ref{lem:choiceconst} wll be particularly usefult later and regards successful derivations with identity computed answer. In this case:
\begin{lem}
  \label{lem:derivation-fresh}
  Let $\bd$ be a derivation $\mfa \derivation{\ids} \NULL$, $X$ is a set of variables and $C$ a set of constants. Then, there exists a derivation $\bd':  \mfa \derivation{\ids} \NULL$ whose newly introduced variables and constants do not belong to $X$ or $C$.
\end{lem}

In the following, we will use Lemma~\ref{lem:derivation-fresh}, even implicitly, to ensure that variables introduced in derivations are fresh.

Finally, we examine a way to combine two derivations into a new one:

\begin{lem}[product lemma]
\label{lem:product}
If there are resolutions $\bd_a: \midnullsres{\mfa}{\ids}$ and
$\bd_b: \midnullsres{\mfb}{\ids}$, then there is a resolution $\bd: \midnullsres{\mfa\otimes\mfb}{[\ids]}$.
Moreover, $\bd$ is flat if $\bd_a$ and $\bd_b$ are both flat.
\end{lem}

\begin{proof}
  Using Lemma~\ref{lem:choicevar} consider a variant $\bd'_a$ of $\bd_a$ which does not use any variable in $\fv{\mfb} \setminus \fv{\mfa}$. Now, given $X$ the set of variables occurring in $\bd'_a$, use Lemma~\ref{lem:derivation-fresh} to obtain a derivation $\bd'_b$ which 
  \begin{itemize}
    \item  does not use any variable in $X \setminus \fv{\mfb}$;
    \item  does not introduce with the \emph{generic} rule any constant occuring in $\bd'_{a}$
  \end{itemize}
  Then, we first take the derivation
  \[
    \bd'_a \otimes \mfb: \mfa \otimes \mfb \derivation{\ids} \mfb
  \]
  obtained by concatenating each state vector in $\bd'_a$ with $\mfb$. Then, $\bd$ is the concatenation of $\bd'_a \otimes \mfb$ with $\bd'_b$.
\end{proof}

The preceding operation might be generalized to the case where the computed substitutions are not identities. Moreover, other operations might be defined to turn derivations into a rich algebraic structure. We do not pursue this here, but refer the reader to \cite{CoLeviMeo} for details.

\subsection{Some specialization properties of \UCTT}

\begin{lem}[Instantiation Lemma]
  \label{lem:weak-lifting}
  Let $\mfa$ be a state vector. Suppose there is a deduction 
\[
  \nullsres{\mfa}{\theta}  
\]
with computed substitution $\theta$. Then there is a flat deduction of the same or shorter length
\[
\nullsres{\mfa\theta}{[id]}.
\]
\end{lem}

\begin{proof}
  We prove this by induction on the length $\ell$ of the given
  proof. The cases $\ell = 0$
  is trivial. Now assume the given deduction has length $\ell>0$ and that the first step is

\case{null}
Then the deduction has the following shape:
\[
\midnullsrest{\mfa\otimes\ipresgoal{i}{P}{\top}\otimes\mfb}
{null}
{\mfa\otimes\mfb}
{\theta} 
\]
By the inductive hypothesis  applied to the deduction starting from the
second state vector, there is a derivation 
\[
\nullsres
{\mfa\theta\otimes\mfb\theta}
{[id]}
\]
By Lemmas~\ref{lem:choicevar} and~\ref{lem:choiceconst} we may assume that fresh variables and constants introduced in this derivation do not occur in $P\theta$. Therefore,
\[
\midnullsrest{\mfa\theta\otimes\ipresgoal{i}{P\theta}{\top}\otimes\mfb\theta}
{null}
{\mfa\theta\otimes\mfb\theta}
{[id]} 
\]
is a flat derivation.

\case{true}
Then the deduction has the following shape:
\[
\midnullsrest{\mfa\otimes\ipresgoal{i}{P}{G}\otimes\mfb}
{\gamma}
{\mfa\gamma\otimes\ipresgoal{i}{P\gamma}{\top}\otimes\mfb\gamma}
{\theta'} 
\]
By induction hypothesis, applied to the deduction starting from the
second state vector, there is a deduction
\[
\nullsres
{\mfa\gamma\theta'\otimes\ipresgoal{i}{P\gamma\theta'}{\top}\otimes\mfb\gamma\theta'}
{[id]}
\]
Since $\gamma \theta' = \theta$ when
restricted to the free variables in the initial state vector, this is actually a derivation
\[
\nullsres
{\mfa\theta\otimes\ipresgoal{i}{P\theta}{\top}\otimes\mfb\theta}
{[id]}
\]
which is what we wanted to prove.

\case{backchain}
We assume the proof has the following shape:
\[
\midnullsrest{\mfa\otimes\ipresgoal{i}{P}{A}\otimes\mfb}
{\gamma}
{\mfa\gamma\otimes\ipresgoal{i}{P\gamma}{G\delta\gamma}\otimes\mfb\gamma}
{\theta'}
\]
where $\forall x_1, \ldots x_n (G \hoe A') \in \elab(P)$, $\delta =\{w_1/x_1, \ldots, w_n/x_n\}$ is the renaming-apart subtitution and $A'\delta\gamma=A\gamma$ and $\gamma\theta'=\theta$ when restricted to variables free in the initial state vector.
Applying the induction hypothesis to the derivation of length $\ell-1$
starting with the second state vector, we have
\[
  \nullsres{\mfa\theta\otimes\ipresgoal{i}{P\theta}{G\delta\gamma\theta'} \otimes \mfb \theta
  }
{[id]}
\]
By Lemmas~\ref{lem:choicevar} and~\ref{lem:choiceconst} we can assume that fresh variables and constants introduced in this derivation are disjoint from those in $A \theta$. By the clause instance lemma (Lemma~\ref{lem:clause-renaming}), 
we have $\forall \vec x (G\gamma\theta'\hoe A'\gamma\theta')\in\elab(P\theta)$. Consider a new set of fresh variables $w'_1, \ldots, w'_n$ and the renaming apart substitution $\delta' = \{ w'_1/x_1, \ldots, w'_n/x_n \}$ for this clauses. Consider also $\gamma' = \{ w_1\gamma\theta' / w'_1, \ldots, w_n\gamma\theta' / w'_n \}$.

We want to prove that $\gamma'$ is an unifier of $A\theta$ and $A'\gamma\theta'\delta'$. Since $A\theta =  A' \delta \gamma \theta'$, we need to prove that $A' \delta \gamma \theta' \gamma' = 
A'\gamma\theta'\delta'\gamma'$. It is enough to prove that $x \delta \gamma \theta' \gamma' =  x \gamma\theta'\delta'\gamma'$ for each $x \in \fv{A'}$. Note that $x \delta \gamma \theta'$ does not contain any of the new fresh $w'_i$'s, hence $x \delta \gamma \theta' \gamma' = x \delta \gamma \theta'$. We distinguish two cases:
\begin{itemize}
  \item if $x = x_i$ for one of the bound variables of the clause, then $x_i \gamma\theta'\delta'\gamma' = x_i \delta' \gamma'$ since $\gamma\theta'$ does not contain variables in $\bvar$, and $x_i \delta' \gamma' = w'_i \gamma' = w_i \gamma \theta' = x_i \delta \gamma \theta'$.
  \item if $x$ is not one of the bound variable of the clause, then  $x \gamma\theta'\delta'\gamma' = x \gamma \theta'= x \delta \gamma \theta'$.
\end{itemize}
For the same reason $G \gamma \theta' \delta' \gamma' = G \delta \gamma \theta' \gamma'$. Therefore, the following is a legal derivation of length
at most $\ell$
\[
\midnullsrest{\mfa\theta\otimes\ipresgoal{i}{P\theta}{A\theta}
\otimes\mfb\theta}
{\gamma'}
{\mfa\theta\otimes\ipresgoal{i}{P\theta}{G\delta\gamma\theta'}
\otimes \mfb \theta}
{[id]}
\]
as we wanted to show. Note that $\gamma'$ respects the conditions for flat derivations.

\case{and}
Suppose the given deduction is 
\[
\midnullsrest{\mfa\otimes\ipresgoal{i}{P}{G_1\wedge G_2}\otimes\mfb}
{\wedge}
{\mfa\otimes\ipresgoal{i}{P}{G_1}
\otimes\ipresgoal{i}{P}{G_2}\otimes\mfb}
{\theta}.
\]
Applying the induction hypothesis to the resolution beginning with the
second step, we obtain a deduction
\[
\midnullsrest{\mfa\theta\otimes\ipresgoal{i}{P\theta}{G_1\theta\wedge
G_2\theta}\otimes\mfb\theta}
{\wedge}
{\mfa\theta\otimes\ipresgoal{i}{P\theta}{G_1\theta}
\otimes\ipresgoal{i}{P\theta}{G_2\theta}\otimes\mfb\theta}
{[id]}.
\]

\case{or}
The proof is similar to the one for the \textit{null} case.

\case{instance}
Suppose the given deduction is 
\[
\midnullsrest{\mfa\otimes\ipresgoal{i}{P}{\exists x G}
\otimes\mfb}
{\exists}
{\mfa\otimes\ipresgoal{i}{P}{G[t/x]}
\otimes\mfb}
{\theta'}
\]
where $\theta'=\theta$ when 
restricted to the free variables  in the original state vector.
%
Using this fact, together with the induction hypothesis on the
deduction beginning with the second state vector, we have
\[
\nullsres{\mfa\theta\otimes\ipresgoal{i}{P\theta}
{G[t/x]\theta'}\otimes\mfb\theta}
{[id]}
\]
Since $x \in \bvar$, hence it does not occur free in $\theta'$, $G[t/x]\theta' = G\theta[t\theta'/x]$. Note that, since all substitutions are valid, $t\theta'$ is a term of the same level of $x$ or less. Using $t\theta'$ as the witness for the instance step, we
obtain the existence of a deduction of length at most $\ell$
\[
\midnullsrest{\mfa\theta\otimes\ipresgoal{i}{P\theta}{(\exists x G)\theta}
\otimes\mfb\theta}
{\exists}
{\mfa\theta\otimes\ipresgoal{i}{P\theta}{G\theta[t\theta'/x]}
\otimes\mfb\theta}
{[id]}
\]
as we wanted to show.

\case{generic}
Then the deduction has the following shape:
\[
\midnullsrest{\mfa\otimes\ipresgoal{i}{P}{\forall x. G}\otimes\mfb}
{\forall}
{\mfa\otimes\ipresgoal{i+1}{P}{G[c/x]}\otimes\mfb}
{\theta}.
\]
By inductive hypothesis, we have a derivation
\[
\nullsres{\mfa\theta\otimes\ipresgoal{i+1}{P\theta}
{G[c/x]\theta}\otimes\mfb\theta}
{[id]}
\]
Since $x \in \bvar$, we have $G[c/x]\theta = G\theta[c/x]$. Therefore, we get the derivation
\[
\midnullsrest{\mfa\theta\otimes\ipresgoal{i}{P\theta}{\forall x. G\theta}\otimes\mfb\theta}
{\forall}
{\mfa\theta\otimes\ipresgoal{i+1}{P\theta}{G\theta[c/x]}\otimes\mfb\theta}
{[id]}. \qedhere
\]
\end{proof}

A natural question to ask here is whether we can also establish
the converse of this lemma,
usually known in the literature, in perhaps a slightly different form,
as the {\em lifting lemma\/}. This question is taken up in
Section~\ref{sec:enriched-resolution}.

We now give a number of technical results about
resolution proofs that
will be needed to prove soundness and completeness with respect to our
semantics.

\begin{lem}[Left weakening]
  \label{lem:left-weakening}
  Let $\bd: \mfa_1 \otimes \ipresgoal{i}{P}{G} \otimes \mfa_2 \derivation{[\ids]} \NULL$ be a derivation. For each program $P'$ of level $i$ or less such that $\elab(P') \supseteq \elab(P)$, there is a derivation $\bd': \mfa_1 \otimes \ipresgoal{i}{P'}{G} \otimes \mfa_2 \derivation{[\ids]} \NULL$.
\end{lem}
\begin{proof}
  First of all, assume without loss of generality that $\bd$ does not use any variable or constant occuring in $\fv{P'} \setminus \fv{\firstd(\bd)}$. The proof is by induction on the length of $\bd$, and it is quite straightforward once we observe that if $\elab(P')\supseteq \elab(P)$, then also $\elab(P'\gamma) \supseteq \elab(P\gamma)$ for each substitution $\gamma$ and $\elab(P', D) \supseteq  \elab(P, D)$ for each program formula $D$. We show just a few cases.
  
  If $\bd$ has length zero, the result is trivial. For the inductive step, we  first distinguish two cases.
  \paragraph*{First case}
  If the selected goal is not $\ipresgoal{i}{P}{G}$, then assume without loss of generality that it is in $\mfa_1$. Then 
  \[
    \bd: \mfa_1 \otimes \ipresgoal{i}{P}{G} \otimes \mfa_2 \derstep{\gamma} \mfa'_1\gamma \otimes \ipresgoal{i}{P\gamma}{G\gamma} \otimes \mfa_2\gamma \derivation{[\ids]} \NULL \enspace .
  \]
  Note that we need flatness of $\bd$ to conclude that the computed answer of the substitution starting from the second state vector is $[\ids]$. By inductive hypothesis, there exists a fresh derivation
  \[
      \bd'': \mfa'_1 \otimes \ipresgoal{i}{P'\gamma}{G\gamma} \otimes \mfa_2\gamma \derivation{[\ids]} \NULL \enspace .
  \] 
  Due to the choice of $\bd$, the same rule used in the first step of $\bd$ may be also applied to $\mfa_1 \otimes \ipresgoal{i}{P'}{G} \otimes \mfa_2$. Therefore, the derivation we are looking for is
  \[
    \bd': \mfa_1 \otimes \ipresgoal{i}{P'}{G} \otimes \mfa_2 \derstep{\gamma} \mfa'_1 \otimes \ipresgoal{i}{P'\gamma}{G\gamma} \otimes \mfa_2\gamma \derivation{[\ids]} \NULL \enspace .
  \]

  \paragraph*{Second case}
  If the selected goal is $\ipresgoal{i}{P}{G}$, we proceed according to the rule used. 

  \case{and}
  The derivation $\bd$ has the form
  \[
    \bd: \mfa_1 \otimes \ipresgoal{i}{P}{G_1 \wedge G_2} \otimes \mfa_2 \derstep{} \mfa_1 \otimes \ipresgoal{i}{P}{G_1} \otimes \ipresgoal{i}{P}{G_2} \otimes \mfa_2 \derivation{[\ids]} \NULL \enspace .
  \]
  Applying two times the inductive hypothesis on the derivation starting from the second step of $\bd$, we obtain a fresh derivation
  \[
    \bd'': \mfa_1 \otimes \ipresgoal{i}{P'}{G_1} \otimes \ipresgoal{i}{P'}{G_2} \otimes \mfa_2 \derivation{[\ids]} \NULL \enspace ,
  \]
  whence
  \[
    \bd': \mfa_1 \otimes \ipresgoal{i}{P'}{G_1 \wedge G_2} \mfa_2  \derstep{\wedge} \mfa_1 \otimes \ipresgoal{i}{P'}{G_1} \otimes \ipresgoal{i}{P'}{G_2} \derivation{[\ids]} \NULL \enspace .
  \]

  \case{backchain}
  The derivation $\bd$ has the form
  \[
    \bd: \mfa_1 \otimes \ipresgoal{i}{P}{A_r} \otimes \mfa_2 \derstep{\gamma} \mfa_1\gamma \otimes \ipresgoal{i}{P\gamma}{H} \otimes \mfa_2 \gamma  \derivation{[\ids]} \NULL \enspace ,
  \]
  By inductive hypothesis, we have a fresh derivation
  \[
    \bd'': \mfa_1\gamma \otimes \ipresgoal{i}{P' \gamma}{H} \otimes \mfa_2 \gamma  \derivation{[\ids]} \NULL \enspace .
  \]
  Since $\elab(P') \supseteq \elab(P)$ and since $\bd$ does not use variables in $\fv{P'} \setminus \fv{\firstd(\bd)}$, the same backchain step used in $\bd$ may be also applied to $\ipresgoal{i}{P'}{A_r}$. 
  Therefore we obtain:
  \[
    \bd': \mfa_1 \otimes \ipresgoal{i}{P'}{G} \otimes \mfa_2 \derstep{\gamma} \mfa_1\gamma \otimes \ipresgoal{i}{P' \gamma}{H} \otimes \mfa_2 \gamma  \derivation{[\ids]} \NULL \enspace .
  \]

  \case{other rules}
  The proof is similar to what is shown above.
  %
\end{proof}

\begin{lem}[Specialization]
\label{lem:sub-t}
  Let $\mfa$ be a state vector. If there is a deduction
\begin{equation}
\label{ded:givent}
\bd: \nullsres{\mfa}{[id]} \enspace ,
\end{equation}
a variable $x \in \fvar$ of level $i$
and a term $t$ of level $i$ or less, then there is a deduction
\[
\nullsres{\mfa[t/x]}{[id]}
\]
\end{lem}
\begin{proof}
First of all, assume $x \in \fv{\mfa}$, otherwise there is nothing to prove. Moreover, without loss of generality, we may assume $\bd$ does not use any variable in $\fv{t} \setminus \fv{\mfa}$. The proof is by induction on the length $\ell$ of the given deduction.
Suppose the deduction has length 0. Then there is nothing to prove.
Suppose the lemma holds for all shorter deductions, and now consider
all possible rules used in the first step.

\case{true}
The case for {\em true\/} resolution step is vacuous.
Since all substitutions are identities, $\bd$ cannot contain  a
{\em true\/} step.

\case{backchain}
We assume that the proof looks like this
\[
\midnullsrest{\stipresgoal{i}{P}{A}}{\gamma} {\stipresgoal{i}{P}{G\delta\gamma}}{[id]}
\]
for some clause $\forall \vec y (G \hoe A') \in \elab(P)$, where $\delta =\{w_1/y_1, \ldots, w_n/y_n\}$ is the renaming-apart substitution and $A' \delta \gamma = A$ (since the derivation is flat, hence $\gamma$ is the identity on the variables of $A$). By induction hypothesis, there is a fresh derivation
\[
\nullsres{\stsipresgoal{i}{P[t/x]}{G\delta\gamma[t/x]}{[t/x]}}{[id]}
\] 
By the hypothesis at the beginning of the proof, we have that $\gamma(x)=x$, $G \delta \gamma[t/x] = G [t/x] \delta \gamma$. 
By the clause instance lemma (Lemma~\ref{lem:clause-renaming}) we
have that $\forall \vec y (G[t/x]\hoe A [t/x]) \in \elab(P[t/x])$.
Then the following resolution exists
\[
\midnullsrest{\stsipresgoal{i}{P[t/x]}{A[t/x]}{[t/x]}}{\gamma}
{\stsipresgoal{i}{P[t/x]}{G[t/x]\delta\gamma}{[t/x]}}{[id]}
\]
since $\gamma$ is an unifier of $A [t/x]$ and $A' [t/x]$: $A [t/x] \gamma = A[t/x]$ and $A' [t/x] \delta  \gamma = A' \delta \gamma [t/x] = A[t/x]$.

\case{instance}
Suppose the first step of the given deduction is an \textbf{instance}
step
\[
\midnullsrest{\stipresgoal{i}{P}{\exists y G}}{\exists}
{\stipresgoal{i}{P}{G[s/y]}}{[id]}
\]
By induction hypothesis there is a proof
\[
\nullsres{\stsipresgoal{i}{P[t/x]}{G[s/y][t/x]}{[t/x]}}{[id]}.
\]
Letting $s'=s[t/x]$ and taking $s'$ as the witness
used in the first instance step, we have a deduction
\[
\midnullsrestarray{\stsipresgoal{i}{P[t/x]}{\exists
y(G[t/x])}{[t/x]}}{\exists}
{\stsipresgoal{i}{P[t/x]}{G[t/x][s'/y]}{[t/x]}}{[id]}
\]
as we wanted to show.

\case{generic}
Suppose the first step of the given deduction is an instance of the
\textbf{generic} rule
\[
\midnullsrest{\stipresgoal{i}{P}{\forall y G}}{\forall}
{\stipresgoal{i+1}{P}{G[c/y]}}{[id]}
\]
Since $y \in \bvar$ we have $G[c/y][t/x]=G[t/x][c/y]$ so,
applying the induction hypothesis to the shorter proof starting at the
second state vector above, we obtain a deduction
\[
\midnullsrestarray{\stsipresgoal{i}{P[t/x]}{(\forall y
G)[t/x]}{[t/x]}}
{\forall}
{\stsipresgoal{i+1}{P[t/x]}{G[t/x][c/y]}{[t/x]}}{[id]}
\]

\noindent The arguments for the {\bf null, augment} and {\bf or} cases are
easy and left to the reader.
\end{proof}
Iteration of this lemma easily proves a bit more.
\begin{cor}
\label{cor:sub-theta}
  Let $\mfa$ be a state vector and $\theta$ a substitution. If there is a deduction
  \[ \nullsres{\mfa}{[id]} \]
  then there is a deduction
  \[
  \nullsres{\mfa\theta}{[id]} \enspace. 
  \]
\end{cor}
\begin{proof}
  The proof is immediate if $\theta$ is safe, since it is enough to proceed one binding at a time by using Lemma~\ref{lem:sub-t}. If $\theta$ is not safe, it is possibile to find safe derivations $\theta_1$ and $\theta_2$ such that $\mfa \theta = \mfa \theta_1 \theta_2$. If $\theta = \{ t_1/x_1, \ldots, t_n/x_n \}$, consider the fresh variables $w_1, \ldots w_n$ and the renaming $\rho = \{ w_1/x_1, x_1/w_1,\ldots, w_n/x_n, x_n/w_n \}$. Then let $\theta_1 = \{ t_1 \rho/x_1, \ldots, t_n\rho /x_n \}$ and $\theta_2 = \rho$, we have $\mfa \theta = \mfa \theta_1 \theta_2$. Therefore, we obtain the required derivation by applying Lemma~\ref{lem:sub-t} one binding at a time, first with the bindings of $\theta_1$, then with the bindings of $\theta_2$.
\end{proof}

%
%

\begin{lem}[Constant replacement]
  \label{lem:constant-replacement}
  If there is a deduction
  \[
    \mfa [c/x] \derivation{[id]} \NULL
  \]
  where $x$ is a variable of level $i$ and $c$ a constant of level $i$ or less, then there exists a deduction
  \[
    \mfa \derivation{[id]} \NULL \enspace .
  \]
\end{lem}
\begin{proof}
  Using Lemma~\ref{lem:derivation-fresh}, we may assume that  $\bd: \mfa [c/x] \derivation{[id]} \NULL$ does not use the variable $x$. Therefore, it should be quite evident that the derivation we are looking for is obtained from $\bd$ by replacing every occurrence of $c$ with $x$. The fact that the level of $c$ is less or equal to the level of $x$ ensures that every valid substitution remains valid after the replacement.
\end{proof}
  

For the following lemma, we need to introduce new pieces of notation. Assume $\mfa$ is a state vector whose states are all at level one or more. Then, we denote by $\mfa\down$ the state vector obtained reducing by one the level of all the states. Note that, in the general case, $\mfa\down$ might not be a valid state vector.

\begin{lem}[Level reduction]
  \label{lem:level-reduction}
  Suppose $\mfa$ is a state vector whose states are all at level $i > 0$ or more. Let:
  \begin{itemize}
    \item $\xi$ be a constant replacer mapping every constant of level $j \geq i$ occuring in $\mfa$ to a constant of level $j-1$ or less, and every constant of level $j < i$ occuring in $\mfa$ to a constant of the same level;
    \item $\theta$ be a substitution mapping every variable of level $j \geq i$ occuring in $\mfa$ to a term of level $j-1$ or less, and every variable of level $j < i$ occuring in $\mfa$ to a term of the same level.
  \end{itemize}
  If there is a flat deduction
  \[  \bd: \nullsres{\mfa}{[id]} \enspace ,\]
  then there is a flat deduction
  \[ \nullsres{\mfa\xi\theta\down}{[id]} \]
  of the same length.
\end{lem}

  \begin{proof}
   First of all, note that by Lemma~\ref{lem:derivation-fresh} we may assume that fresh variables and constants introduced in $\bd$ are disjoint from $\oc(\xi)$ and $\fv{\theta}$. Also note that $\mfa\xi\theta\down$ is a valid state vector, since applying $\xi$ and $\theta$ reduces by one the effective level of each program -- goal pair. The proof is by induction on the length $\ell$ of the given proof, with the base case ($\ell=0$) trivial. For the inductive step, we proceed according to the first rule used in the derivation.

  \case{backchain}
  Suppose $\bd$ begins with the \textbf{backchain} rule
  \[
    \mfa_1 \otimes \ipresgoal{j}{P}{A} \otimes \mfa_2 \derstep{\gamma} \mfa_1 \otimes \ipresgoal{j}{P}{G\delta\gamma} \otimes \mfa_2 \derivation{[id]} \NULL
  \]
  for some clause $\forall \vec x (G \hoe A_r) \in \elab(P)$, renaming-apart derivation $\delta=\{w_1/x_1, \ldots, w_n/x_n \}$ and unifier $\gamma$ such that $A = A \gamma = A_r\delta\gamma$. Let $\xi'$ and $\theta'$ be extensions of $\xi$ and $\theta$ mapping new constants and variables of level $j \geq i$ occurring in $G \delta \gamma$ but not in $\mfa$ to constants and variables of level $j-1$, and behaving like $\xi$ and $\theta$ everywhere else. By Lemmas~\ref{lem:clause-renaming} and~\ref{lem:cren-ep} it turns out that there is a clause $\forall \vec x (G \xi\theta \hoe A_r\xi\theta) \in \elab(P\xi\theta)$.
  
 Consider $\gamma' = \{ w_1\gamma\xi'\theta' / w_1, \ldots, w_n\gamma\xi'\theta' / w_n \}$ and observe that it is a legal substitution. It is the case that $\gamma'$ is an unifier for $A \xi\theta$ and $A_r\xi\theta\delta$. Actually $A \xi\theta$ does not contain variables $\{ w_1, \ldots, w_n\}$, hence $A \xi\theta \gamma' = A \xi\theta$. Since $\xi$ and $\theta$ do not contain $w_1, \ldots, w_n$, and since $\xi'\theta'$ behave on $A_r$ like $\xi\theta$, we have 
  \[
    A_r \xi\theta\delta\gamma' = A_r \xi\theta\{ w_1 \gamma \xi' \theta' / x_1, \ldots, w_n \gamma \xi' \theta'/ x_n \} = A_r \xi'\theta' \{ w_1 \gamma \xi'\theta' / x_1, \ldots, x_n \gamma \xi' \theta'/ x_n \} \enspace .
  \] 
  Now, since $\xi'$ and $\theta'$ do not contain variables $x_1, \ldots, x_n$, and since $\xi'\theta'$ behave on $A$ like $\xi\theta$,  we have 
  \[ 
    A' \xi'\theta' \{ w_1 \gamma \xi'\theta' / x_1, \ldots, w_n \gamma\xi'\theta' / x_n \} = A \{ w_1 \gamma/x_1, \ldots, w_n \gamma/x_n\} \xi'\theta' = A' \delta \gamma \xi'\theta' = A \xi'\theta' = A \xi\theta \enspace .
  \] 
  Similarly, $G\xi\theta\delta\gamma' = G\delta\gamma\xi'\theta'$. Note that $\mfa_1\xi\theta=\mfa_1\xi'\theta'$ and $\mfa_2\xi\theta=\mfa_2\xi'\theta$. Moreover,  $P\xi'\theta'=P\xi\theta$, $A \xi'\theta'=A \xi\theta$ and $G\delta\gamma\xi'\theta'$ have level $j-1$ or less, therefore we can replace the index $j$ with $j-1$ in the first state vector and by inductive hypothesis we get
  \[
    \mfa_1 \xi\theta\down \otimes \ipresgoal{j-1}{P\xi\theta}{A\xi\theta}\otimes \mfa_2 \xi\theta\down \derstep{\gamma'} \mfa_1 \xi'\theta'\down \otimes \ipresgoal{j-1}{P\xi'\theta'}{G\delta\gamma\xi'\theta'} \otimes \mfa_2 \xi'\theta'\down \derivation{[id]} \NULL
  \]
  \case{and}
  Suppose $\bd$ begins with the {\bf and} rule
  \[
  \mfa_1 \otimes \ipresgoal{j}{P}{G_1\wedge G_2} \otimes \mfa_2
    \derstep{\wedge} \mfa_1 \otimes \ipresgoal{j}{P}{G_1}\otimes
    \ipresgoal{j}{P}{G_2} \otimes \mfa_2 \derivation{[id]}  \NULL \enspace .
  \]
  Applying the induction hypothesis to the shorter deduction starting
  from the second step, we have a deduction
 \begin{multline*}
   \mfa_1\xi\theta\down \otimes \ipresgoal{j-1}{P\xi\theta}{(G_1\wedge G_2)\xi\theta}\otimes  \mfa_2\xi\theta\down \derstep{\wedge} \\
        \mfa_1\xi\theta\down  \otimes \ipresgoal{j-1}{P\xi\theta}{G_1\xi\theta} \otimes \ipresgoal{j-1}{P\xi\theta}{G_2\xi\theta} \otimes  \mfa_2\xi\theta\down  \derivation{[id]} \NULL \enspace .
 \end{multline*}
  
  \case{instance}
  Suppose $\bd$ begins with the \textbf{instance} rule
  \[
    \mfa_1\otimes \ipresgoal{j}{P}{\exists_y G} \otimes \mfa_2\derstep{\exists} \mfa_1\otimes \ipresgoal{j}{P}{G[t/y]}\otimes \mfa_2 \derivation{[id]} \NULL
  \]
  Consider $\xi'$ and $\theta'$ which behave like $\xi$ and $\theta$, but map all fresh constants and variables of level $l \geq i$ in $t$ to fresh constants and variables respectively of level $l-1$ of the same type. Then, by inductive hypothesis, we have the derivation
  \[
     \mfa_1 \xi'\theta'\down \otimes \ipresgoal{j-1}{P\xi'\theta'}{G[t/y]\xi'\theta'} \otimes \mfa_2 \xi'\theta'\down \derivation{[id]} \NULL
  \]
  Note that $\mfa_1\xi'\theta' = \mfa_1\xi\theta$,  $\mfa_2\xi'\theta' = \mfa_2\xi\theta$, $P\xi'\theta' = P\xi\theta$ and $G[t/y]\xi'\theta' = G\xi\theta[t\xi'\theta'/y]$, hence we have the following deduction:
  \[
    \mfa_1 \xi\theta\down \otimes \ipresgoal{j-1}{P\xi\theta}{(\exists_y G)\xi\theta} \otimes  \mfa_2 \xi\theta\down \derstep{\exists} \mfa_1 \xi\theta\down \otimes \ipresgoal{j-1}{P\xi\theta}{G\xi\theta[t\xi\theta'/y]} \otimes \mfa_2 \xi\theta\down \derivation{[id]} \NULL \enspace .
  \]
  
  \case{generic}
  Suppose $\bd$ begins with the \textbf{generic} rule
  \[
    \mfa_1 \otimes \ipresgoal{j}{P}{\forall_y G} \otimes \mfa_2 \derstep{\forall}
    \mfa_1 \otimes \ipresgoal{j+1}{P}{G[d/y]} \otimes \mfa_2 \derivation{[id]} \NULL
  \]
  where $d$ is a constant of level $j+1$. Consider $\xi'$ which behaves like $\xi$ but maps the constant $d$ to a fresh constant $d'$ of level $j$ and the same type of $d$. Then, by inductive hypothesis we have
  \[
    \mfa_1 \xi'\theta\down \otimes \ipresgoal{j}{P\xi'\theta'}{G[d/y]\xi\theta'} \otimes \mfa_2 \xi'\theta\down \derivation{[id]} \NULL
  \]
  with $\mfa_1\xi'\theta = \mfa_1\xi\theta$,  $\mfa_2\xi'\theta = \mfa_2\xi\theta$, $P\xi'\theta = P \xi\theta$, $G[d/y]\xi'\theta = G\xi\theta[d'/y]$. Hence, we have the deduction
  \[
    \mfa_1 \xi\theta\down \otimes \ipresgoal{j-1}{P\xi\theta}{(\forall_y G)\xi\theta} \otimes \mfa_2  \xi\theta\down \derstep{\forall}  \mfa_1 \xi\theta\down \otimes  \ipresgoal{j}{P\xi\theta}{G\xi\theta[d'/y]} \otimes \mfa_2 \xi\theta\down  \derivation{[id]} \NULL 
  \]

  \case{other cases} 
  They are similar to the previous ones.
\end{proof}

\begin{lem}
  \label{lem:level-increase}
  Suppose $\mfa$ is a state vector whose states are all at level $i > 0$ or more. Let $\xi$ be a constant replacer mapping constants of level $j \geq i$ occuring in $\mfa$ to constants of level $j+1$. If there is a flat deduction
  \[ \bd: \nullsres{\mfa}{[id]} \enspace ,\]
  then there is a flat deduction
  \[ \nullsres{\mfa\xi\upvec}{[id]} \]
  of the same length.
\end{lem}
\begin{proof}
  Note that by Lemma~\ref{lem:derivation-fresh} we may assume that fresh variables and constants introduced in $\bd$ are disjoint from $\oc(\xi)$.  The proof is by induction on the length of the derivation. The base case is trivial. For the inductive step, we consider the first rule used in the derivation. 
  
  \case{backchain}
  Suppose $\bd$ begins with the \textbf{backchain} rule
  \[
    \mfa_1 \otimes \ipresgoal{j}{P}{A} \otimes \mfa_2 \derstep{\gamma} \mfa_1 \otimes \ipresgoal{j}{P}{G\delta\gamma} \otimes \mfa_2 \derivation{[id]} \NULL
  \]
  for some clause $\forall \vec x (G \hoe A_r) \in \elab(P)$, renaming-apart derivation $\delta=\{w_1/x_1, \ldots, w_n/x_n \}$ and unifier $\gamma$ such that $A = A \gamma = A_r\delta\gamma$. By Lemma~\ref{lem:cren-ep} turns out that there is a clause $\forall \vec x (G \xi \hoe A_r\xi) \in \elab(P\xi)$. Consider now:
  \begin{itemize}
    \item  $\delta' = \{x_1/w'_1, \ldots, x_n/w'_n\}$,  a new renaming apart derivation, where $w'_1, \ldots, w'_n$ are fresh variables, whose level is one more than the level of the corresponding $w_i$;
    \item $\gamma' = \{ w_1\gamma\xi / w'_1, \ldots, w_n\gamma\xi / w'_n \}$. 
  \end{itemize}
  Note that $\gamma'$ is a legal substitution, since the level of each $w_i\gamma\xi$ is at most the level of $w_i$ plus one. It is the case that $\gamma'$ is an unifier for $A \xi$ and $A_r\xi\delta'$, since $A\xi\gamma'=A\xi$ and 
  \[
    A_r \xi\delta'\gamma' = A_r \xi\{ w_1 \gamma \xi / x_1, \ldots, w_n \gamma \xi/ x_n \} = A_r \{ w_1 \gamma \xi / x_1, \ldots, w_n \gamma \xi/ x_n \} \xi = A_r \delta \gamma \xi = A \xi \enspace .
  \] 
  In the same way, we have that $ G\xi \delta' \gamma = G\delta\gamma\xi$. Therefore, by induction on the derivation starting from the second state vector, we have
  \[
    \mfa_1 \xi\upvec \otimes \ipresgoal{j+1}{P\xi}{A\xi}\otimes \mfa_2 \xi\upvec \derstep{\gamma'} \mfa_1 \xi\upvec \otimes \ipresgoal{j+1}{P\xi}{G\delta\gamma\xi} \otimes \mfa_2 \xi\upvec \derivation{[id]} \NULL \enspace .
  \] 

  \case{and}
  Suppose $\bd$ begins with the \textbf{and} rule
  \[
  \mfa_1 \otimes \ipresgoal{j}{P}{G_1\wedge G_2} \otimes \mfa_2
    \derstep{\wedge} \mfa_1 \otimes \ipresgoal{j}{P}{G_1}\otimes
    \ipresgoal{j}{P}{G_2} \otimes \mfa_2 \derivation{[id]}  \NULL \enspace .
  \]
  Applying the induction hypothesis to the shorter deduction starting from the second step, we get the deduction
 \begin{multline*}
   \mfa_1\xi\upvec \otimes \ipresgoal{j+1}{P\xi}{(G_1\wedge G_2)\xi}\otimes \mfa_2\xi\upvec \derstep{\wedge} \\
        \mfa_1\xi\upvec  \otimes \ipresgoal{j+1}{P\xi}{G_1\xi} \otimes \ipresgoal{j+1}{P\xi}{G_2\xi} \otimes \mfa_2\xi\upvec  \derivation{[id]} \NULL \enspace .
 \end{multline*}

  \case{generic}
  Suppose $\bd$ begins with the \textbf{generic} rule:
  \[
    \mfa_1 \otimes \ipresgoal{j}{P}{\forall_x G} \otimes \mfa_2 \derstep{\forall} \mfa_1 \otimes \ipresgoal{j+1}{P}{G[d/x]} \otimes \mfa_2 \derivation{[id]} \NULL \enspace .
  \]
  Let $\xi'$ be a costant replacer which behaves like $\xi$ but mapping $d$ to a constant $d'$ of level $j+2$. Note that $\mfa_1 \xi' = \mfa_1 \xi$, $\mfa_2 \xi' = \mfa_2 \xi$, $P \xi' = P \xi$ and $G\xi[d'/x] = G[d/x]\xi'$. By induction hypothesis, we have a derivation,
  \[
    \mfa_1 \xi \upvec \otimes \ipresgoal{j+1}{P\xi}{\forall_x G\xi} \otimes \mfa_2\xi\upvec \derstep{\forall} \mfa_1\xi'\upvec \otimes \ipresgoal{j+2}{P}{G[d/x]\xi'} \otimes \mfa_2\xi'\upvec \derivation{[id]} \NULL \enspace .
  \]

  \case{other cases}
  They are similar to the previous ones.
\end{proof}

\begin{cor}[Level increase]
  \label{cor:level-increase}  
  If there is a derivation of the form
  \[
    \ipresgoal{i}{P}{G} \derivation{id} \NULL
  \]
  then there is a derivation
  \[
    \ipresgoal{i+1}{P}{G} \derivation{id} \NULL \enspace .
  \]
\end{cor}
\begin{proof}
  First of all, without loss of generality we may assume that the derivation of $\ipresgoal{i}{P}{G}$ is flat. The result follow by applying Lemma~\ref{lem:level-increase} with $\xi$ being the identity function.
\end{proof}

\begin{cor}[Generic constants]
\label{cor:genconst}
  If there is a deduction
  \[
\nullsres{\ipresgoal{i+1}{P}{G[c/x]}}{[id]}
  \]  
where $c$ is a fresh constant of level $i+1$,
$x$ is not free in $P$  and $P,G$ is a
program-goal pair of level $i$, then for any positive term $t$ of
level $i$
  of the appropriate type there is a deduction
\begin{equation}
\label{eq:any-pos-t}
\nullsres{\ipresgoal{i}{P}{G[t/x]}}{[id]}
\end{equation}
Conversely if $P,G$ and $c$ satisfy the stated conditions and
for any positive term $t$ of level $i$ there is a deduction
\eqref{eq:any-pos-t}
then there is a deduction $\nullsres{\ipresgoal{i+1}{P}{G[c/x]}}{[id]}$.
\end{cor}

\begin{proof}
  For the first claim note that, without loss of generality, we may assume $x$ to be a variable of level $i$. Consider the constant renamer $\xi$ mapping $c$ to a constant $c'$ of level $i$. Using Lemma~\ref{lem:level-reduction} we get a derivation 
  $\nullsres{\ipresgoal{i}{P}{G[c'/x]}}{[id]}$. Then, by Lemma~\ref{lem:constant-replacement}, we get $\nullsres{\ipresgoal{i}{P}{G}}{[id]}$
  and finally, by Lemma~\ref{lem:sub-t}, $\nullsres{\ipresgoal{i}{P}{G[t/x]}}{[id]}$
  for any term $t$ of level $i$ or less.

  For the converse claim, consider the case when $t$ is a fresh constant $c'$ of level $i$ and we get a deduction $\nullsres{\ipresgoal{i}{P}{G[c'/x]}}{[id]}$. If $c$ is a fresh constant of level $i+1$ and $\xi$ is a constant replacer mapping $c'$ to $c$, then by Lemma~\ref{lem:level-increase} we have $\nullsres{\ipresgoal{i+1}{P}{G[c/x]}}{[id]}$.
\end{proof}







\subsection{Herbrand-Robinson Derivation: ``True'' Resolution}
The nondeterministic resolution rule system just given, equivalent
(see~\ref{app-subsec:res(t)}) to
the Miller-Nadathur-Pfenning-Scedrov \cite{uniform} definition of the
uniform proof system \HOHH, captures some of the behavior of a logic
programming implementation of \HOHH, especially since we use backchain
and \textsf{elab}, but not all. As in Prolog,
witnesses to existential queries are not magically guessed in one
step. They are produced incrementally by successive unifications when
the substitution generating rules {\bf backchain, true, identity}
occur in the derivation. In
other words \emph{instead of ranging over all possible terms, the
  existential quantifier's range can be restricted to witnesses
  produced by unification of queries and clause heads to determine how
  $\forall_L$ and $\exists_R$ should be executed. } This is a central
component of logic programming, and part of the import of Herbrand's
Theorem and Robinson's work in the Horn Clause context\footnote{See
  the discussion on p.130 of {\it Miller et al.}  \cite{uniform}
  regarding the deliberate omission of backchain in their definition
  of uniform proofs.}.

The purpose of this paper is to give several semantics with respect to
which the proof system of $\lambda$Prolog is sound and complete,
\emph{not} to analyze or optimize proof search, which has been done in
great detail by Nadathur
\cite{nadathur:thesis,milnad:higher-order,HHHorn,nadathurPP}. Thus,
we work with the non-deterministic RES(t) system, with all unifiers
allowed, not just those produced by UNIFY and SIMPL. We replace
existentially quantified variables with arbitrary terms, and show, in
the appendix that our system is equivalent to uniform proofs, and
hence to more efficient proof procedures in the literature. We also omit
discussion of the special kind of polymorphism present in
$\lambda$Prolog, which deserves a paper in its own right, as does the
use of focusing \cite{liangMiller2009,Andreoli2001}.

Our soundness  and completeness theorems, however, hold for these
practical resolution systems.

\begin{thm}
  \label{thm:herbrand}
  The resolution system $RES(Y)$ obtained by replacing  {\bf instance}  in the
  rule system
  above by the $\exists_Y$ rule
 \begin{equation}
    \label{eq:existsx}
\vladres{\hresgoal{\cc}{i}{P}{\exists x.G}}{\exists_Y}{
  \hresgoal{\cc}{i}{P}{G[Y/x]}}
\end{equation}
  where $Y$ is a fresh logic variable in {\sf fvar} of level $i$ or less,
  is equivalent to the one presented above: a goal succeeds in the
  first with some computed substitution $\theta$ if and only if it does
  with the second, with possibly another, more general  substitution.
 \end{thm}
A detailed treatment of this for \FOHH resolution can be found in
\cite{nadathurPP}, which also discusses how the argument must
change, using results in
\cite{nadathur:thesis,milnad:higher-order,HHHorn}, to adapt it to the
HOHH case.

Theorem~\ref{thm:equivalence} in the
Appendix, asserting the equivalence of RES(t) with uniform proofs,
obviously establishes equivalence of our resolution system with any
proof procedure equivalent to uniform proofs, in particular the one
presented using a rule similar to RES(Y) and a restricted class of
unifiers, in \cite{nadathurPP}.




\section{Semantics}


\subsection{The Main Obstacles}
\jlbox{severely reduce this subsection in size, or merge (shortened)
  with preface}
\label{subsec:obstacles}
If it is to capture the declarative content of
programs, a model theory for \HOHH must interpret 
an intuitionistic fragment
of Church's type 
theory. In particular,  it must be a model for the simply typed
lambda-calculus, sound for $\beta\eta$ conversion. 
It must also supply a carrier for predicates and
interpret the logical constants. Since quantifiers and connectives
are higher-order constants in the language,  the interpretation of terms 
cannot be easily divorced from the meaning of propositions. This
is one of the reasons we construct a model in which {\em the object of
truth values, a (parametrically) complete Heyting Algebra, is distinct
from the carriers of the types of propositions\/}.

\paragraph*{Impredicativity}
In Church's Type theory, instances of quantified formulas may have
arbitrarily large complexity, since variables of logical type range
over formulas. For example, any formula $G$ {\em including $\varphi$ itself\/}
may be an instance of the formula $\varphi= \exists x_o.x_o$ 
Thus the subformula relation does not give rise to a well-founded
order on formulas.
This means we may not define truth in a model by induction on
subformula structure.
For this reason assignments of truth values to
atomic formulas may not suffice to specify a model, and common elementary
properties (e.g.\@ monotonicity of truth in Kripke models) cannot be
defined by fiat for atoms and extended by induction to all formulas.
There are a number of ways to deal with this kind of problem, most
famously the approaches due to Takahashi, Tait, Andrews and Girard to prove
cut-elimination in impredicative systems (see \cite{andrews71,ccct05}
for extensions of these techniques to classical and intuitionistic versions 
 of Church's calculus). Sometimes one is able to define models,
and prove completeness by inducting on something else, e.g.\@ type
structure, term-complexity, or, as with our treatment of completeness
here, the number of iterations of an operator used to build initial
models. 

For this reason, truth in the models defined in this paper, {\em is
not inductively defined\/}. The conditions a definition of truth must satisfy with
respect to the connectives and quantifiers must be seen as constraints
to be met by a particular model construction.

There are, in addition, a series of requirements peculiar to logic
programming that must be met to supply a semantics with respect to
which resolution deductions on \HOHH are sound and complete. We
summarize them here.

  \paragraph*{Breaking truth-functionality of higher-order predicates:}

In $\lambda$Prolog we observe the following behavior. 
{\small 
\begin{alltt}
{\bf {\sf \large Program:}}
      module ntf.
      type p o.                               p :- q.
      type q o.                               q :- p.
      type f o -> o.                         f(p).   

{\bf {\sf \large Queries:}}
      ?- f(p).
      solved

      ?- f(q).
      no

      ?- f(p & p).
      no
\end{alltt}
}
The logical equivalence of $p$ and $q$, or of $p$ and $p\wedge p$
plays no role in the search for a proof of $f(q)$ or of  $f(p\wedge
p)$. In other words, the following inference fails in \ICTT, the underlying logic
  of $\lambda$Prolog, when $A,B$ are propositions and the
  type of $p$ is $o\imp o$:
\begin{equation}
\label{eq:leib}
\sfrac{A \longleftrightarrow B}{p(A)\imp p(B)}
\end{equation}

There are good programming language reasons for omitting this
inference. For example, a goal 
\begin{alltt}
 \hspace{2cm}      \(\cdots \) ,p(t) :- \(\cdots,\, \cdots\) ?- p(u).
\end{alltt}
with $t,u$ of the  type of propositions,
proceeds by higher-order unification of $t$ and $u$ and the appropriate
modification of the stack of current goals. If \eqref{eq:leib} were
a legitimate derived rule, we would have to invoke the entire theorem
proving mechanism of $\lambda$Prolog to see if $t$ and $u$ were 
logically equivalent.
Note that, if appropriate syntactic restrictions are imposed on programs,
truth-functionality becomes computationally feasible (see, for example,
\cite{wadge91,wadge13}).

Thus, 
if \HOHH-resolution is to be sound and complete with respect to our
model theory,
the meaning of the formula $f(p \wedge p)$ must not be
  identified with $f(p)$. Once again, 
  this is achieved by separating the object of truth values
  $\Omega$ from the carrier  {\sf D$_o$} of propositions. Although $p$
  and $p\wedge p$ will have the same truth value in  
  $\Omega$, they will not be identified in  {\sf D$_o$}, making it
  possible to consistently assign different truth values to 
 $f(p \wedge p)$ and  $f(p)$.


%



Our model theory for \UCTT is an indexed variant of Heyting Applicative
Structure semantics for the
intuitionistic version of Church's Theory of types given in
\cite{ccct05}. It brings together several of the ideas discussed above:
\begin{itemize}
\item Miller's indexed model theory \cite{modules} for a first-order fragment of \FOHH
(meaning is assigned to goals $G$ {\em with respect to ambient programs\/}
$P$);
\item   stratification of carriers to represent different index
levels, i.e.\@ change of ambient signatures;
\item separate carriers for programs and goals and their components, and an
object $\Omega$ of truth values
distinct from the carrier of logical formulae.
\end{itemize}

The carrier of the program type is a pre-ordered monoid, and so, in particular,
has a Kripke-like structure.
\jlbox{compare with Def of Unif Alg preorder structure. It may not be a pullback.}
Unlike Miller's model, the carriers of programs and goals do not
consist of the programs and goals themselves, but of interpreted
programs and goals.


The interpretation of logic (program-goal pairs)
 in our semantics arises from the
composition of two maps:
\[
\pee\times\gee\stackrel{\umod\times\umod}{\longrightarrow}\deeip\times \deeig
\stackrel{\hoint}{\imp}\Omega
\]
where $\pee$ and $\gee$ stand for the sets of terms of program and
goal type respectively, mapped by $\umod$ into carriers of
their type in some {\em Uniform Applicative Structure\/} $\dmod$,
defined below,  and $\hoint=\{\hoint_i:i\in\nn\}$ is an
integer-indexed family of maps from
stratified carriers $\deeip\times\deeig$ into the object of truth
values $\Omega$. Readers familiar with Miller's forcing notation
\cite{modules}  will
see that our interpretation $I_i(\umod(P),\umod(G))=\top_{\Omega}$
of the program-goal pair $P,G$
corresponds to his
\[
\hoint,P\mforces G.
\]

But we generalize his forcing in a number of ways: we include
 universal quantification (which requires adding the signature-levels
 $i\in\nn$), higher-order logic (which requires the presence of an
 applicative structure to model lambda-abstraction and application),
 and we factor our interpretation through a structure which supplies
 carriers $\deeip,\deeig$ for programs as well as goals.

We now give the details. We start by recapitulating
 the definition of a standard notion of model for the simply
 typed lambda-calculus, which will be built into our structures.
 They were first introduced systematically by H. Friedman in \cite{Friedman75},
although obviously implicit to a lesser or greater degree
in \cite{henkin50,lauchli:realizability-70,
plotkin:lambdaDef}.  (See also \cite{mi96} for a detailed discussion.)

\subsection{Applicative Structures}

Let ${\cal S}$
be a language for \ICTT, that is to say a set of typed constants.
\begin{defn}
\label{appstru}
A typed applicative structure $\langle \dee, \app, \const \rangle$
for ${\cal S}$
consists of an indexed family $\dee=\{\dee_\alpha\}$ of
sets $\dee_{\alpha}$
for each type $\alpha$, an indexed family $\app$ of
functions $\app_{\alpha,\beta}:
\dee_{\beta\alpha} \times \dee_\alpha \rightarrow
\dee_\beta$ for each pair $(\alpha, \beta)$ of types, and an
(indexed) intepretation function $\const=\{\const_\alpha\}$
taking constants of each type $\alpha$ to
elements of $\dee_\alpha$.
\end{defn}

%

A typed applicative structure is (functionally)
{\it extensional\/} if for every $f$, $g$ in
$\dee_{\beta\alpha}$, $\app(f,d) = \app(g,d)$ for all $d \in
\dee_\alpha$ implies $f = g$.

The models we construct in this paper are extensional.
It should be noted, however, that we do not work with equality, so
there is no way to explicitly specify (or prove) extensionality in the
theory. We now modify the semantics, first to introduce the richer notion of
indexed terms, and to model some of the relationships between certain
logical types.

We retain the notation ${\cal S}$ for the language we are working
over, by which we mean a set of \UCTT-typed constants, each with its
associated level.
\begin{defn}
\label{def:usappst}
A uniform ${\cal S}$-applicative structure $\fd= \langle
\dee, \app, \const, \inc, \pmonoid
 \rangle$ for \UCTT is an extensional, typed
applicative structure indexed over the natural numbers.
That is to say:
\begin{itemize}
\item $\dee$ is a family of sets $\tdee{i}{\alpha}$ for each index $i$ and type $\alpha$;
\item $\app$ is a family of functions, indexed over the natural numbers and pairs of types
\[
\app_{i\beta\alpha}:\tdee{i}{\beta\alpha}\times
\tdee{i}{\alpha}\imp\tdee{i}{\beta} ;
\]
\item $\const$ is a family of fuctions, indexed over natural numbers and types, such that $\const_{i\alpha}$ maps constants of type $\alpha$ and label $i$ to elements of $\tdee{i}{\alpha}$.
\end{itemize}

In addition, the structure contains a family of injective
functions $\inc_{i\pee\ree}:\deeir \rightarrow
\deeip$, $\inc_{i\gee\aee}:\deeia \rightarrow \deeig$,
$\inc_{i\aee\ree}:\deeir \rightarrow \deeia$, $\inc_{i\gee\hee}:\deeih
\rightarrow \deeig$ and $\inc_{(i+1)i\alpha}:\tdee{i}{\alpha} \rightarrow
\tdee{i+1}{\alpha}$ (and by composition $\inc_{jk\alpha}$ for $j \geq k$, with $\inc_{jj\alpha}$  the identity),
which must satisfy
$\inc_{ij\beta}(\app(g,s))=\app(\inc_{ij(\beta\alpha)}(g),
\inc_{ij\alpha}(s))$. Indexes may be omitted when clear from the context
or irrelevant to the discussion.

%


Also, $\deeip$ has a pre-ordered monoid\gabox{The term pre-ordered monoid is misleading since $\wedge_{\pee\pee\pee}$ neighte associative nor monotone} structure ${\sf M}_{i,\pee} = \langle \deeip, \sqq_{\pee}, \top_\pee, \mland_{\pee\pee\pee}, \mPi_{\pee(\pee\alpha)} \rangle$ where:
\begin{itemize}
     \item $\sqq_\pee$ is a pre-order over $\deeip$;
     \item $\mland_{\pee\pee\pee} = \inc_{i0(\pee\pee\pee)}(\const(\wedge_{\pee\pee\pee})) \in \dee_{i,\pee\pee\pee}$ is a binary operation;
     \item $\top_\pee \in \dee_{i,\pee}$ is an identity element for $\mland_{\pee\pee\pee}$;
     \gabox{Are  $\top_\pee$'s at different levels linked by $\inc$?}
     \item for every type $\alpha$, $\mPi_{\pee(\pee\alpha)} = \inc_{i0(\pee(\pee\alpha))}(\const(\Pi_{\pee(\pee\alpha)})) \in \tdee{i}{\pee(\pee\alpha)}$ is an unary operator
\end{itemize}
\gabox{Do we need to add the condition that $\mland_{\pee\pee\pee}$ and  $\mPi_{\pee(\pee\alpha)}$ are monontone?}
\noindent such that for $p, q \in \deeip$, $d \in \tdee{i}{\alpha}$ and $f \in 
\tdee{i}{\pee\alpha}$,
$p \sqq_\pee p \mland q$,
$\land_{\pee\pee\pee}$ is commutative
{\em up to equivalence induced by\/} $\sqq_{\pee}$
(that is to say $p\land q$ and $q \land p$ are $\sqq_\pee$ equivalent)
 and $\app(f,d) \sqq_\pee
\app(\mPi_{\pee(\pee\alpha)},f)$.
\end{defn}

\begin{defn}
  Given a uniform applicative structure, we will often refer to the
{\bf interpreted parameters\/} or, equivalently,
the  interpreted constants of the structure, as the elements
$\const(c)$ where $c$ ranges over logical constants,
 or the non-logical constants of ${\cal S}$. When the structure is
clear from context, we will write $\ol{c}$  (or just $c$ if clear from
context) for the interpreted 
constants, in particular the logical ones. With an abuse of notation,
we will also denote by $\ol{c}$ any possible injection of $\inc_{j0\alpha}(\const(c))$,
as it is needed in the context.
\end{defn}


\subsection{Uniform Algebras}

We define a valuation function from the uniform
applicative structure into an  {\em object of truth values\/},
namely a distributive lattice, or in some cases a Heyting algebra
$\Omega$%
\footnote{These structures must have enough suprema and infima
  as will be specified below.}.
A recapitulation of the definitions and some simple properties of
distributive lattices and Heyting Algebras is given in the appendix.
We distinguish
the lattice operations
from their formal counterparts in the syntax
 with subscripts $\Omega$ where
necessary.

\begin{defn}
\label{interp}
Let  $\Omega$ be a distributive lattice (or a Heyting Algebra).
An $\Omega$-{\bf valuation} on a uniform applicative structure \fd\
 is a function $\hoint:i \times
\deeip \times \deeig \rightarrow \Omega$, $i \in \mathbb{N}$,
satisfying the condition
\begin{equation}
\label{uaatom}
\thoint{i}{\inc_{\pee\ree}(r)}{\inc_{\gee\aee}(\inc_{\aee\ree}(r))} =
     \top_\Omega
\end{equation}
for rigid atoms, the monotonicity condition
\begin{equation}
\label{umonotone}
\thoint{i}{p}{g} \oleq
\thoint{i+1}{\inc_{(i+1)i}(p)}{\inc_{(i+1)i}(g)}
\end{equation}
 and the following equalities:
\begin{align}
\label{uatopgoal}
\thoint{i}{p}{\ol{\top}_\gee} &= \top_\Omega\\
\label{uaand}
\thoint{i}{p}{g \oland_{\gee\gee\gee} g'}& = \thoint{i}{p}{g} \land_\Omega
     \thoint{i}{p}{g'}\\
\label{uaor}
\thoint{i}{p}{g \olor_{\gee\gee\gee} g'} &= \thoint{i}{p}{g} \lor_\Omega
     \thoint{i}{p}{g'}\\
\label{uahorseshoe}
\thoint{i}{p}{q\olhorseshoe_{\gee\gee\pee} g} &=
	\thoint{i}{p\land_{\pee\pee\pee} q}{g}\\
\label{uainstance}
\thoint{i}{p}{ \app(\olSigma_{\gee(\gee\alpha)},f)} &= \biglor_{d \in
     \tdee{i}{\alpha}} \thoint{i}{p}{\app(f,d)}
     && \textrm{ for }f \in \tdee{i}{\gee\alpha}\\
\label{uageneric}
\thoint{i}{p}{\app(\olPi_{\gee(\gee\alpha)},f)} &=
     \bigland_{d \in \tdee{i}{\alpha}}
     \thoint{i}{p}{\app(f,d)} && \textrm{ for }f \in
     \tdee{i}{\gee\alpha}.
\end{align}
%
For ${\cal I}$ to be an $\Omega$-valuation, the infima and suprema on
the right hand side in the last two conditions must exist.
We also add the following
conditions on $\sqq_\pee$:
\begin{align}
\label{uainclusion}
p \sqq_\pee p' &\textrm{ implies } \thoint{i}{p}{g} \leq_\Omega
     \thoint{i}{p'}{g}\\
\label{uabackchain}
(g \horseshoe_{\pee\gee\ree} r) \sqq_\pee p &\textrm{ implies }
     \thoint{i}{p}{g} \leq_\Omega \thoint{i}{p}
     {\inc_{\gee\aee}(\inc_{\aee\ree}(r))}
\end{align}
\end{defn}

Since an ordering of
formulas in higher-order logic
that allows instances of quantified formulas to be considered as subformulas
does {\it not\/} give rise to a well-ordering (as discussed in
Section~\ref{subsec:obstacles})
the reader should resist any attempt to view conditions
(\ref{uaand},~\ref{uaor},~\ref{uahorseshoe}), etc., as inductive definitions.
They are the properties that a valuation must satisfy, and it remains
to be shown that such valuations exist. We will construct one as the
limit of finite approximations following Lemma~\ref{lem:t-inflationary} below.

Finally, we interpret the logic in our structure. We will need to
make use of assignments or environments to handle free variables.
A \dee-assignment $\aph$ is a function from the free variables of the
language into $\dee$ which respects indices and types. That is to say,
if $x$ of type $\alpha$ has label $i$, then $\aph(x)\in \dee_{i,\alpha}$.


\begin{defn}
\label{def:ua}
A {\bf uniform algebra} or {\bf model} for \UCTT
is a quadruple $\ang{\fd, \umod,\Omega,\hoint}$, where
$\fd$ is a uniform applicative structure
$\ang{\dee, \app, \const, \inc,\pmonoid}$,
$\umod$ is a family of maps $\umod=\{\umod_\varphi\}$
from terms of level $i$ and type $\alpha$  to $\dee_{i,\alpha}$
idsexed over \dee-assignments $\varphi$,
$\Omega$ is a distributive lattice  and
$\hoint$ an $\Omega$-valuation
satisfying certain conditions spelled
out below.
\begin{align*}
\umod(c)_\varphi = & \ \const(c)\\
\umod(x)_{\varphi} =& \ \varphi(x)\\
\umod(t_1 t_2)_\aph  =&\  \app(\umod(t_1)_\aph, \umod(t_2)_\aph) \\
\umod(\lambda x.t)_\aph = & \text{ the unique member of $\fd$ s.t.\@ $\app(\umod(\lambda x.t)_\aph,d) = \umod(t)_{\aph[d/x]}$}
\end{align*}
\end{defn}

\begin{rem}
     In the definition of $\umod(t_1 t_2)_\aph$ there are some implicit inclusions we have omitted to ease notation. If $t_1$ has level $i$, $t_2$ has level $j$ and $m= \max(i,j)$, the full definition of $\umod(t_1 t_2)_\aph$ should be
     \[
     \umod(t_1 t_2)_\aph  = \app(\inc_{mi}\umod(t_1)_\aph, \inc_{mj}\umod(t_2)_\aph) \enspace .
     \]
\end{rem}

\begin{rem}
     Like in the definition of $\elab$, in order to respect the condition that only variables in \fvar may occur freely, the definiton of $\umod(\forall x. t)_\varphi$ should be
     \begin{align*}
          \umod(\lambda x.t)_\aph = & \text{ the unique member of $\fd$ s.t.\@ $\app(\umod(\lambda x.t)_\aph,d) = \umod(t[y/x])_{\aph[d/y]}$}\\
          & \text{ where $y \in \fvar$ is a fresh variable replacing the
            bound variable $x \in \bvar$,}
     \end{align*}
     but we prefer to use the more informal notation, letting $x$ to benote both a variable in $\bvar$ and a fresh variable in $\fvar$.
\end{rem}

The notation $\umod(P)$ where $P$ is a program rather than a single program formula, is
meant to denote $\umod(\bigland P)$ if $P$ is non-empty, $\top_\pee$ otherwise.

The last condition for ensuring $\umod$ is a model of the
lambda-calculus
 above, called the {\em environmental model condition\/}
in \cite{mi96}, asserts the existence and uniqueness of a ``right''
interpretation of abstraction. It ensures the equivalence of semantic
 and syntactic substitution  even in the absence of
 functional extensionality.

%

\begin{lem}[Free variables]
\label{lem:umod-freevars}
Let $\ang{\fd,\umod,\Omega,\hoint}$ be a uniform algebra and $N$ any term. If $\varphi$ and $\varphi'$ are two \dee-assignments that agree on the free variables of $N$, then $\umod(N)_\varphi = \umod(N)_{\varphi'}$.
\end{lem}

\begin{lem}[Substitution]
\label{lem:substitutions}
Let $\ang{\fd,\umod,\Omega,\hoint}$ be a uniform algebra and let $\sth$ be the substitution $\{t_1/x_1,\dots, t_n/x_n\}$,
abbreviated $[\overline{t}/\overline{x}]$. Then, for any term $N$, $\umod(N\sth)_\aph=
\umod(N)_{\aph[\overline{\umod(t)_{\aph}}/\overline{x}]}$.
\end{lem}
The proofs of these lemmas are routines, by induction on the structure of $\lambda$-terms.
For the case $N=\lambda x.M$, we must appeal to
the uniqueness provision of the environmental model condition
of Definition~\ref{def:ua}.

\begin{lem}
\label{lem:preservation}
In any uniform algebra $\ang{\fd,\umod,\Omega,\hoint}$, if $\forall \vec x (G \hoe A_r) \in \elab(P)$, then for any renaming-apart substitution $\delta=\{w_1/x_1, \ldots, w_n/x_n\}$, substitution $\gamma$ and \dee-assignment $\varphi$, we have $\umod(G \delta\gamma \hoe A_r\delta\gamma)_\varphi  \sqq_{\pee} \umod(P\gamma)_\varphi$.
\end{lem}
\begin{proof}
The lemma is trivially true if $P$ is empty. Therefore, assume $P$ to be non-empty.
By definition of $\elab$, there exists $D \in P$ such that $\forall \vec x (G \hoe A_r) \in \elab(D)$. Moreover, $\umod(P\gamma)_\varphi = \bm{\bigwedge}_{D' \in P\gamma} \umod(D')_\varphi \sqsupseteq_\pee \umod(D\gamma)_\varphi$, hence it is enough to prove that $\umod(G \delta\gamma \hoe A_r\delta\gamma)_\varphi  \sqq_{\pee} \umod(D\gamma)_\varphi$. The proof is by induction on the complexity of the program formula $D$.

If $D = G \hoe A_r$ there are no bound variables. Hence, we have $\delta = \epsilon$ and the result is trivial since $G\gamma \hoe A_r \gamma = D \gamma$.


If $D = D_1 \wedge D_2$, then assume without loss of generality that $\forall \vec x (G \hoe A_r) \in \elab(D_1)$. By inductive hypothesis, we have  $\umod(G \delta\gamma \hoe A_r\delta\gamma)_\varphi  \sqq_{\pee} \umod(D_1\gamma)_\varphi \sqq_{\pee} \umod(D_1\gamma)_\varphi \mland \umod(D_2\gamma)_\varphi = \umod((D_1 \wedge D_2)\gamma)_\varphi$.

Finally, if $D = \forall x_1. D'$ then  $\forall x_2 \ldots x_n (G \hoe A_r) \in \elab(D')$. By inductive hypothesis, we have 
\[ 
     \umod(G \delta\gamma \hoe A_r\delta\gamma)_\varphi  = \umod(G \delta'\gamma\gamma' \hoe A_r\delta'\gamma\gamma')_\varphi  \sqq_{\pee} \umod(D'\gamma\gamma')_\varphi 
\] 
where $\delta' = \{w_2/x_2, \ldots, w_n/x_n\}$ and $\gamma' = \{ \gamma(w_1)/x_1 \}$. By Lemma~\ref{lem:substitutions} we have $\umod(D'\gamma\gamma')_\varphi = \umod(D'\gamma)_{\varphi[v/x_1]}$, where $v=\umod(\gamma(w_1))_\varphi$. Then, by eta-conversion and the definition of uniform algebra:
\[
     \umod(D'\gamma)_{\varphi[v/x_1]} = \umod((\lambda x_1. D'\gamma)x_1)_{\varphi[v/x_1]} = \app(\umod(\lambda x_1. D' \gamma)_{\varphi[v/x_1]}, v) \enspace .
\]
Since $\lambda x_1. D' \gamma$ does not contain $x_1$ free, using Lemma~\ref{lem:umod-freevars} we have $\app(\umod(\lambda x_1. D' \gamma)_{\varphi[v/x_1]}, v) = \app(\umod(\lambda x_1. D' \gamma)_\varphi, v)$ and
\[
     \app(\umod(\lambda x_1. D' \gamma)_\varphi, v) \sqq_\pee \app(\mPi_{\pee(\pee\alpha)}, \umod(\lambda x_1. D' \gamma)_\varphi) =  \umod(\forall x_1. D'\gamma)_{\varphi} = \umod(D\gamma)_{\varphi} \enspace . \qedhere
\]
\end{proof}

\subsection{Soundness}
Let $\mfa$ be the state $\ipresgoal{i}{P}{G}$.
For the rest of the paper we will abbreviate our
notation by writing
$\sem{G}{\varphi}$ instead of $\umod_\varphi(G)$ and
$\sem{\ipresgoal{i}{P}{G}}{\varphi}$ instead of
$\hoint_i(\umod_\varphi(P), \umod_\varphi(G))$. We extend the semantic operator
$\me{\_}$ to state vectors by defining
\[
\sem{\mfa_1\otimes\cdots\otimes\mfa_n}{\varphi} =
\sem{\mfa_1}{\varphi}\wedge_\Omega\cdots\wedge_\Omega\sem{\mfa_n}{\varphi}.
\]
and 
\[
  \sem{\NULL} = \top_\Omega \enspace .
\]
\begin{thm}
  \label{thm:soundness}
  Let $\mfa$ and $\mfb$ be state vectors. Suppose there is a resolution
\[
\resdots{\mfai}{\theta}{\mfaii}.
\]
Then in any uniform algebra, with any \dee-assignment
\[
\sem{\mfaii}{\varphi}\oleq\sem{\mfai\theta}{\varphi}.
\]
\end{thm}
The proof is by induction on the length of the given resolution.
Choose a uniform algebra $\cat{A}=\ang{\fd,\umod,\Omega,\hoint}$, where
$\fd$ is the uniform applicative structure
$\ang{\dee, \app, \const, \inc,\pmonoid}$.
Let $\varphi$ be a \dee-assignment.

\begin{proof}
There is nothing to show in the case of a proof of length 0.
Let us inductively assume  that for all deductions of length less that
some natural number $n>0$, the conclusion of the theorem holds.
Now suppose the given deduction
has length $n$. We consider all
possible first steps.

\case{true}
Then the deduction has the following shape:
\[
\midsrest{\mfa\otimes\ipresgoal{i}{P}{G}\otimes\mfb}
{\gamma}
{\mfa\gamma\otimes\ipresgoal{i}{P\gamma}{\top}\otimes\mfb\gamma}
{\theta'}
{\mfaii}
\]
where $G\gamma = \top$ and $\gamma\theta' = \theta$ when restricted
to the free variables in the initial state vector.
By induction hypothesis, applied to the deduction starting from the
second state vector, we have
\[
\sem{\mfaii}{\varphi}\oleq\sem{\mfa\gamma\theta'\otimes\ipresgoal{i}
{P\gamma\theta'}{\top}\otimes\mfb\gamma\theta'}{\varphi} = \sem{\mfa\theta\otimes\ipresgoal{i}
{P\theta}{\top}\otimes\mfb\theta}{\varphi} .
\]
Since $\top$ is (the normal form of) $G\gamma\theta'$ this is precisely what the theorem claims.

\case{backchain}
We assume the proof has the following shape:
\[
\midsrest{\mfa\otimes\ipresgoal{i}{P}{A}\otimes\mfb}
{\gamma}
{\mfa\gamma\otimes\ipresgoal{i}{P\gamma}{G\delta\gamma}\otimes\mfb\gamma}
{\theta'}
{\mfaii}
\]
where $\forall \vec x(G\hoe A_r)$ is a clause in $\elab(P)$, $\delta=\{w_1/x_1, \ldots, w_n/x_n\}$ is the renaming-apart substitution, $A_r \delta \gamma = A \gamma$ and $\gamma\theta'=\theta$ when
restricted to the free variables in the initial state vector. By induction hypothesis
\[
\sem{\mfaii}{\varphi}\oleq
\sem{\mfa\gamma\theta'\otimes\ipresgoal{i}{P\gamma\theta'}{G\delta\gamma\theta'}
\otimes\mfb\gamma\theta'}{\varphi} \enspace .
\]
Now
\[
  \sem{\mfa\gamma\theta'\otimes\ipresgoal{i}{P\gamma\theta'}{G\delta\gamma\theta'}
  \otimes\mfb\gamma\theta'}{\varphi}  =
\sem{\mfa\theta}{\varphi}\wedge_\Omega
\thoint{i}{\sem{P\gamma\theta'}{\varphi}}{\sem{G\delta\gamma\theta'}{\varphi}}
\wedge_\Omega
\sem{\mfb\theta}{\varphi} \enspace .
\]
By Lemma~\ref{lem:preservation}
$\umod(G\delta\gamma\theta'\hoe A_r\delta\gamma\theta')_{\varphi} \sqq_{\pee} 
 \umod(P\gamma\theta')_{\varphi}$, and by
condition \eqref{uabackchain} of the Definition~\ref{interp}
of a valuation, this implies
\[
\thoint{i}{\sem{P\gamma\theta'}{\varphi}}{\sem{G\delta\gamma\theta'}{\varphi}} \oleq
\thoint{i}{\sem{P\gamma\theta'}{\varphi}}{\sem{A_r\delta\gamma\theta'}{\varphi}}
=  \thoint{i}{\sem{P\theta}{\varphi}}{\sem{A_r\delta\gamma\theta'}{\varphi}} \enspace.
\]
From this, and the fact that $A_r\delta\gamma\theta'=A\gamma\theta'=A\theta$,  we conclude
\[
\sem{\mfaii}{\varphi}\oleq
\sem{\mfa\theta\otimes\ipresgoal{i}{P\theta}{A\theta}
\otimes\mfb\theta}{\varphi}
\]
which is what we wanted to show.

\case{and}
Suppose the first step is an instance of the {\bf and}-rule:
\[
\midsrest{\mfa\otimes\ipresgoal{i}{P}{G_1\wedge G_2}\otimes\mfb}
{\wedge}
{\mfa\otimes\ipresgoal{i}{P}{G_1}\otimes\ipresgoal{i}{P}{G_2}\otimes\mfb}
{\theta}
{\mfaii}.
\]
By the induction hypothesis
\begin{eqnarray*}
\sem{\mfaii}{\varphi}&\oleq &
\sem{\mfa\theta\otimes\ipresgoal{i}{P\theta}{G_1\theta}\otimes
\ipresgoal{i}{P}{G_2\theta}\otimes\mfb\theta}{\varphi}\\
&=&\sem{\mfa\theta}{\varphi}\wedge_\Omega
\thoint{i}{\sem{P\theta}{\varphi}}{\sem{G_1\theta}{\varphi}}\wedge_\Omega
\thoint{i}{\sem{P\theta}{\varphi}}{\sem{G_2\theta}{\varphi}}\wedge_\Omega
\sem{\mfb\theta}{\varphi}
\end{eqnarray*}
Since
\[
\thoint{i}{\sem{P\theta}{\varphi}}{\sem{(G_1\wedge G_2)\theta}{\varphi}} =
\thoint{i}{\sem{P\theta}{\varphi}}{\sem{G_1\theta}{\varphi}}\wedge_\Omega
\thoint{i}{\sem{P\theta}{\varphi}}{\sem{G_2\theta}{\varphi}},
\]
we have
\[
\sem{\mfaii}{\varphi}\oleq
\sem{\mfa\theta\otimes\ipresgoal{i}{P\theta}{(G_1\wedge G_2)\theta}\otimes
\mfb\theta}{\varphi},
\]
as we wanted to show.

\case{or}
Suppose the first step is an instance of the {\bf or}-rule:
\[
\midsrest{\mfa\otimes\ipresgoal{i}{P}{G_1\vee G_2}\otimes\mfb}
{\vee}
{\mfa\otimes\ipresgoal{i}{P}{G_k}\otimes\mfb}
{\theta}
{\mfaii}.
\]
for $k \in \{1,2\}$. By the induction hypothesis
\begin{eqnarray*}
  \sem{\mfaii}{\varphi}&\oleq &
  \sem{\mfa\theta\otimes\ipresgoal{i}{P\theta}{G_k\theta}\otimes
  \mfb\theta}{\varphi}\\
  &=&\sem{\mfa\theta}{\varphi}\wedge_\Omega
  \thoint{i}{\sem{P\theta}{\varphi}}{\sem{G_k\theta}{\varphi}}\wedge_\Omega
  \sem{\mfb\theta}{\varphi}
  \end{eqnarray*}
  Since
  \begin{eqnarray*}
    \thoint{i}{\sem{P\theta}{\varphi}}{\sem{G_k\theta}{\varphi}} &\leq& 
    \thoint{i}{\sem{P\theta}{\varphi}}{\sem{G_1\theta}{\varphi}}\vee_\Omega
    \thoint{i}{\sem{P\theta}{\varphi}}{\sem{G_2\theta}{\varphi}} \\
    &=&  \thoint{i}{\sem{P\theta}{\varphi}}{\sem{(G_1\vee G_2)\theta}{\varphi}}
  \end{eqnarray*}
  we have
  \[
  \sem{\mfaii}{\varphi}\oleq
  \sem{\mfa\theta\otimes\ipresgoal{i}{P\theta}{(G_1\vee G_2)\theta}\otimes
  \mfb\theta}{\varphi} \enspace, 
  \]
  as we wanted to show.
  
  \case{null / augment} 
  The proof for these rules are similar to the previous cases and are left to the reader.

\case{instance}
Suppose the given proof starts with an {\bf instance} step:
\[
\midsrest{\mfa\otimes\ipresgoal{i}{P}{\exists x.G}\otimes\mfb}
{\exists}
{\mfa\otimes\ipresgoal{i}{P}{G[t/x]}\otimes\mfb}
{\theta'}
{\mfaii},
\]
where $\theta'$ is an extension, possibly by the identity, of the
substitution $\theta$ to include fresh variables occuring in $t$.
Thus
$\mfa\theta'=\mfa\theta, \mfb\theta'=\mfb\theta, P\theta'=P\theta$ and
$G[t/x]\theta'=G\theta[t\theta'/x]$. Letting $u=t\theta'$
and applying the inductive hypothesis to the deduction starting with the
second state, we have

\begin{eqnarray*}
 \sem{\mfaii}{\varphi}&\oleq &
\sem{\mfa\theta}{\varphi}\owedge
\thoint{i}{\sem{P\theta}{\varphi}}{\sem{G\theta[u/x]}{\varphi}}\owedge
 \sem{\mfb\theta}{\varphi}\\
&\oleq&
\sem{\mfa\theta}{\varphi}\owedge
\thoint{i}{\sem{P\theta}{\varphi}}{\sem{\exists x.(G\theta)}{\varphi}}\owedge
 \sem{\mfb\theta}{\varphi}\\
&=&
\sem{\mfa\theta}{\varphi}\owedge
\thoint{i}{\sem{P\theta}{\varphi}}{\sem{(\exists x.G)\theta}{\varphi}}\owedge
 \sem{\mfb\theta}{\varphi},
\end{eqnarray*}
which is what we wanted to show.

\case{generic}
Suppose the given proof starts with a {\bf generic} step:
\[
\midsrest{\mfa\otimes\ipresgoal{i}{P}{\forall x.G}\otimes\mfb}
{\forall}
{\mfa\otimes\ipresgoal{i+1}{P}{G[c/x]}\otimes\mfb}
{\theta}
{\mfaii}.
\]
Let $\alpha$ be the type of $x$ (and $c$).
Pick $d\in \dee_{i,\alpha}$ and define $\cat{A}^d$ to be the uniform algebra
$\ang{\fd^d,\umod,\Omega,\hoint}$ where
$\fd^d$ is the uniform applicative structure
$\ang{\dee, \app, \const^d, \inc,\pmonoid}$
with the same domains as $\fd$, but where $\const_{(i+1)\alpha}^d(c)=\inc_{(i+1)i\alpha}(d)$,
(which we will also denote by $d$)
and hence $\umod(c)_\varphi=d$.

Applying the induction hypothesis (and the substitution lemma)
to the uniform algebra $\cat{A}^d$
we have
\begin{eqnarray*}
 \sem{\mfaii}{\varphi}&\oleq &
\sem{\mfa\theta}{\varphi}\owedge
\thoint{i+1}{\sem{P\theta}{\varphi}}{\sem{G\theta[c/x]}{\varphi}}\owedge
 \sem{\mfb\theta}{\varphi}\\
&=&
\sem{\mfa\theta}{\varphi}\owedge
\thoint{i}{\sem{P\theta}{\varphi}}{\sem{G\theta}{\varphi[d/x]}}\owedge
 \sem{\mfb\theta}{\varphi}\\
&=&
\sem{\mfa\theta}{\varphi}\owedge
\thoint{i}{\sem{P\theta}{\varphi}}{{\sf App}(\sem{\lambda x.G\theta}{\varphi}, d)}\owedge
 \sem{\mfb\theta}{\varphi}.
\end{eqnarray*}

Since the choice of $d$ was arbitrary in $\dee_{i,\alpha}$, we may
take the meet of the right hand side over all $d\in \dee_{i,\alpha}$,
obtaining, by clause \eqref{uageneric} of Definition~\ref{interp}
\begin{eqnarray*}
 \sem{\mfaii}{\varphi}&\oleq &
\sem{\mfa\theta}\owedge
\thoint{i}{\sem{P\theta}{\varphi}}{\sem{\Pi_{\gee(\gee\alpha)}\lambda x.G\theta}{\varphi}}\owedge
 \sem{\mfb\theta}{\varphi}\\
&=&
\sem{\mfa\theta}\owedge
\thoint{i}{\sem{P\theta}{\varphi}}{\sem{\forall x.G\theta}{\varphi}}\owedge
 \sem{\mfb\theta}{\varphi},
\end{eqnarray*}
where we have acitly used the fact that
$\sem{\Pi_{\gee(\gee\alpha)}\lambda x.G\theta}{\varphi}=
{\sf App}(\overline{\Pi}_{\gee(\gee\alpha)},\sem{\lambda x.G\theta}{\varphi})$.
This completes the proof of the theorem.
\end{proof}




\section{Completeness}

We now establish completeness of resolution. First of all, we need to define the an auxiliary notation for the set of all instances of a program clause.

\begin{defn}[extension of a program]
  \label{def:extension}
  Given a program formula $D$ of level $i$  and $j \geq i$, we define the \emph{extension} of $D$ at level $j$ as
  \[
    \begin{split}
      [j; D] = \{ A_r\gamma \leftarrow G\gamma \mid&\ \forall x_1 \ldots x_n (G \hoe A_r) \in \elab(D), \gamma=\{t_1/x_1, \ldots, t_n/x_n \} , \\
      &\ t_1, \ldots, t_n \text{ are positive terms of level $j$ or less}  \} \enspace .
    \end{split}
  \]
  If $P$ is a program of level $i$ and $j \geq i$, we define
  \[
    [j; P] = \bigcup_{D \in P} [j; D] \enspace .
  \]
\end{defn}

The following lemma is a direct consequence of the definition of the extension of a program, and will be useful in some of the proofs.

\begin{lem}
  \label{lem:extension-ts}
  Consider the transition system $\to_i$ on program formulas (of level $i$ or less) given by the following transition rules:
  \begin{itemize}
    \item $D_1 \wedge D_2 \to_i D_k$ for $k \in \{1,2\}$;
    \item $\forall x_\alpha D \to_i D[t_\alpha/x_\alpha]$ for any positive term $t$ of level $i$.
  \end{itemize}
  Then, for any program formula $D$, we have that $Q \in [i; D]$ iff $D \to_i^* Q \not \to_i$.
\end{lem}
\begin{proof}
  Assume $Q = G \hoe A  \in [i; D]$ and we prove $D \to_i^* Q \not \to_i$.  First of all, note that $Q \hoe A \not\to_i$ is immediate.  We know there exists $C = \forall \vec x(G_0 \hoe A_0) \in \elab(D)$ and $\gamma=\{t_1/x_1, \ldots, t_n/x_n\}$ such that $G=G_0\gamma$ and $A=A_0 \gamma$. We proceed by induction on the complexity of $D$.

  If $D = G' \hoe A'$, then it is $G_0 = G = G'$, $A_0 = A = A'$, hence $G \hoe A = D$, and the result follows by reflexivity of $\to_i^*$.
  If $D = D_1 \wedge D_2$, then it is  $C \in \elab(D_k)$ for $k \in \{1, 2\}$, hence $Q \in [i, D_k]$. By inductive hypothesis, $D_1 \wedge D_2 \to_i D_k \to_i^* Q$.
  If $D = \forall x D'$, then $C = \forall x C'$ with $C' \in \elab(D')$. If $\gamma'= \{t_2/x_2, \ldots, t_n/x_n\}$ we also have $Q' = G\gamma' \hoe A\gamma' \in [i; D']$. By inductive hypothesis, we have $D' \to_i^* Q'$. It is easy to check that $D'[t_1/x_1] \to_i^* Q'[t_1/x_1] = Q$, therefore we get the derivation $D \to D'[t_1/x_1] \to_i^* Q$.

  For the opposite equality, we proceed by induction on the length of the derivation of $D \to_i^* Q$. If this length is $0$, since $Q \not \to_i$, then $D = Q = G \hoe A$, and $Q \in [i; D]$ follows. Otherwise:
  \begin{itemize}
    \item if the derivation has the form $D_1 \wedge D_2 \to_i D_k \to_i^* Q$, then $Q \in [i; D_k]$ by inductive hypothesis. Therefore $Q = G\gamma \hoe A\gamma$ with $C = \forall \vec x(G \hoe A) \in \elab(D_k)$ and $\gamma=\{t_1/x_1, \ldots, t_n/x_n\}$. Obviously $C \in \elab(D_1 \wedge D_2)$, hence $Q \in [i; D_1 \wedge D_2]$.
    \item if the derivation has the form $\forall x D \to_i D[t/x] \to_i^* Q$, then $Q \in [i; D[t/x]]$ by inductive hypothesis. Therefore $Q= G\gamma \hoe A\gamma$ with $C = \forall \vec y(G \hoe A) \in \elab(D[t/x])$ and $\gamma=\{t_1/y_1, \ldots, t_n/y_n\}$. By Lemma~\ref{lem:elab-shift} there is $\forall x \forall \vec y C' \in \elab(\forall x. D)$ such that $C = C'[t/x]$, whence $Q \in [i; \forall x. D]$ with $\gamma = \{t/x, t_1/y_1, \ldots, t_n/y_n\}$. \qedhere
  \end{itemize}
\end{proof}

In what follows,
take $\Omega$ to be the Boolean algebra
$\{ \top, \bot \}$.\footnote{The reader should note
  that this does not make our semantics classical.  See
  Section~\ref{subsec:intuitionistic} for details.}
  In the next section we will consider broader classes of lattices
Consider the uniform applicative structure
(which we will call
the  {\bf term} structure)
$\tmod= \langle \dee, \app, \const, \inc,\pmonoid \rangle$
with $\dee_{i,\alpha}$ the set
of principal normal forms of \UCTT-terms of type $\alpha$ and level $i$ or less;
let $\app$ and $\const$ build the obvious terms and $\inc$ consist of
the appropriate inclusions.
The (commutative) pre-ordered monoid structure $\pmonoid$ is given as follows.
Let $P \sqq_\pee P'$ iff $[i ; P] \subseteq
  [i; P']$. The monoid and product operations are supplied by the
logical constants $\wedge_{\pee\pee\pee}$ and
$\Pi_{\pee{(\pee\alpha)}}$.
\gabox{IMPORTANT! We have a problem with $\top_\pee$, since there is no such element at
the synatactic level.}
The requirement on $\Pi$ in Definition~\ref{def:usappst},
namely that for any $f,d$
$\app(f,d) \sqq_\pee \app(\Pi_{\pee(\pee\alpha)},f)$ is easily checked,
since $\app(\Pi,f) = \Pi f$, which is beta-eta equivalent to
$\Pi\lambda x.fx$ and hence $\forall x.f x$
for any $x \in \bvar$. The requirement thus reduces to
  [$i;f d$]$\subseteq$ [$i;\forall x.f x$], which follows easily
from the definition of [$i;P$].

The quickest path to completeness is to show that,
with the term structure \tua,
the interpretation $I_{\vdash}$ given by
\[
  (I_{\vdash})_i(P,G)=\top_{\Omega} \iff
  \nullsres{\ipresgoal{i}{P}{G}}{id}
\]
gives rise to a genuine uniform algebra. Thus any closed
program-goal pair true in all uniform algebras is provable.
Since we want a bit more, namely how to define such a structure by
approximation from below, we will take a more circuitous approach.


We show how to define a higher-order analogue of the
least Herbrand Model, which validates only resolution-derivable
states, as the join of a countable directed set of
finitary partial models here called weak valuations. Weak valuations,
defined only on the term structure,
are closer to the proof theory than uniform algebras, and our main
result will
follow from showing that, in the limiting case, such valuations give
rise to a uniform algebra.

A {\bf weak valuation} for the uniform applicative structure $\tmod$
is an indexed family of functions
$\whoint=\{\whoint_i:i\in\nn\}$, where for each natural number $i$,
$\whoint_i: \deeip \times \deeig
  \rightarrow \Omega$  satisfying the two
  {\em monotonicity\/} conditions
\begin{eqnarray}
  \label{eq:moni}
  P \sqq_\pee Q & \wimp &\wthoint{i}{P}{G} \leq_\Omega
  \wthoint{i}{P}{G},\\
  \label{eq:monp}
  i \leq j  & \wimp &\wthoint{i}{P}{G} \leq_\Omega
  \wthoint{j}{P}{G}
\end{eqnarray}
as well as the conditions
\begin{eqnarray}
  \wthoint{i}{P}{\top} &=&\top_\Omega\\
  \wthoint{i}{P}{G_1\wedge G_2} & \leq_\Omega & \wthoint{i}{P}{G_1}\wedge_{\Omega} \wthoint{i}{P}{G_2}\\
  \wthoint{i}{P}{G_1\vee G_2} &  \leq_\Omega & \wthoint{i}{P}{G_1}\vee_{\Omega}
  \wthoint{i}{P}{G_2}\\
  \wthoint{i}{P}{Q \hoe G} &
  \leq_\Omega & \wthoint{i}{P\wedge Q}{G}\\
  \wthoint{i}{P}{\exists x_\alpha G } &\leq_\Omega & \bigvee
  \{\wthoint{i}{P}{G[t_\alpha/x_\alpha]} : t_\alpha\in\dee_{i,\alpha}\}
  \label{eq:vee}\\
  \wthoint{i}{P}{\forall x_\alpha G} & \leq_\Omega &
  \wthoint{i+1}{P}{G[c_\alpha/x_\alpha] } \quad
  \mbox{\small for all $c_\alpha\in \dee_{i+1,\alpha}\setminus
      \dee_{i,\alpha}$.}
\end{eqnarray}
Moreover, given $P \in \deeip$, $G \in \deeig$ and  fresh constants $c_\alpha, c'_\alpha \in \dee_{i+1,\alpha}$, the following invariant condition must hold:
\begin{equation}
  \label{eq:invariant}
  \wthoint{i+1}{P[c_\alpha/x_\alpha]}{G[c_\alpha/x_\alpha]} = \wthoint{i+1}{P[c'_\alpha/x_\alpha]}{G[c'_\alpha/x_\alpha]} \enspace .
\end{equation}
%

Note that we are using a more compact notation w.r.t.\@ the definition of uniform algebra (Definition~\ref{interp}), since we are only considering term structures here. In particular, we are compeltely omitting the $\inc$ operators, since they are just inclusions. 

We will call a weak interpretation {\em bounded\/} if it satisfies the additional condition

\begin{equation}
  \wthoint{i}{P}{A_r} \ \leq_\Omega \
  \bigvee \{\wthoint{i}{P}{G}: A_r\pmi G\in [i;P] \} \quad
  \mbox{\small for rigid atoms $A_r$ }. \label{eq:uwcptails}
\end{equation}

As we shall see below, the
only weak valuations we will need to prove the completeness theorem,
(the $T^n(I_{\bot})$ below) satisfy it automatically.

\begin{lem}
  If $\Omega$ is a lattice
  \gabox{I don't understand the point of this lemma and most of the next paragraph, since
    we do not work with complete lattice.
JL: Kleene's fixed point theorem only requires (1) T commutes with
$\omega$-chains, (2) $\Omega$ is a poset. But we need $\Omega$ chains
in the lattice of interpretations. Also, I think, $\bigvee,\bigwedge$
 \emph{parametrized} by terms. But this may be imprecise. Any
 countable set can be parametrized by some function on terms. So we
 probably mean monotone functions?
  }
  then the set ${\cal J}$ of  weak valuations form a lattice with the
  same suprema and infima and the same distributive laws
  with the partial order
  \[
    \whoint\sqq \whoint^\prime
    := (\forall i\in \nn)\  \whoint_i\sqq_i \whoint_i^\prime
  \]
  where the order $\sqq_i$ on the valuation fibers $\whoint_i$ is defined
  pointwise.
\end{lem}
Since the the order (on the $\whoint_i$ and the
$\whoint$) is built up pointwise, it is easily
seen to inherit the structure of a complete lattice from $\Omega$.
In particular, the lattice has a bottom element
given by $\whoint_{\bot}(P,G) = \top_\Omega \mbox{ if } G= \top
  \mbox{ else } \bot_\Omega$. The top weak valuation maps all arguments
to $\top_{\Omega}$.


We now define an operators on weak valuations  and show it is 
continuous.
\begin{defn}
  \label{def:tp}
  Let $T:{\cal J}\imp {\cal J}$ be the following function:
  \begin{eqnarray*}
    T(I)_i(P, \top) &=&\top_\Omega\\
    T(I)_i(P,A) &=& I_i(P,A)\vee \bigvee \{\wthoint{i}{P}{G}: A \pmi G \in [i;P] \}
    \quad \mbox{\small $A$ atomic} \\
    T(I)_i(P,G_1\wedge G_2) &=& \wthoint{i}{P}{G_1} \wedge_\Omega
    \wthoint{i}{P}{G_2}\\
    T(I)_i(P,G_1\vee G_2) &=&   \wthoint{i}{P}{G_1}\vee_\Omega
    \wthoint{i}{P}{G_2}\\
    T(I)_i(P,Q\hoe G) &=& \wthoint{i}{P\wedge Q}{G}\\
    T(I)_i(P,\exists x_\alpha G) &=& \bigvee_{t_\alpha\in\tdee{i}{\alpha}}
    \wthoint{i}{P}{G[t_\alpha/x_\alpha]} \\
    T(I)_i(P,\forall x_\alpha G) &=&
    \wthoint{i+1}{P}{G[c_\alpha/x_\alpha]} \text{ \small for any $c_\alpha\in \dee_{i+1,\alpha}\setminus\dee_{i,\alpha}$}
  \end{eqnarray*}
\end{defn}
\noindent
Note that, in the last rule, the choice of $c_\alpha$ is irrelevant, thanks to \eqref{eq:invariant}.


\warn{Stronger treatment of flex goals?}

Some comments on this defintion are in order. First, the reader should
resist any temptation to see it as an inductive definition on formula
structure. {\em $T$ is defined separately for each kind of logical
    formula in terms of the pre-existing values of $I$, itself already
    defined on all goals.\/} Thus none of the
impredicativity problems discussed in the introduction
arise in this definition.

The first clause of the definition, which is the truth-value analogue
of the $T_P$ operator of Apt and Van Emden \cite{lloyd}, may be
reformulated in the following simpler and more familiar way
\[
  T(I)_i(P,A) \ = \  \bigvee \{\wthoint{i}{P}{G}: A\pmi G \in [i;P] \}
  \quad \mbox{\small for $A$ an atomic goal formula}.
\]
The additional disjunct in our definition is a technical convenience
that enforces the inflationary character of $T$ for all weak interpretations.
In the proof of the completeness theorem we will only consider
interpretations which are finite iterates of $T$ on the least
interpretation $I_{\bot}$, in which case the condition
\eqref{eq:uwcptails} is automatically satisfied and $T$ is trivially
inflationary.

Below we show that $T$ is well defined, mapping weak valuations
to weak valuations. Until then we will only assume it maps weak
valuations to $\Omega$-valued functions.

Note that once we have shown $T$ is well-defined, its monotonicity
as an operator on weak valuations
is immediate since the order on weak valuations is defined pointwise,
and $T$ is defined exclusively in terms of the lattice
operators on $\Omega$, all of them monotone.


\begin{lem}
  \label{lem:t-inflationary}
  The operator $T$ is inflationary: $I\sqq T(I)$ for any weak
  valuation $I$. It is
  well-defined on ${\cal J}$, that is to say it takes  weak
  valuations to  weak valuations, and it is
  continuous: if $\{I_n:n\in\nat\}$ is an increasing sequence of
  weak valuations, then
  \[
    T\left(\bigsqcup_{n}I_n\right) = \bigsqcup_{n}T(I_n).
  \]
\end{lem}

\begin{proof}~
  \paragraph*{$T$ is inflationary}
  We begin by showing $T$ is inflationary, just as a function from weak
  valuations to $\Omega$-valued operations on program-goal pairs.
  This must be shown
  pointwise: for any $i,P,G$ we have
  $\wthoint{i}{P}{G}\leq_{\Omega} T(I)_i(P,G)$.
  We consider some of the various cases  for $G$:

  \case{$G$ is an atom $A$}
  This case is immediate, since the value of $T$ is defined as the
  supremum of the value of $I$ and another expression, as discussed
  above.

  %
  \case{$G_1\wedge G_2$} By assumption
  $\wthoint{i}{P}{G_1\wedge G_2}\leq_{\Omega} \wthoint{i}{P}{G_1}\wedge_{\Omega}
    \wthoint{i}{P}{G_2}$, which coincides with $T(I)_i(P,G_1\wedge G_2)$.

  \case{other cases} The remaining cases are similar and left to the reader.

  \paragraph*{$T$ is well defined}
  Now to show that $T$ is well defined, we must show that if $I$ is a
  weak valuation, then $T(I)$ is too.

  The monotonicity properties~(\ref{eq:moni},~\ref{eq:monp}) follow easily
  from the monotonicity of operators in the clauses of the definition of
  $T$, and the fact that the supremum in the atomic case of the
  definition of $T$ is taken over a bigger set of clauses for bigger
  programs.

  The logical cases are straightforward consequences of
  the inflationary character of $T$.
  For example,
  $T(I)_i(P,G_1\wedge G_2) = \wthoint{i}{P}{G_1}\wedge_{\Omega}
    \wthoint{i}{P}{G_2}$.
  Since $T$ is inflationary, the value of the right hand side is less
  than or equal to $T(I)_i(P,G_1)\wedge_{\Omega}T(I)_i(P,G_2)$.
  In the universal case,
  $T(I)_{i}(P,\forall x G) = \wthoint{i+1}{P}{G[c/x]} \leq
    \wthoint{i+2}{P}{G[c/x]}$. By the invariant condition \eqref{eq:invariant}, 
    $\wthoint{i+2}{P}{G[c/x]} =  \wthoint{i+2}{P}{G[c'/x]}$ for any constant $c'$ of level $i+1$, 
    and finally  $\wthoint{i+2}{P}{G[c'/x]} =  T(I)_{i+1}(P,\forall x G)$.
  The remaining cases involving logical constants similarly follow from the fact that $T$ is
  inflationary.

  For the invariant condition \eqref{eq:invariant}, the only non-trivial case is the atomic one. Given programs $\bar P \in \deeip$ and atom $\bar A \in \deeia$, let us denote $\bar P[c_\alpha/x_\alpha]$ and $\bar A[c_\alpha/x_\alpha]$ by $P$ and $A$ respectively, while $\bar P[c'_\alpha/x_\alpha]$ and $\bar A[c'_\alpha/x_\alpha]$ will be denoted by $P'$ and $A'$ respectively. It is enough to prove that $T(I)_{i+1}(P, A) \leq T(I)_{i+1}(P', A')$. We have that:
  \begin{equation*}
    T(I)_{i+1}(P, A) = I_{i+1}(P, A) \vee  \bigvee \{ I_{i+1}(P, G): A \pmi  G \in [i+1; P] \}
  \end{equation*}
  First, note that $I_{i+1}(P, A) = I_{i+1}(P', A')$. Consider now the clause $A \pmi G \in [i+1; P]$. We would like to write $G$ as $\bar G[c_\alpha/x_\alpha]$ for some goal $\bar G$, in order to use the invariant property of $I$, but this is not possible since $G$ might contain $x_\alpha$. Therefore, let $\hat x_\alpha$ be a fresh variable of type $\alpha$ and same level as $x_\alpha$, not occurring in $A \pmi G$ nor in $P$. Let $\hat A$, $\hat G$ and $\hat P$ be obtained by replacing $c_\alpha$ in $A$, $G$ and $P$ with $\hat x_\alpha$, in such a way that $A = \hat A[c_\alpha/\hat x_\alpha]$, $G = \hat G[c_\alpha/\hat x_\alpha]$ and $P =  \hat P[c_\alpha/\hat x_\alpha]$. Note also that $A' = \hat A[c'_\alpha/\hat x_\alpha]$ and $P'  =  \hat P[c'_\alpha/\hat x_\alpha]$.

  By Definition~\ref{def:extension}, there exists a clause $\forall x (G_0 \hoe A_0) \in \elab(P)$ and a substitution $\gamma = \{t_1/x_1, \ldots, t_n/x_n\}$ of level $i+1$ or less such that $G = G_0\gamma$ and $A = A_0\gamma$. Consider the constant replacer $\xi$ mapping $c_\alpha$ to $c'_\alpha$. By Lemma~\ref{lem:cren-ep} we have that there is a clause $\forall \vec x (G_0 \xi \hoe  A_0 \xi) \in \elab(P \xi) = \elab(P')$ and again by Definition~\ref{def:extension}, $A_0 \xi (\gamma \xi) \pmi G_0 \xi (\gamma \xi)  \in [i + 1; P']$. Note that $G_0 \xi (\gamma \xi) = G_0 \gamma \xi = G \xi = \hat G[c_\alpha/\hat x_\alpha]\xi = \hat G[c'_\alpha/\hat x_\alpha]$ and  $A_0 \xi (\gamma \xi)  = A_0 \gamma \xi = A \xi = \hat A[c_\alpha/\hat x_\alpha]\xi = \hat A [c'_\alpha/\hat x_\alpha]$, which is also equal to $A'$.

  Therefore, in the right side of $T(I)_{i+1}(P', A')$ the term $I_{i+1}(\hat P[c'_\alpha/\hat x_\alpha], \hat G[c'_\alpha/\hat x_\alpha])$ appears. By hypothesis, this is equal to the term  $I_{i+1}(\hat P[c_\alpha/\hat x_\alpha], \hat G[c_\alpha/\hat x_\alpha]) = I_{i+1}(P, G)$. Therefore, each disjunct in $T(I)_{i+1}(P, A)$ has a corresponding disjunct in
  $T(I)_{i+1}(P', A')$, proving $T(I)_{i+1}(P, A) \leq T(I)_{i+1}(P', A')$.

  \paragraph*{$T$ is continuous}
  The proof of continuity reduces to applying various distributivity
  properties,  the assumed $\wedge\bigvee$
  distributive law,
  and the fact that suprema commute with suprema. More
  precisely, we make use of the facts that in any lattice with enough
  suprema
  \begin{itemize}
    \item $\bigvee_i\bigvee_j m_{ij} = \bigvee_j\bigvee_i m_{ij}$.
    \item If $k_n$ and $l_n$ are possible infinite increasing sequences in $\Omega$
          then $\bigvee_n k_n \wedge \bigvee_m l_m = \bigvee_n (k_n \wedge
            l_n)$. This follows  just from  $\wedge\bigvee$ distributivity
          in $\Omega$ and the fact that for any $m,n$ the meet
          $k_n\wedge l_m$ is bounded by $
            (k_m\vee k_n)\wedge (l_m\vee l_n)$ which is $k_{max(m,n)}\wedge
            l_{max(m,n)}$.
  \end{itemize}
  We consider the various cases.

  \case{Atomic case}
  \begin{eqnarray*}
    \Big(\bigsqcup_{n}T(I_n)\Big)_i(P,A) &=& \bigvee_n
    T(I_n)_i(P,A)\\
    &=& \bigvee_n \big[ (I_n)_i(P,A) \vee \bigvee
      \{(I_n)_{i}({P},{G}): A\pmi G \in [i;P] \} \big] \\
    &=& \Big(\bigvee_n (I_n)_i(P,A)\Big) \vee \bigvee
    \Big\{\bigvee_n(I_n)_{i}({P},{G}): A\pmi G \in [i;P] \Big\} \\
    &=& \Big(\bigsqcup_n I_n\Big)_i(P,A) \vee \bigvee
    \Big\{\Big(\bigsqcup_n(I_n)\Big)_{i}({P},{G}): A\pmi G \in [i;P] \Big\} \\
    &=& T\Big(\bigsqcup_{n}(I_n)\Big)_i(P,A)
  \end{eqnarray*}
  repeatedly using the commutation of joins with joins.

  \case{conjunction}
  \begin{eqnarray*}
    \Big(\bigsqcup_{n}T(I_n)\Big)_i(P,G_1\wedge G_2) &=& \bigvee_n
    T(I_n)_i(P,G_1\wedge G_2) \\
    &=& \bigvee_{n}[(I_n)_i(P,G_1)
      \wedge(I_n)_i(P,G_2)]\\
    &=& \bigvee_{n}(I_n)_i(P,G_1)
    \wedge
    \bigvee_{n}(I_n)_i(P,G_2)\\
    &=& \Big(\bigsqcup_{n}(I_n)\Big)_i(P,G_1)
    \wedge
    \Big(\bigsqcup_{n}(I_n)\Big)_i(P,G_2)\\
    &=& T\Big(\bigsqcup_{n}(I_n)\Big)_i(P,G_1\wedge G_2).
  \end{eqnarray*}

  \case{$\exists$ case}
  \begin{eqnarray*}
    \Big(\bigsqcup_{n}T(I_n))_i(P,\exists x_\alpha G) &=& \bigvee_n
    T(I_n)_i(P,\exists x_\alpha G)\\
    &=& \bigvee_n \bigvee_{t_\alpha\in\tdee{i}{\alpha}}
    (I_n)_i(P,G[t_\alpha/x_\alpha])\\
    &=&  \bigvee_{t_\alpha\in \tdee{i}{\alpha}}
    \bigvee_n (I_n)_i(P,G[t_\alpha/x_\alpha])\\
    &=&  \bigvee_{t_\alpha\in \tdee{i}{\alpha}}\Big(\bigsqcup
    I_n\Big)_i(P,G[t_\alpha/x_\alpha])\\
    &=&  T\Big(\bigsqcup_{n}I_n\Big)_i(P,\exists x_\alpha G)
  \end{eqnarray*}

  \case{$\forall$ case}
  \begin{eqnarray*}
    \Big(\bigsqcup_{n}T(I_n)\Big)_i(P,\forall x_\alpha G) &=& \bigvee_n
    T(I_n)_i(P,\forall x_\alpha G)\\
    &=& \bigvee_n (I_n)_{i+1}(P,G[c_\alpha/x_\alpha])\\
    &=& \Big(\bigsqcup_n (I_n)\Big)_{i+1}(P,G[c_\alpha/x_\alpha])\\
    &=&  T\Big(\bigsqcup_{n}(I_n)\Big)_i(P,\forall x_\alpha G)
  \end{eqnarray*}
  \noindent
  where it is assumed that
  we always take the same generic witness $c_\alpha$ for all $n$.
\end{proof}

By Kleene's fixed point theorem, $T$ has a least fixed point
$T^\infty(\whobot)$, which we will
call $\ifix$, and it is the supremum of
the finite iterates of $T$
applied to the least weak valuation
$\hoint_\bot,T(\hoint_\bot),T^2(\hoint_\bot),\ldots$.

%
%
%

%
%

We now want to show that those program-goal pairs mapped to
$\top_{\Omega}$ by $\ifix$, the least
fixed point of $T$, are
precisely those provable by resolution with the identity
substitution. This will give us completeness of uniform algebras once
we show $\ifix$ is a genuine valuation, and not just a weak one.
Because of the different way weak valuations
treat universal goals,
this is not so straightforward, and
first requires showing directly the soundness of $\ifix$.

We need to show first that one property of genuine valuations is
shared by $\ifix$, namely \eqref{uabackchain}.

\begin{lem}
\label{lem:tp-clause}
  If $G \hoe A_r \sqq_\pee P$ then $\ifix_i(P,G) \leq_{\Omega} \ifix_i(P,A_r)$.
  \gabox{Why put this lemma here instead of inside (or just before of) Lemma~\ref{lem:istar-is-int}?}
\end{lem}
\begin{proof}
By definition, $G \hoe A_r \sqq_\pee P$ means $[i; G \hoe A_r] \subseteq [i; P]$ for each $i > 0$. Since $[i; G \hoe A_r] = \{ G \hoe A_r\}$, this is equivalent to $G \hoe A_r \in [i; P]$.

By Definition~\ref{def:tp} of $T$, for any natural number $n$ (and any
index $i$)
\[
T^{n+1}(\whoint_\bot)_i(P,A)  = T^{n}(\whoint_\bot)_i(P,A) \vee \bigvee_{A \pmi G' \in [i; P]}
T^n(\whoint_\bot)_i(P,G').
\]
Since $A \pmi G \in [i; P]$,
$T^n(\whoint_\bot)_i(P,G)$ is one of the terms
over which the supremum is being taken.
Therefore, $T^n(\whoint_\bot)_i(P,G)\leq
T^{n+1}(\whoint_\bot)_i(P,A)$. Taking suprema, (first on the right,
then on the left), we obtain
$
T^\infty(\whoint_\bot)_i(P,G)\leq
T^\infty(\whoint_\bot)_i(P,A)
$
as we wanted to show.
\end{proof}

\begin{lem}
  \label{lem:ifix_sound}
  Resolution in \UCTT is sound for the  interpretation $\ifix$ in the
  following sense: if $\mfa$ is a state vector and there is a resolution
  \[
    \nullsres{\mfa}{id} \enspace,
  \]
  then for each state $\ipresgoal{i}{P}{G}$ in $\mfa$,
  $\ifixx_i(P,G)=\top_\Omega$.
\end{lem}
This proof, by induction on the length of the given resolution,
is quite similar to the soundness proof for
uniform algebras, in fact simpler in many of the
inductive cases since only the identity
substitution is involved, but with the added wrinkle that we must
manipulate the infinite joins that define $\ifix$.
Without loss of generality, thanks to the instantiation lemma,
we may assume that the resolution is flat.
\begin{proof}
  If the proof has length $0$, then $\mfa = \NULL$ and the result is trivially true. We now suppose $n>0$ and that
  the result holds for any deduction of length less than
  $n$. Assume the resolution $\nullsres{\mfa}{id}$ has length $n$. If the derivation has the form
  \[
    \mfa \otimes \ipresgoal{i}{P}{G} \otimes \mfb \derstep{id} \mfa \otimes \mfc \otimes \mfb \derivation{[id]} \NULL
  \]
  where $\ipresgoal{i}{P}{G}$ is the selected state, the inductive hypothesis automatically gives the theorem for all states in $\mfa$ and $\mfb$. Therefore, the only state which is interesting to examine is the selected state.  We now consider the possible cases arising in the first step.
  
  \case{null}
  Suppose the given resolution is of the form
  \[
    \mfa \otimes \ipresgoal{i}{P}{\top} \otimes \mfb \derstep{} \mfa \otimes \mfb \derivation{[id]} \NULL
  \]
  For the state  $\ipresgoal{i}{P}{\top}$, we have that $\ifixx_i(P, \top) = \top_\Omega$ by definition of weak evaluation.

  \case{true}
  The case for {\em true\/} resolution step is vacuous. Since all substitutions are identities in a flat derivation, the {\em true\/} step is not allowed.

  \case{backchain}
  Suppose the given resolution is of the form
  \[
    \mfa\otimes\ipresgoal{i}{P}{A}\otimes \mfb \derstep{\theta}
    \mfa\otimes\ipresgoal{i}{P}{G\delta\theta}\otimes \mfb \derivation{[id]} \NULL
  \]
  where $\forall x (G \hoe A_r) \in \elab(P)$, $\delta =\{w_1/x_1, \ldots, w_n/x_n\}$ is the renaming apart derivation and $\theta$ the unificator of $A$ and $A_r\delta$, which is the identitity on the variables of the initial state vector.
  By inductive hypothesis $\ifixx_i(P, G\delta\theta)= \top_\Omega$. By definition of the extension of a program $A \pmi G \delta \theta \in [i; P]$. Therefore
  \begin{eqnarray*}
    \ifixx_i(P, A)&=&\bigvee_n T^{n+1}(I_\bot)_i(P, A)\\
    &=& \bigvee_n \Big(T^{n}(I_\bot)_i(P, A) \vee \bigvee\{ T^{n}(I_\bot)_i(P, G') \mid A \pmi G' \in [i; P] \}\Big)\\
    &\geq & \bigvee_n  T^{n}(I_\bot)_i(P, G\delta\theta) \\
    &=&\ifixx_i(P,G\delta\theta)\\
    &=& \top_{\Omega},
  \end{eqnarray*}
  as we wanted to show.

  \case{and}
  Suppose the given deduction is of the form
  \[
    \midnullsrest{\mfa\otimes\ipresgoal{i}{P}{G_1\wedge G_2}\otimes\mfb}
    {\wedge}{\mfa\otimes\ipresgoal{i}{P}{G_1}\otimes
      \ipresgoal{i}{P}{G_2}\otimes\mfb}{[id]}
  \]
  By induction hypothesis $\ifixx_i(P,G_k)=\top_{\Omega}$ for every $k \in \{1,2\}$.
  Now for every natural number $n$ we have
  \[
    T^{n+1}(I_\bot)(P,G_1\wedge G_2) =
    T^n(I_\bot)_i(P,G_1)\wedge_{\Omega}
    T^n(I_\bot)_i(P,G_2).
  \]
  Taking first the supremum of the left hand side and then of the right,
  and using $\wedge\bigvee$-distributivity we have
  \[
    \ifixx_i(P,G_1\wedge G_2) \geq \ifixx_i(P,G_1)\wedge_{\Omega}
    \ifixx_i(P,G_1),
  \]
  from which we obtain $\ifixx_i(P,G_1\wedge G_2)=\top_{\Omega}$, as we
  wanted to show.

  \case{or / augment}
  The proofs for these rules are similar to the previous ones, and
  left to the reader.

  \case{instance}
  Suppose the given resolution is of the following form
  \[
    \midnullsrest{\mfa\otimes\ipresgoal{i}{P}{\exists xG}\otimes \mfb}{\exists}
    {\mfa\otimes\ipresgoal{i}{P}{G[t/x]}\otimes \mfb}{[id]} .
  \]
  By induction hypothesis $\ifixx_i(P,G[t/x])=\top_{\Omega}$,
  whence
  \begin{eqnarray*}
    \ifixx_i(P,\exists xG)&=&\bigvee_n T^{n+1}(I_\bot)_i(P,\exists x G)\\
    &=& \bigvee_n \bigvee_{u\in D_{i,\alpha}}T^n(I_\bot)_i(P,G[u/x])\\
    &\geq &\bigvee_n T^n(I_\bot)_i(P,G[t/x])\\
    &=&\ifixx_i(P,G[t/x])\\
    &=& \top_{\Omega},
  \end{eqnarray*}
  as we wanted to show.

  \case{generic}
  Suppose the given resolution is of the following form
  \[
    \midnullsrest{\mfa\otimes\ipresgoal{i}{P}{\forall xG}\otimes \mfb}{\exists}
    {\mfa\otimes\ipresgoal{i+1}{P}{G[c/x]}\otimes \mfb}{[id]}
  \]
  for a fresh constant $c$ of label $i$. By induction hypothesis $\ifixx_{i+1}(P,G[c/x])=\top_{\Omega}$, whence
  \begin{eqnarray*}
    \ifixx_i(P,\forall xG)&=&\bigvee_n T^{n+1}(I_\bot)_i(P,\forall x G)\\
    &=& \bigvee_n T^n(I_\bot)_{i+1}(P,G[c/x])\\
    &\geq&\ifixx_i(P,G[c/x])\\
    &=& \top_{\Omega},
  \end{eqnarray*}
  with a fixed canonical choice of the witness $c$ for all the iterates of $T^n$.

\end{proof}

\begin{thm}
  \label{thm:completeness}
  Suppose $\ifixx_i(P,G) =\top_{\Omega}$. Then
  there is a resolution
  \[
    \nullsres{\ipresgoal{i}{P}{G}}{id}.
  \]

\end{thm}

\begin{proof}
  Since $\top_\Omega$ is join irreducible in $\Omega$, the condition
  $\ifixx_{i}(P,G) = \top_\Omega$ implies that for some
  $k\in\nat$ we must have
  $\hoforce{T^k(I_\bot)}{i}{P}{G}$.
  Thus it will suffice to
  show that this condition implies the conclusion of the
  theorem for any $k$.
  The proof is by induction on $k$ (not on goal structure).

  For $k=0$, the claim is trivial.
  $\hoforce{(I_\bot)}{i}{P}{G}$ can only happen if
  $G=\top$ in which case there is a one step resolution
  \[
    \srest{\ipresgoal{i}{P}{\top}}{\text{null}}{\NULL}
  \]

  Now we suppose the theorem holds for all $i,P,G$ for $k=n$ and
  assume $\,\hoforce{T^{n+1}(I_\bot)}{i}{P}{G}$.
  We may as well assume that $n+1$ is the least iteration $k$ of $T$ for
  which this holds, else we are done by the induction hypothesis.

  Now we consider the six cases that arise depending on the
  structure of $G$. The reader should {\em not\/} take this to be a
  proof by induction on the structure of $G$, which would lead us right
  into the trap of inducting over an impredicative higher-order formal
  system. We never appeal inductively to the truth of the claim for
  subformulas, only to the truth of the claim for smaller iterations of
  $T$.

  \case{atomic}
  If the goal is a flex atomic formula, the statement is vacuously true, since the $I_i^*(P, G) = \bot_\Omega$ in this case.
  If the goal is a rigid atomic formula $A_r$, then
  $\hoforce{\tpni}{i}{P}{A_r}$ and
  by join-irreducibility of $\top_\Omega$ (trivially in
  $\{\top,\bot\}$) one of the disjuncts
  on the right hand side of the first clause of Definition~\ref{def:tp} must evaluate to $\top_\Omega$.
  In particular this means for some
  $A_r \pmi G \in [i;P]$ we have
  $\hoforce{\tpn}{i}{P}{G}$, and by the induction hypothesis there is a
  resolution $\nullsres{\ipresgoal{i}{P}{G}}{id}$.

  By definition of the extension of the program, there is a clause $\forall \vec x (G_0 \hoe A_0) \in \elab(P)$ and $\gamma = \{t_1/x_1, \ldots, t_n/x_n\}$ such that $A_0 \gamma = A_r$ and $G_0 \gamma = G$, with terms $t_1, \ldots, t_n$ of level $i$ or less. Therefore, consider $\delta = \{w_1/x_1, \ldots, w_n/x_n\}$ a renaming-apart substitution and $\gamma' = \{t_1/w_1, \ldots, t_n/w_n\}$,
  we have
  \[
    \midnullsrest{\ipresgoal{i}{P}{A_r}}{\gamma'}{\ipresgoal{i}{P}{G}}{id}
  \]
  since $G_0 \delta \gamma' = G$ and $A_0 \delta \gamma' = A_r\gamma' = A_r$.

  \case{goal $\top$}
  If the goal is $\top$, there is a one-step SLD proof using
  the null rule.

  \case{disjunction}
  If the goal is of the form $G_1 \lor G_2$, then we
  have, by the definition of $T$,
  and irreducibility of $\top_\Omega$,
  $\hoforce{\tpn}{i}{P}{G_j}$ for some $j\in\{1,2\}$.
  By the induction hypothesis, we
  have a derivation $\nullsres{\hresgoal{}{i}{P}{G_j}}{id}$
  which by the {\em or} rule extends to
  $\midnullsrest{\hresgoal{}{i}{P}{G_1 \lor G_2}}
    {\vee}{\hresgoal{}{i}{P}{G_j}}{id}$.

  \case{conjunction}
  If the goal is $G_1 \land G_2$,
  then, by definition of $T$, we have both
  $\hoforce{\tpn}{i}{P}{G_1}$ and
  $\hoforce{\tpn}{i}{P}{G_2}$.
  By inductive hypothesis on $n$
  iterations of $T$ there are two resolutions
  $\bd_1: \nullsres{\hresgoal{}{i}{P}{G_1}}{id}$,
  $\bd_2: \nullsres{\hresgoal{}{i}{P}{G_2}}{id}$.

  Consider new derivations $\bd_1'$ and $\bd_2'$ obtained from $\bd_1$ and $\bd_2$ by renaming variables and constant in such a way that the only variables and constants in common between them are those which are already in common in the original state vectors.

  Then we may build a new derivation starting from $\ipresgoal{i}{P}{G_1} \otimes \ipresgoal{i}{P}{G_2}$ by first appending the state $\ipresgoal{i}{P}{G_2}$ to all state vectors of $\bd'_1$, then continuing with $\bd'_2$:
  \[
    \rlap{$\underbrace{\ipresgoal{i}{P}{G_1} \otimes \ipresgoal{i}{P}{G_2} \derivation{id}  \ipresgoal{i}{P}{G_2}}_{\bd_1 \otimes \ipresgoal{i}{P}{G_2}}$}
    \ipresgoal{i}{P}{G_1} \otimes \ipresgoal{i}{P}{G_2} \derivation{id}  \overbrace{\ipresgoal{i}{P}{G_2}
    \derivation{id} \NULL}^{\bd'_2} \enspace .
  \]
  Prefixing with the step $\ipresgoal{i}{P}{G_1 \wedge G_2} \derstep{\wedge}  \ipresgoal{i}{P}{G_1} \otimes \ipresgoal{i}{P}{G_2}$ we have the required derivation.


  \case{implication}
  In case $\hoforce{\tpni}{i}{P}{D \horseshoe G'}$, we have
  $\hoforce{\tpn}{i}{P \cup D}{G'}$ and so by the induction
  hypothesis, we have a derivation $\nullsres{\hresgoal{}{i}{P \cup
        D}{G'}}{id}$.  By the {\em
      augment} rule, we then have the desired SLD resolution
  $\midnullsrest{\hresgoal{}{i}{P}{D \horseshoe G'}}{\hoe}{\hresgoal{}{i}{P \cup
        D}{G'}}{id}$.

  \case{exists}
  Assume $\hoforce{\tpni}{i}{P}{\exists x_\alpha G'}$. Then
  \[
    \top_\Omega =\bigvee_{t_\alpha\in D_{i.\alpha}}
    \hotpn{i}{P}{G'[t_\alpha/x_\alpha]}.
  \]
  Since $\top_\Omega$ is join-irreducible in $\Omega$, for some positive\gabox{We should check that we use positive terms whenever we need them.} term $t$ of level $i$ or less,
  $\hoforce{\tpn}{i}{P}{G'[t_\alpha/x_\alpha]}$ and by the induction hypothesis
  there is a derivation $\nullsres{\hresgoal{}{i}{P}{G'[t_\alpha/x_\alpha]}}{id}$.
  This new derivation can be extended to
  $\midnullsrest{\hresgoal{}{i}{P}{\exists x_\alpha G'}}{\exists}
    {\hresgoal{}{i}{P}{G'[t_\alpha/x_\alpha]}}{id}$.

  \case{forall}
  Suppose $G$ is of the form $\forall x_\alpha G'$, so  $\hoforce{\tpni}{i}{P}{\forall x_\alpha G'}$. Then, for some fresh constant $c_\alpha$ of level $i+1$, we have $\hoforce{\tpn}{i+1}{P}{G'[c_\alpha/x_\alpha]}$. By induction hypothesis
  there is a deduction
  \[\nullsres{\ipresgoal{i+1}{P}{G'[c_\alpha/x_\alpha]}}{id}\]
  and therefore
  \[
    \midnullsrest{\ipresgoal{i}{P}{\forall x_\alpha G'}}
    {\forall}{\ipresgoal{i+1}{P}{G'[c_\alpha/x_\alpha]}}{id} \enspace .
  \]
\end{proof}

We note that we have shown that \ifix\ coincides with the
proof-theoretic mapping $I_{\vdash}$ defined by
\[
  (I_\vdash)_i(P,G)=\top_{\Omega} \text{ just in case }
  \nullsres{\ipresgoal{i}{P}{G}}{id}.
\]

Now we will show that \ifix\ is an $\Omega$-interpretation for
$\Omega=\{\top_\Omega,\bot_\Omega\}$ as given by Definition~\ref{interp}.
\begin{lem}
  \label{lem:istar-is-int}
  If $\Omega=\{\top_\Omega,\bot_\Omega\}$, then $\ifix$ is an $\Omega$-interpretation, and
  taking $\me{t}_{\varphi}=\nf{t\varphi}$, we have that
  $\ang{\tmod,\me{\ },\Omega,\ifix\ }$ is a uniform algebra.
\end{lem}

\begin{proof}
  We have already shown in Lemma~\ref{lem:tp-clause} that \ifix\ satisfies property
  \eqref{uabackchain} of valuations. 
  Inspection of the proof of Lemma~\ref{lem:ifix_sound} shows that we
  have already proven that the left-hand sides of equations
  (\ref{uatopgoal},~\ref{uaand},~\ref{uaor},~\ref{uahorseshoe},~\ref{uainstance})
  in the definition of valuation are greater than or
  equal to the right-hand sides. Since $T$ is well-defined, \ifix\ is a
  weak valuation, which shows the inequalities going the other way.
  Monotonicity conditions \eqref{umonotone} and \eqref{uainclusion} are already part of the definition of weak evaluation.  It remains to prove both \eqref{uaatom} and \eqref{uageneric}.

  The proof of \eqref{uaatom}, i.e., $\ifixx_i(A_r, A_r) = \top_{\Omega}$ for each rigid atom $A_r$, follows immediately follows from the completeness theorem, given the derivation
  \[
    \ipresgoal{i}{A_r}{A_r} \derstep{id} \top \derstep{null} \NULL \enspace .
  \]

  We use soundness and completeness of \ifix\ to show the universal case \eqref{uageneric}.
  Suppose $\ifixx_i(P,\forall x.G) = \top_{\Omega}$. Then
  there is a deduction $\nullsres{\ipresgoal{i}{P}{\forall x G}}{id}$
  hence a derivation
  \begin{equation}
    \label{eqn:cres}
    \nullsres{\ipresgoal{i+1}{P}{G[c/x]}}{id}.
  \end{equation}
  By the first half of the
  generic constants corollary (Corollary~\ref{cor:genconst})
  for any positive term $t$ of
  level $i$ of the appropriate type there is a deduction
  \begin{equation}
    \label{eqn:tres}
    \nullsres{\ipresgoal{i}{P}{G[t/x]}}{id}.
  \end{equation}
  By soundness of \ifix, for any such $t$ we have
  $\ifixx_i(P,G[t/x]) = \top_{\Omega}$. Taking infima we obtain
  $\bigwedge_{t} \ifixx_i(P,G[t/x]) = \top_{\Omega}$.

  \medskip

  \noindent
  Using the second half of the
  generic constants corollary (Corollary~\ref{cor:genconst}),
  which tells us that if for every suitable term $t$ there is a
  deduction \eqref{eqn:tres} then there is one of the form
  \eqref{eqn:cres},
  we obtain the converse, namely that
  if $\bigwedge_{t} \ifixx_i(P,G[t/x]) = \top_{\Omega}$ then
  $\ifixx_i(P,\forall x.G) = \top_{\Omega}$. Since $\Omega=\{\top_\Omega,\bot_\Omega\}$, 
  this is enough to ensure that $\ifixx_i(P,\forall x.G) =
  \bigwedge_{t} \ifixx_i(P,G[t/x])$.
\end{proof}

\begin{cor}[Completeness]
  Let $P$ and $G$ be closed
  program and goal formulas respectively.
  If $\thoint{i}{\umod(P)}{\umod(G)} =
    \top_\Omega$ in all uniform algebras $\umod$ (ranging over all $\Omega$)
  then there is an SLD proof of $G$ from $P$  whose
  substitutions are the identity on the free variables of $P$ and $G$.
\end{cor}
\begin{proof}
It suffices just to consider just
the two-element Boolean algebra, where, obviously,
$\top$ is join irreducible, the uniform algebra
$(\tmod,\me{\_ },\Omega,\ifix\ )$  defined in
Lemma~\ref{lem:istar-is-int} and the identity valuation $\varphi$. Then,
the result is an immediate corollary of Theorem~\ref{thm:completeness}.
\end{proof}

It should be noted that by a straightforward adaptation of the
proof of the corresponding theorem in Miller et al. \cite{uniform},
and Theorem~\ref{thm:equivalence} in the Appendix,
an SLD proof of $G$ from $P$ exists iff a proof of $\sequent{P}{G}$
exists in \ICTT. Thus,
$\thoint{i}{\umod(P)}{\umod(G)} =
  \top_\Omega$ in all uniform algebras $\umod$ if and only if
$\psem{\sequent{P}{G}} = \top$
in all \ICTT global models $\hemod$, as defined in
\cite{ccct05}.

\begin{thm}\label{hoprecomplete}
  Let $\Omega$ be a complete Heyting Algebra
  in which $\top_\Omega$ is join-irreducible.
  Assume $\thoint{i}{\umod(P)}{\umod(G)}_\varphi = \top_\Omega$ in all uniform
  algebras $(\tmod,\umod,\Omega,\hoint)$ and environments $\varphi$.
  Then, there is an SLD derivation $\nullres{\hresgoal{id}{i}{P}{G}}$.
\end{thm}
%
\begin{proof}
  We consider the uniform algebra
  \gabox{IMPORTANT! This does not work since if $\Omega \neq \{\top,\bot\}$ we 
  are not sure that $I^*$ is a valuation.}
  $(\tmod,\me{\_ },\Omega,\ifix)$
  where $\tmod$ is the uniform term algebra defined
  at the beginning of this section
  and  $\ifix$ the valuation $T^\infty(\hoint_\bot)$
  discussed above,
  and where we fix the environment $\varphi$ to be the identity
  and take $\me{\_ }$ (with the identity environment)
  to be the interpretation that maps open \UCTT terms to their normal form.
\end{proof}

\subsection{Why the Semantics is Intuitionistic}
\label{subsec:intuitionistic}
Even if we choose $\Omega$ to be a Boolean Algebra, (such as the two
element algebra) our
completeness theorem is for intuitionistic logic. This is
immediate from the fact that we obtain completeness at all, above.
Nonetheless, it is worth pointing out why a potentially classical
algebra of truth values will not sacrifice the inuitionistic character
of our semantics.

The main reason is that the object of truth values $\Omega$ is
modelling the metatheory, whereas, the program denotations are
modelling the theory, as the following remark should make clear.
One should think of
\[
  I(p,g) = u\in \Omega
\]
as saying something similar to what, in Kripke semantics, would read

\[
  ``p\kforces g'' \mbox{ has truth-value } u.
\]
Just to illustrate, we will look at the term-model
semantics of the classically,
but not intuitionistically valid law:
$(P\hoe Q)\vee (Q\hoe P)$,
using the empty program.

\begin{eqnarray*}
  I(\epr,(P\hoe Q)\vee (Q\hoe P)) &=& I(\epr,P\hoe Q)\vee I(\epr,Q\hoe
  P)\\
  &=& I(P,Q)\vee I(Q, P)
\end{eqnarray*}
If $Q$ and $P$ are distinct atoms, neither
$I(P,Q)$ nor $I(Q, P)$ will, in general,  be true.

Thus the cHa $\Omega$ gives us an extra degree of freedom to interpret
the metatheory of the underlying Kripke-like structure. How this can
be exploited is the subject of the rest of the paper.




\subsection{Enriched Resolution}
\label{sec:enriched-resolution}



We now need
to consider an enriched notion of
resolution, in which arbitrary instantiations of variables are allowed
in a single step. We can think of this as a way of making the converse
of the instantiation lemma (Lemma~\ref{lem:weak-lifting}), called the
\emph{lifting lemma} in the first-order Horn case \cite{lloyd}
which \emph{fails} in RES(t), hold.

If the lifting lemma held, for any state $\ipresgoal{i}{P}{G}$
from which we can resolve to $\NULL$, the following resolution 
would have to exist
\[
\midnullsres{\ipresgoal{i}{P}{X}}{\theta}.
\]
where $X$ is a predicate variable of type $o$.
It is impossible to guarantee the existence of a deduction 
with this particular substitution (it does exist with $\{X/\top\}$)
unless
there happens to be a clause in $\elab(P)$ whose head is $G$, which is not
even allowed if $G$ is non-atomic. But, in a sense, the obstacle seems
more a ``side-effect'' of the existence of higher-order variables than
essential.

\begin{defn}
\label{def:star-theta}
Let $\theta$ be a substitution, $\mfa,\mfb$ state vectors, $P$ a program
and $G$ a goal. Then the {\bf sub} resolution rule is defined by the
following transition:
\[
\vladsub{\hresgoal{\cc}{i}{P}{G}}{\sth}
{\hresgoal{\cc}{i}{P\sth}{G\sth}}
\]
We will say that a sequence of resolution steps possibly using this
additional sub-rule is a
\stto
resolution, and
write
  \begin{equation}
\label{ded:star-theta}
\starmidres{\mfa}{\theta}{\mfb}
\end{equation}
to indicate that such an enriched resolution has taken place with
$\theta$ the composition of all intervening substitutions, restricted
to the variables free in the initial state.
\end{defn}

First we show that *-deduction enriched with the possible use of the
{\bf sub}-rule  is sound for our
semantics, from which we immediately obtain its conservativity over
standard resolution.

\begin{thm}
  If there is a deduction
$\starmidres{\mfa}{\theta}{\mfb}$ then
in any uniform algebra $\sem{\mfb}{\eta}\oleq\sem{\mfa\theta}{\eta}$.
\end{thm}
The proof, by induction on the length of \stto-deductions, proceeds
identically to the preceding soundness proof for
$\bleadsto$-resolution. The only new case to consider is a proof that
starts with a {\bf sub}-step, in which case we have to show that this
single step is sound. But if $\mfa\derstep{\theta}\mfa\theta$ then
the conclusion of the soundness theorem, namely
 $\sem{\mfb}{\eta}\oleq\sem{\mfa\theta}{\eta}$,
reduces, in this case, where
$\mfb$ is $\mfa\theta$,
to $\sem{\mfa\theta}{\eta}\oleq\sem{\mfa\theta}{\eta}$,
which is trivially true.

\begin{cor}
\stto-deduction is conservative, for closed program-goal pairs
 over $\bleadsto$. If there is a deduction
  $\starnullsres{\mfa}{id}$, then there is a deduction
$\midnullsres{\mfa}{id}$.
\end{cor}

\begin{proof}
If there is a deduction   $\starnullsres{\mfa}{id}$, then by
soundness, $\top_{\Omega}\oleq\sem{\mfa}{\eta}$ in any uniform
algebra, with any environment. By the completeness theorem, there is a
proof $\midnullsres{\mfa}{id}$.
\end{proof}

The enriched \stto-deduction trivially satisfies what was called a
``lifting'' lemma in the discussion following the proof of
Lemma~\ref{lem:weak-lifting}: if there is a resolution
\[
\bd: \starmidres{\mfa\theta}{\varphi}{\mfb}
\]
then the substitution $\theta$ can be lifted
\[
\starmidres{\mfa}{\theta\varphi}{\mfb} \enspace .
\]
Just consider a variant $\bd'$ of $\bd$ which does not contain any variable 
in $\fv{\theta} \setminus \fv{\mfa \theta}$, and precede its first step with one {\bf sub}-step
$\mfa\stto^\theta\mfa\theta$. Thus, although \bleadsto\ does not
satisfy lifting, it can be conservatively added.

With the {\bf sub}-rule added to resolution, we are obviously also
able to restrict the {\bf instance} rule to replacement of
existentially quantified variables by fresh logic variables:
\[
\vladres{\hresgoal{\cc}{i}{P}{\exists x.G}}{\exists}{
\hresgoal{\cc}{i}{P}{G[z/x]}}
\]
with $z$ fresh and of level $i$ or less
since subsequent instantiation of $z$ can then be
carried out using the {\bf sub} rule.

\section{Some alternative models}

\subsection{Completeness: Substitution Semantics}

Inspection of the soundness and completeness theorems show that we
may allow $\Omega$ to be any complete Heyting Algebra  in which
the top element is join-irreducible. Furthermore, the bottom-up
\gabox{Not sure about this, at least with the current proof.}
(fixed point) characterization
of the least term-model used in the completeness theorem still holds.

We will now investigate how the as yet untapped degree of freedom
we have in choosing the object of truth values $\Omega$ can be
exploited to give more operationally sensitive information about
programs.
%

\begin{defn}
  Define the preorder $\leq_s$ on {\it safe\/} (i.e. idempotent) substitutions as follows.
\[
\theta\sleq\varphi \mbox{ iff } \exists\gamma(\varphi = \theta\gamma).
\]
We will also define the partial order $\theta\sqq\varphi$ to mean
that $\varphi$ extends
$\theta$ to a possibly larger domain, i.e. $\varphi$ restricted to the
variables in the domain of $\theta$ is identical to $\theta$.
\end{defn}

 Note that $\theta\sqq\varphi$
is a special case of $\theta\sleq\varphi$ for {\em safe\/}
substitutions,
since what we mean by the {\em domain\/} of a substitution is the set of
variables $x$ for which $x\neq \theta(x)$. In conventional
mathematical terms, its
true domain is the entire set of variables, so if  $\theta\sqq\varphi$
then letting $\gamma$ be the substitution
\[
\gamma(x) = \begin{cases}
  \varphi(x) & \mbox{ if $x \in \dom(\varphi) \setminus \dom(\theta)$}\\
  x & \mbox{otherwise,}
\end{cases}
\]
we have $\varphi = \theta\gamma$. This definition is consistent precisely
because $\varphi$ and $\theta$ are assumed safe. Any variable $x$ among
$\fv{t}$ for some $t$ in the range of $\theta$
is also a variable in the range of $\varphi$ by definition of
 $\theta\sqq\varphi$, hence not in the domain of $\varphi$.

Also $\sqq$ is a partial order, but $\sleq$ is not since
$\theta\sleq\varphi \wedge \varphi\sleq\theta$ implies that
$\theta\sim\varphi$ but not necessarily that $\theta=\varphi$
\begin{defn}
  Let $\sos$ be the set of all upward $\sleq$-closed sets of safe
  substitutions. That is to say $U\in\sos$ iff for every $u\in U$ if
$u\sleq v$ then $v\in U$.
\end{defn}
It is easy to see that $\sos$ has a topology (the so-called Alexandroff
\gabox{Wht care about the topology of $\sos$?}
Topology induced by a preorder), that is to say, the set $\tops$ of
all safe substitutions and the empty set are in $\sos$, and $\sos$ is
closed under unions and finite (in fact, arbitrary) intersections.
We also have that $\tops$ is join-irreducible, since $U=\tops$ iff
the identity substitution $id$ is in $U$. Thus,
\[
\tops\ssq V\cup V' \quad \mbox{iff} \quad id\in V\cup V'
 \quad \mbox{iff}  \quad id\in V \quad \mbox{or}\quad id\in V'
 \quad \mbox{iff} \quad \tops= V \quad \mbox{or}\quad \tops= V'.
\]
Thus, $(\sos,\ssq)$ is a cHa, and is a suitable object of truth
values for a uniform algebra. We now define one such algebra, the
substitution-algebra, over $\sos$.

\begin{defn}
  \label{def:cats}
  Let $\cat{S}=\ang{\sfd,\mee,\sos,\shoint}$ be the uniform algebra
  defined as follows. Take $(\sfd,\mee)$ to be the underlying structure
of the term-model, i.e.\@ the same uniform applicative structure used
  in the   completeness theorem, above. Define
\begin{equation}
\label{eq:is}
({\shoint})_i(P,G) = \{\theta: \theta \mbox{ is safe, of level $i$ or less, and }
\nullsres{\ipresgoal{i}{P\theta}{G\theta}}{id} \},
\end{equation}
where the domain of $\theta$ is {\bf not necessarily restricted to the
  variables free in $P,G$. }
\end{defn}

Note that the condition on the substitutions $\theta$  mentioned in
Definition~\ref{def:cats}, namely that they are not restricted to
the program and goal pairs, is essential here, otherwise the r.h.s.\@ of
\eqref{eq:is} would not be upward closed. It is easily seen so, since
for any substitution $\gamma$, if
$\nullsres{\ipresgoal{i}{P\theta}{G\theta}}{id}$ then
$\nullsres{\ipresgoal{i}{P\theta\gamma}{G\theta}}{id}$ by
Corollary~\ref{cor:sub-theta}.

\begin{lem}
  The indexed family of mappings $\shoint$ defined in \eqref{eq:is}
  makes\, $\cat{S}$ into a uniform algebra.
  \gabox{It would be nice if the fixpoint construction could build
    this valuation. JL: should we say something about this?
  }
\end{lem}

\begin{proof}
  We must check that $\shoint$ satisfies all the conditions of
Definition~\ref{interp}, above.

\case{condition \eqref{uaatom}}
If $A_r$ is a rigid atom,  we have the following derivation
\[
  \ipresgoal{i}{A_r}{A_r} \derstep{\ids} \ipresgoal{i}{A_r}{\top} \derstep{null} \NULL
\]
Therefore $\ids\in\sthoint{i}{A_r}{A_r}$, and hence  $\sthoint{i}{A_r}{A_r}=\tops$.

\case{condition \eqref{umonotone}}
If $\theta \in \sthoint{i}{P}{G}$ then there is a derivation 
\[
\nullsres{\ipresgoal{i}{P\theta}{G\theta}}{\ids} \enspace .
\]
By Corollary~\ref{cor:level-increase} there is also a derivation
\[
\nullsres{\ipresgoal{i+1}{P\theta}{G\theta}}{\ids} \enspace ,
\]
hence $\theta \in \sthoint{i+1}{P}{G}$.



\case{condition \eqref{uatopgoal}}
It states that $\thoint{i}{P}{\top} = \top_\Omega$, which is immediate.

\case{condition \eqref{uaand}}
Suppose $\theta$ is in $\sthoint{i}{P}{G_1\wedge G_2}$. Then
there is a resolution
\[
\midnullsrest{\ipresgoal{i}{P\theta}{G_1\theta \wedge G_2\theta}}
{id}
{\ipresgoal{i}{P\theta}{G_1\theta}\otimes
\ipresgoal{i}{P\theta}{G_2\theta}}
{id}
\]
By the product lemma (Lemma~\ref{lem:product})
this condition is equivalent to the existence of resolutions
$\midnullsres{\ipresgoal{i}{P\theta}{G_1\theta}}{\ids}$
and
$\midnullsres{\ipresgoal{i}{P\theta}{G_2\theta}}{\ids}$,
or equivalently,
 $\theta\in\sthoint{i}{P}{G_1}\cap \sthoint{i}{P}{G_2}$.

\case{condition \eqref{uaor} and \eqref{uahorseshoe}}
These cases are immediate, since the existence of an identity derivation
of $\NULL$
from $\ipresgoal{i}{P\theta}{G_1\theta\vee G_2\theta}$ is equivalent to
the existence of  such a derivation from
 $\ipresgoal{i}{P\theta}{G_i\theta}$ for some $i\in\{1,2\}$,
and the existence of an identity derivation of $\NULL$
from $\ipresgoal{i}{P\theta}{Q\theta\hoe G\theta}$ to one from
$\ipresgoal{i}{P\theta\wedge Q\theta}{G\theta}$.\\

\case{condition \eqref{uainstance}}
Suppose $\theta \in \sthoint{i}{P}{\exists xG}$,
where $x$ is of type $\alpha$. This is equivalent to the existence of a flat deduction
\[
\midnullsrest{\ipresgoal{i}{P\theta}{\exists x(G\theta)}}
{\exists}{\ipresgoal{i}{P\theta}{G\theta[t/x]}}{[\ids]}{} \enspace .
\]
By the specialization lemma (Lemma~\ref{lem:sub-t}), there is also a derivation of 
\[
  \ipresgoal{i}{P\theta}{G\theta[t/x]\theta} \derivation{[\ids]} \NULL \enspace ,
\]
and since $G\theta[t/x]\theta = G[t/x] \theta$, we have $\theta \in \sthoint{i}{P}{G}$.
Hence $\theta\in\bigvee_{d\in\dee_{i,\alpha}}
\sthoint{i}{P}{G[d/x]}$, which, in the term structure is equivalent to
the right hand side of condition \eqref{uainstance}.

Conversely, if $\theta\in\bigvee_{d\in\dee_{i,\alpha}}
\sthoint{i}{P}{G[d/x]}$, then for some term $t$
\[
\nullsres{\ipresgoal{i}{P\theta}{G[t/x]\theta}}{\ids}.
\]
Since $(G[t/x]\theta) = G\theta[t\theta/x]$, we have
\[
\midnullsrest{\ipresgoal{i}{P\theta}{\exists x(G\theta)}}
{\exists}{\ipresgoal{i}{P\theta}{G\theta[t\theta/x]}}{\ids}{} \enspace,
\]
and hence that
 $\theta\in\sthoint{i}{P}{\exists xG}$.

\case{condition \eqref{uageneric}}
Suppose $\theta$ is in $\sthoint{i}{P}{\forall xG}$,
with $x$ of type $\alpha$.
Then there is a derivation
\[
\nullsres{\ipresgoal{i}{P\theta}{(\forall xG)\theta}}{\ids}.
\]
Since no bound variable is in the domain of $\theta$,
the goal $(\forall xG)\theta$ is equivalent to
$\forall x(G\theta)$, and the rest of the derivation must be of the form
\[
\midnullsrest{\ipresgoal{i}{P\theta}{\forall x(G\theta)}}{\forall}
{\ipresgoal{i+1}{P\theta}{G\theta[c/x]}}{\ids}.
\]
By the (first claim) of
the  corollary on generic constants (Corollary~\ref{cor:genconst}), for
any positive term $t$, there is a derivation
\[
\nullsres{\ipresgoal{i}{P\theta}{G\theta[t/x]}}{\ids},
\]
and hence, for any positive term $u$, letting $t=u\theta$,
\[
\nullsres{\ipresgoal{i}{P\theta}{G[u/x]\theta}}{\ids}.
\]
Therefore,
$\theta\in \bigcap_{u\in\dee_{i,\alpha}}\sthoint{i}{P}{G[u/x]}$.

Conversely, suppose
$\theta\in \bigcap_{u\in\dee_{i,\alpha}}\sthoint{i}{P}{G[u/x]}$. In particular,
taking as $u$ a fresh constant $c$ of level $i$ and type $\alpha$, we have
\[
  \ipresgoal{i}{P\theta}{G\theta[c/x]} \derivation{\ids} \NULL \enspace ,
\]
By looking at the proof of the second part of the corollary on generic constants
(Corollary~\ref{cor:genconst}), it is clear that this is enough to ensure the existence of
a derivation
\[
  \ipresgoal{i+1}{P\theta}{G\theta[c'/x]} \derivation{\ids} \NULL \enspace ,
\]
where $c'$ is a fresh contant of level $i+1$, hence
\[
\midnullsrest{\ipresgoal{i}{P\theta}{(\forall x.G)\theta}}{\forall}
{\ipresgoal{i+1}{P\theta}{G\theta[c'/x]}}
{\ids},
\]
and $\theta \in \sthoint{i}{P}{\forall xG}$, as we wanted to show.

\case{condition \eqref{uainclusion}}
Assume $\theta \in \sthoint{i}{P}{G}$ and $P \sqq_\pee P'$. Then, there is a derivation
\[
  \ipresgoal{i}{P\theta}{G\theta} \derivation{\ids} \NULL \enspace .
\]
Since $P \sqq_\pee P'$, we have $\elab(P) \subseteq\elab(P')$. It is
easy to check that, under this condition, we have a derivation
\gabox{CHECK the proof with all details ?}
\[
  \ipresgoal{i}{P'\theta}{G\theta} \derivation{\ids} \NULL \enspace .
\]
hence $\theta \in \sthoint{i}{P'}{G}$.

\case{condition \eqref{uabackchain}}
Let $\theta \in \sthoint{i}{P}{G}$ and $G \hoe A_r \sqq_\pee P$. Then, there is a derivation
\[
  \ipresgoal{i}{P\theta}{G\theta} \derivation{\ids} \NULL \enspace .
\]
and $[i; G \hoe A_r] \subseteq [i; P]$. Since $[i; G \hoe A_r] = \{ G \pmi A_r\}$, this means there is a clause $\forall x (G' \hoe A'_r) \in \elab(P)$ and $\gamma' =\{t_1/x_1, \ldots, t_n/x_n \}$ such that $G'\gamma' = G$ and $A'_r\gamma'=A_r$. Therefore, by considering the renaming-apart substitution $\delta =\{w_1/x_1, \ldots, w_n/x_n\}$, we have that $\gamma =\{t_1/w_1, \ldots, t_n/w_n\}\theta$ is an unifier for $A_r'\delta$ and $A\theta$ which is the identity over the variables of $A\theta$, since $A_r \delta \gamma  = A_r \gamma' \theta = A_r\theta$ and $A_r \theta \theta = A_r \theta$ since $\theta$ is safe. Threfore, we have a derivation:
\[
  \ipresgoal{i}{P\theta}{A_r \theta} \derivation{\gamma} \ipresgoal{i}{P\theta}{G\theta}\derivation{\ids} \NULL \enspace ,
\]
whence $\theta \in  \sthoint{i}{P}{A_r}$.

\end{proof}



\subsection{Lindenbaum Algebra and \UCTT}

We will  now weaken the definition of Uniform Algebras to a class of
\gabox{I don't understand... the definition of UA in Section 3 already is quite liberal
in the meet and join we need.}
structures maintaining soundness, using an
idea in \cite{TvD}. It has long been known that it is sufficient
for a complete semantics to stipulate only that the object of
truth-values $\Omega$ contain {\em enough\/}  infinitary meets and
joins to interpret quantifiers and implication. In particular, the
{\em parametrized\/} family $\{I_i(P,G)_{\eta[x:=d]} \mid d\in D_{i,\alpha}\}$
must have a meet and join, and the $\wedge\bigvee$ distributive law
must hold for {\em existing\/} joins. Furthermore, we will restrict
meets and joins to the interpretation of different goals \emph{with
  the same ambient program $P$}. Letting \bbp\ stand for the partially
ordered set of programs, we will define a model \emph{indexed over}
\bbp, as explained below.

Then we will show how to build a (global) uniform algebra out of \bbp-indexed
and parametrized algebras, using an \emph{ideal completion} to build
the algebra, a construction with a long history in lattice theory and logic.
See \cite{TvD} for details.

\subsubsection{A \bbp-indexed Lindenbaum Algebra for \UCTT}

We are now going to give constructive completeness theorem for
$\UCTT$ which uses a Lindenbaum-style algebra construction. In essence
program-goal pairs are interpreted as ``themselves'' in the object of
truth-values.  More precisely, we will map states $\pgstate{i}{P}{G}$ to
equivalence classes $\braseq{P}{G}$ of sequents in \ICTT, with respect
to an equivalence relation to be defined below. We first fix our
notation conventions. We will write an
{\em indexed\/}
sequent $\seqi{P}{G}$ with
natural number index $i$ to indicate that the formulas in $P,G$ are
over universe $\univ_i$. Thus, if the notation $\seqi{P}{G}$ is
correct with respect to levels, so does $\seqx{k}{P}{G}$ for any $k>i$.
Indices may be omitted if clear from context, or when it is clear that
they are ranging over all possible natural number values, as in the
definition below.
\begin{defn}
  We define the binary relation $\leqt$ on indexed \ICTT
sequents as follows:
 $\seq{P}{G} \leqt \seq{P'}{G'}$ if and only if
in some metatheory $\ct$, presumed constructive, we can prove that if the sequent
$\seq{P}{G}$ has a {\em cut-free\/}
 proof in \ICTT then so does $\seq{P'}{G'}$.
\gabox{Note that this definition does not use indexes, since there is no reason
to use them. And I think there is no reason even later.}
We also describe
this state of affairs by
$\ct\btraca\: \cfseq{P}{G}\wimp \cfseq{P'}{G'}$, omitting mention of \ct
when clear from context.
We will also use the notation $\cfseq{P}{G}$ (with or without
numerical index) to abbreviate the predicate
``there is a cut-free proof of $\seq{P}{G}$.''

We denote by $\eqt$ the
equivalence relation induced by this order, i.e.
\[
\seqi{P}{G}\eqt \seqx{k}{P'}{G'} \quad \mbox{ iff } \quad
\seqi{P}{G}\leqt \seqx{k}{P'}{G'} \quad \mbox{ and } \quad
\seqx{k}{P'}{G'}\leqt \seqi{P}{G},
\]
and write the induced equivalence classes thus:
$\braseq{P}{G}=\{\seqx{k}{P'}{G'}: \seqi{P}{G} \eqt \seqx{k}{P'}{G'}
\}.$
\end{defn}
Cut-elimination holds in \ICTT \cite{ccct05}, but we do not require
it to be provable in \ct in the results that follow.

The reader should note that it is possible for a sequent indexed by
$i$ to imply sequents with other indices. For example, any sequent
lies below the identity $\seqi{A}{A}$ for any $i$ and, say, any atom
$A$ over $\univ_i$.
It is easy to see that
$\leqt$ is a preorder, $\eqt$ an equivalence relation and
 $\braseq{P}{G}$ is well-defined. We will abuse language and use the same
symbol for the preorder $\leqt$ and the induced order on equivalence
classes.

The reader should note that 

\begin{lem}
  \label{lem:supinf}
  Let $\oml$ be the poset of equivalence classes $\braseq{P}{G}$
with order induced by $\leqt$.
Then the following suprema and infima exist in
$\oml$:
\begin{itemize}
  \item $\braseq{P}{G_1}\wedget\braseq{P}{G_2} = \braseq{P}{G_1\wedge
  G_2}$
  \item $\braseq{P}{G_1}\veet\braseq{P}{G_2} = \braseq{P}{G_1\vee
  G_2}$
  \item $\bigvee_{t\in \dee_{i,\alpha}} \braseq{P}{G[t/x_\alpha]} = \braseq{P}{\exists
  x_{\alpha} G}$
 \item $\bigwedge_{t\in \dee_{i,\alpha}} \braseq{P}{G[t/x_\alpha]} =
\braseq{P}{\forall  x_{\alpha} G}$
\end{itemize}
The operations are indexed over a fixed program $P$. More formally, we
can think of $\oml$ as a mapping from programs $P$ in \bbp\ to the set
of all (equivalence classes of)  pairs $\braseq{P}{G}$ for goals
$G$. The $\wedge$ and $\vee$ operations are defined on each of these
sets, but not for pairs with different programs.
\end{lem}
The choice of \ct matters in this lemma. It must be strong enough to
support the arguments below. A key point here is that these
arguments are clearly constructive. It will suffice, for our purposes
to take \ct to be the constructive set-theory IZF
\cite{beeson,mccarty86} 
although a considerably weaker
theory could be used.

\begin{proof}
  For the first claim, observe that
 $\cfseq{P}{G_1\wedge G_2}\wimp \cfseq{P}{G_1} \mbox{ and }
  \cfseq{P}{G_2}$, by induction
  on the depth of the given proof. If the depth is 0, then
 $\seq{P}{G_1\wedge G_2}$ was derived by the identity axiom, i.e.
$P=P',G_1\wedge G_2$, where $G_1,G_2$ must be supposed {\em core\/}
formulas, i.e. both program and goal formulas.
 But then, applying $\wedge_L$ to
the axiom $\seq{P',G_1,G_2}{G_i}$, we have
$\seq{P',G_1\wedge G_2}{G_i}$. 

The inductive cases are easy.
Consider the last rule used in a cut-free proof of
 $\seq{P}{G_1\wedge G_2}$.
Because of  the program-goal structure of $P,G$ the
only possibilities are $\wedge_R$, which would give us the desired
 result immediately, or the $\forall_L,\hoe_L$ and $\wedge_L$ rules.
In the latter cases, we apply the inductive hypothesis to the
 relevant premiss to obtain the desired proof.
From this we have that
any sequent below  $\seq{P}{G_1\wedge G_2}$ is below
the infimum of the other two.

%
Now suppose $\cfseq{Q}{G}
\wimp \cfseq{P}{G_1} \mbox{ and }
\cfseq{P}{G_2}$. Then
 $\cfseq{Q}{G} \wimp \cfseq{P}{G_1\wedge G_2}$, so
$\braseq{P}{G_1\wedge G_2} \eqt \braseq{P}{G_1}\wedget\braseq{P}{G_2}$.

One direction of the second claim follows from $\vee_R$:
$\cfseq{P}{G_i}\wimp \cfseq{P}{G_1\vee G_2}$. To show
the other direction we must show
\[
\cfseq{P}{G_1\vee G_2} \wimp \cfseq{P}{G_1} \mbox{ or }
\cfseq{P}{G_2}
\]
making use of the fact
that $P$ is a program formula. Thus the sequent
$\seq{P}{G_1\vee G_2}$ could not have been derived by the identity
axiom, since disjunctions are not allowed on the left.
 If the last rule used was $\vee_R$ we are done. If
it was one of the three allowed left rules, we proceed as in the
preceding case, applying the induction hypothesis to the relevant
premiss(es).

For the third claim we use $\exists_R$ for one direction, and then
establish the existence property for cut-free proofs of
program formulas \gabox{Note we are using here $\ICTT$ and not $\ICTT^c$.}
\[
\cfseq{P}{\exists x_{\alpha} G} \wimp 
 \cfseq{P}{G[t/x_\alpha]} \mbox{ for some } t \in \dee_{i,\alpha}
\]
As above, we induct on the depth of the given cut-free proof. The
depth cannot be 0 (or if the reader prefers to say it this way, the
base case is vacuously true) since an existential formula cannot be a
program formula. Suppose the claim holds for all shorter cut-free
proofs of an existential goal from a program formula. If the last rule
used was $\exists_R$ we are done. Otherwise one of the three allowed
left rules was used, and applying the induction hypothesis to the
relevant premiss(es) we are done.

This result {\em almost\/}
gives the other direction of the third claim.
 Since $t_\alpha\in\dee_{i,\alpha}$ means
that $t$ is a positive formula, we need to appeal to Theorem 3.2.14 of
Nadathur's dissertation and the subsequent Miller \emph{et al.}
\cite{nadathur:thesis,uniform} to show that if a cut-free proof of
$\seq{P}{\exists x_{\alpha} G}$ exists, then there must be a cut-free
proof for a {\em positive\/} witness.

Finally, the fourth claim follows from the fact that if
$\braseq{Q}{G} \leqt \braseq{P}{G[t/x]}$ for each $t\in \dee_{i,\alpha}$,
then $\cfseq{Q}{G}$ implies the existence of a cut-free proof of
every positive instance, in particular, for $t$ a fresh
variable, whence, by $\forall_R$, the derivability of
$\cfseq{P}{\forall x_\alpha G}$. For the other direction
argue by induction as in the preceding cases.
If a cut-free proof of $\seq{P}{\forall  x_{\alpha} G}$ has depth 0,
then $P=P',\forall x G$, and $G$ is both a program and goal formula.
 But then we can produce the desired cut-free
proof for any positive instance $t$
using the identity axiom and $\forall_L$:
\[
\infer[\forall_L]
{\seq{P',\forall x G}{G[t/x]}}
{\seq{P',G[t/x]}{G[t/x]}}
\]
For the inductive case, suppose the last rule used was $\forall_R$.
Then $x$ does not occur freely in the final sequent, and
there is a cut-free proof of its immediate premises, of the form
$\seq{P}{G}$. It is a straightforward exercise to show, by induction
on the length of this proof, that a cut-free proof of
  $\seq{P}{G}[t/x]$ must also exist.
If the last rule used is one of the three allowed left rules, the
argument proceeds as above.
\end{proof}

\begin{lem}[Generalized contraction]
  Suppose $Q \in [i;P]$ for some $i$. If $\cfseq{P,Q}{G}$ then
  $\cfseq{P}{G}$.
 \end{lem}
\begin{proof}
  By definition of extension there is a program formula $D \in P$ such that $Q \in [i;D]$. We consider a generalization of the proof where $Q$ is any formula derivable from $D$ through $\to_i$, and we prove the result by induction on the length $n$ of the derivation.

  If  $n = 0$, then $Q = D = G \hoe A$, and the resut follows from contraction in \ICTT. Otherwise, consider the last step in the derivation of $Q$ through the transition system in Lemma~\ref{lem:extension-ts}.

\case{$Q_1 \wedge Q_2 \to_i Q_k$}
Then the derivation is of the following form: $D\to_i^{n-1}\, Q_1\wedge
Q_2\imp Q_k$. Assume without loss of generality that $k=1$. By the induction hypothesis, if $\cfseq{P,Q_1\wedge Q_2}{G}$ then $\cfseq{P}{G}$. Now suppose $\cfseq{P,Q_1}{G}$. Then
$\cfseq{P,Q_1,Q_2}{G}$ by weakening-left (an easily derived rule of \ICTT),
so $\cfseq{P,Q_1\wedge Q_2}{G}$ by $\wedge_L$, hence  $\cfseq{P}{G}$.


\case{$\forall x Q' \to_i Q'[t/x]$}
By the induction hypothesis
$\cfseq{P,\forall x Q'}{G}$ implies $\cfseq{P}{G}$. Now suppose
$\cfseq{P,Q'[t/x]}{G}$ and use $\forall_L$.
\end{proof}
Now we are able to prove the following lemma.
\begin{lem}
  \label{lem:cut25}
  If $G\hoe R\in [i;P]$ and $\cfseq{P}{G}$ then $\cfseq{P}{R}$.
\end{lem}
\begin{proof}
We have that  $\cfseq{P}{G}$, and trivially $\cfseq{P,R}{R}$. By arrow-left
$\cfseq{P,G\hoe R}{R}$. By the preceding lemma $\cfseq{P}{R}$
  \end{proof}
Now we use $\oml$ as an object of truth-values for a parametrically complete
uniform algebra over the term model.
\begin{defn}
  Let $\cat{L}=\ang{\sfd,\umod,\oml,\cl}$ be the parametrically
complete uniform algebra
  defined as follows. Take $(\sfd,\umod)$ to be the underlying structure
of the term-model,    i.e. the same uniform applicative structure used
  in the  first completeness theorem, above.
Define the indexed family of mappings
$\cl_i:\dee_{i,\pee}\times\dee_{i,\gee}\imp \oml$ as follows.
\begin{equation}
\label{eq:L}
\cl_i(P,G) = \braseq{P}{G}
\end{equation}
\end{defn}

\begin{thm}
  \label{thm:lindenbaum}
  The mappings $\cl_i$ make $\cat{L}$ into a p.c.\@ uniform algebra.
\end{thm}

\begin{proof}
  We check all the conditions of Definition~\ref{interp}, above.
  
\case{condition \eqref{uaatom}}
Let $R$ be a rigid atom of level $i$.
Then $\cl_i(R,R)=\braseq{R}{R}=\top_{\oml}$.

\case{condition \eqref{umonotone}}
Since $\cfseqi{P}{G}\wimp \cfseqx{i+1}{P}{G}$ this case is immediate.

\case{condition \eqref{uatopgoal}}
It states that $\cl_i(P,\top)=\braseq{P}{\top}=\top_{\oml}$ and it is immediate.

\case{condition \eqref{uaand}}
$\cl_i(P,G_1\wedge G_2) = \braseq{P}{G_1\wedge G_2} =
\braseq{P}{G_1}\wedget \braseq{P}{G_2}$ by the preceding lemma.

\case{conditions \eqref{uaor}, \eqref{uainstance}, \eqref{uageneric}}
Similar to the previous case. 

\case{condition \eqref{uahorseshoe}}
It is immediate, by $\hoe_R$ in one direction and cut-free invertibility of $\hoe_R$ by a
straightforward induction on the length of proofs in the other.

\case{condition \eqref{uainclusion}}
We have to show that $P\sqq P'$ implies $\cl_i(P,G) \leqt \cl_i(P',G)$. We can prove that $P \sqq P'$ implies $\elab(P) \subseteq \elab(P')$. Now use Lemma~\ref{lem:clause-deriv}.
\gabox{All this proof is quite sketchy.}

\case{condition \eqref{uabackchain}}
We need to prove that
\[
(G \horseshoe_{\pee\gee\ree} R) \sqq_\pee P \textrm{ implies }
     \linden{i}{P}{G} \leq_\Omega \linden{i}{P}
     {\inc_{\gee\aee}(\inc_{\aee\ree}(R))}.
\]
Note that
$(G \horseshoe_{\pee\gee\ree} R) \sqq_\pee P$  implies that
$(G \horseshoe_{\pee\gee\ree} R)\in [i; P]$, so by Lemma~\ref{lem:cut25}
if $\cfseq{P}{G}$ then $\cfseq{P}{R}$.
\end{proof}

%

\subsubsection{The Algebra of Ideals}
We now define a complete Heyting algebra $\widehat{\Omega}$ as follows. The elements
are sets of members $\brseq{P}{G}$ of $\oml$ ranging over all programs
$P$ and  goals $G$ \emph{downward closed under} $\leqt$. We define
$\bigvee$ and $\bigwedge$   by union and intersection. We call these
sets $S$ \emph{ideals}. They are
also required to satisfy the additional condition
\[
  X\subseteq S \mbox{ and } \bigvee X \in \oml \quad \wimp \quad \bigvee X \in S
\]


The ideals together with their members form a topology, in fact an Alexandroff topology (closed under arbitrary intersections) as is well known, so they are also a Heyting Algebra \omfib, which is what we require.

Define, for each $\braseq{P}{G}$ in $\oml$, the induced \emph{principal ideal} $\idseq{P}{G}_i$ to be the set
\[
  \{\braseq{Q}{H}\, :\, \braseq{Q}{H} \leqt \braseq{P}{G}\}
  \]

\begin{defn}
  Let \omfib\ be as defined above. We now define an interpretation
  \[
    \fri:\deep\times\deegg\imp \omfib
  \]
  tacitly indexed over the signatures $i$, as follows
  \[
    \fri(P,G) = \idseq{P}{G}
    \]
\end{defn}
Now the results of Lemma~\ref{lem:supinf} and Theorem~\ref{thm:lindenbaum} easily give us the following theorem.

\begin{thm}
\fri\ and \omfib\ define a uniform algebra, not just a \bbp-indexed model.
\gabox{We sometime use the term $\mathbb{P}$-indexed model, somtime p.\. uniform algebra. Choose one.}
\end{thm}
\begin{proof}
  First observe that $\fri_i(R,R)$ and $\fri_i(P,\top)$ each equal the
  set of all pair classes $\braseq{P}{G}$, which is $\top_{\omfib}$,
  since $R\cfpf R$ and $P\cfpf\top$.
  We check a few of the remaining conditions
  (\ref{uaand}),(\ref{uaor}),(\ref{uahorseshoe}),(\ref{uainstance}),(\ref{uageneric}).

  \case{condition \eqref{uaand}}
 $ \idseq{P}{G\wedge G'} = \idseq{P}{G} \cap \idseq{P}{G}$ since
 if $\braseq{Q}{H} \in  \idseq{P}{G\wedge G'} $ then  $\braseq{Q}{H} \leqt  \braseq{P}{G\wedge G'}$, hence  $\braseq{Q}{H} \leqt \braseq{P}{G}$ and
 $\braseq{Q}{H} \leqt \braseq{P}{G}$ by Lemma~\ref{lem:supinf}, therefore
 $\braseq{Q}{H} \in  \idseq{P}{G}\cap \idseq{P}{G'} $. The converse follows similarly by the same lemma.

 \case{conditions~(\ref{uaor}),(\ref{uainstance}),(\ref{uageneric})}
 The arguments are similar and left to the reader.

 \case{condition \eqref{uahorseshoe}}
 From Theorem~\ref{thm:lindenbaum} we know that $\braseq{P}{D \hoe G} = \braseq{P \wedge D}{G}$ for each $i$. Therefore, also $\fri_i(P,D \hoe G) = \fri_i(P \wedge D,G)$.
 
 \case{conditions \eqref{uainclusion} and \eqref{uabackchain}}
 Similar to the previous one.
\end{proof}

We can now give a second proof of completeness, using the results of
\cite{uniform} and the lemmas just established.
\begin{thm}
  If the program-goal pair $P$, $G$ is mapped to $\top$ in all uniform algebras, 
  then the sequent  $P \vdash G$ is derivable by a uniform proof.  
\end{thm}
\begin{proof}
\gabox{I don't understand the point of this theorem, since we need to appeal to the results of Millet et al.}
If $P,G$ is sent to $\top_{\omfib}$ then $P\vdash G$ has a cut-free
proof. By the results of Miller et. al., it has a uniform proof.
\end{proof}
\paragraph*{Semantic proof of uniformity}
A considerable reworking of the tools developed in this section is
required to show, directly in the semantics, that validity in the
algebra of ideals just developed
implies a uniform proof. It amounts to showing that the
implication operator in a Heyting Algebra interpretation coincides
with the index-shifting treatment of implication in goals in uniform
algebras, which is beyond the scope of this paper.
Such a \emph{semantic proof of uniformity}
is given, in the setting of Kripke models and $\lambda$Prolog with
constraints, in \cite{lipton-nieva-tcs}. We expect to explore the
Lindenbaum-style model further in future work.

\subsection{Indexed Models and Fibrations}

For the reader familiar with category theory, it should be clear that
we have been working throughout the paper with certain categories indexed over
a category \bbp\ of programs. Because of the operational
interpretation given to implications in goals, it is usually
sufficient to take these categories to be \clat, a category
of parametrically complete lattices, indexed over a category of
programs, whose objects are interpreted
goals or state vectors with
existing $\bigvee,\bigwedge,\top,\bot$ preserving
maps. 


For this treatment of both the syntax and
semantics of Horn and Hereditarily Harrop logic programming, see,
e.g., \cite{asdp09}. Some other categorical treatments are
\cite{powerK96,kp16,kp2011, ffl}.



\newpage

\appendix

\section{Appendix}

\subsection{The Variable Indexing Problem}
\label{sec:varindex}

Consider the following incorrect ``deduction'' of
\[
\forall x. Qxx \horseshoe \exists y\forall z.Qyz
\]
from the empty program.
\[
\begin{array}{ll}
 & \?  \forall x. Qxx \horseshoe \exists y\forall z.Qyz \\
\stackrel{\horseshoe}{\bleadsto} & \forall x.Qxx\ \?
\exists y\forall z. Qyz\\
\stackrel{\exists}{\bleadsto}&\forall x.Qxx\  \? \forall z.Qyz\\
\stackrel{\forall}{\bleadsto}&\forall x.Qxx\ \? Qyc\\
\stackrel{[c/y]}{\bleadsto}& \NULL
\end{array}
\]
versus the following failed deduction in which a counter keeps track
of changing signatures, and the (indexed) logic variable
$y_0$ created when signature 0 was the ambient signature, is inhibited
from unification with constant $c_1$ belonging to signature 1.

\[
\begin{array}{lll}
 & 0& \emptyset \ \?  \forall x. Qxx \horseshoe \exists y\forall z.Qyz \\
\stackrel{\horseshoe}{\bleadsto}&0&\forall x.Qxx\ \? \exists y\forall z. Qyz\\
\stackrel{\exists}{\bleadsto}&0&\forall x.Qxx\ \? \forall z.Qy_0z\\
\stackrel{\forall}{\bleadsto}&1&\forall x.Qxx\ \? Qy_0c_1\\
 & & FAIL
\end{array}
\]

  \subsection{Distributive Lattices and Heyting Algebras}
  \gaboxcut{Is this section really needed ? The paper only marginally deals with HAs. In any case,
  anyone reading a part on completenss in intuitionistic logic should know what a HA is.} 

In the interest of making this paper more self-contained we
briefly recapitulate some elementary definitions.

\begin{defn}A lattice is a partially ordered (nonempty) set with two
operations $\land$ and $\lor$, the meet and join, satisfying the
following axioms for all members $a$, $b$ and $c$ of the set:
\begin{gather}
\label{haor}a \leq a \lor b \ \ \ \ \  b \leq a \lor b\\
\label{haand}a \land b \leq a \ \ \ \ \ a \land b \leq b\\
\label{haandor}a \leq c \textrm{ and } b \leq c \textrm{ implies } a
\lor b \leq c\\
\label{haandand}c \leq a \textrm{ and } c \leq b \textrm{ implies } c
\leq a \land b
\end{gather}
A lattice is distributive if for all members $a$, $b$ and $c$ of the
partially ordered set,
\begin{gather}
\label{hadistand}a \land (b \lor c) = (a \land b) \lor (a \land c)\\
\label{hadistor}a \lor (b \land c) = (a \lor b) \land (a \lor c)
\end{gather}
The join $\biglor S$ of an arbitrary set S, if it exists, is the least
upper bound of the set.  Likewise, the meet $\bigland S$ would be its
greatest lower bound. A lattice is called complete if arbitrary joins
and meets exist.
\end{defn}

The meet and join operations serve to model conjunction 
and disjunction. Because of the way implications in goals can be
handled by \emph{shifting} or \emph{reindexing}, namely by augmenting
the program, \HOHH sequents do not strictly requite a semantic
implication operator, although the interpretations will remain (as we will show)  but we
need an additional operator corresponding to implication.  The
following definition of a Heyting (or pseudo-Boolean) algebra is taken
from \cite{TvD}.
\begin{defn}A Heyting algebra $\Omega$ is a lattice with a least element
$\bot$ and an operation $\harr$ defined on all pairs of elements of
$\Omega$ such  that
\begin{equation}\label{haarrow}
a \land b \leq c \textrm{ if and only if } a \leq b \harr c
\end{equation}
In a complete Heyting algebra, arbitrary meets and joins exist.
\end{defn}
In later sections, we will consider 
complete Heyting algebras, as well as those
 with only certain {\em parametrized\/}
infinite meets and joins
corresponding to universal and existential quantification.

It can be shown \cite{TvD} that a Heyting algebra is distributive
and that if $\biglor B$ exists then
\begin{equation}\label{hainfdist}
a \land \biglor B = \biglor \{ a \land b : b \in B\},
\end{equation}
an identity called $\wedge\bigvee$ distributivity.
It is clear from \eqref{haand} and \eqref{haarrow} that a Heyting
algebra has a top element ($a\harr a$ for any $a$), which we denote $\top$. 
The following characterization of the arrow connective will be useful below. 
Any complete lattice satisfying $\wedge\bigvee$ distributivity
\eqref{hainfdist} can be turned into a Heyting algebra by defining
\begin{equation}
\label{eq:haimpchar}
a\harr b := \bigvee \{x: x\wedge a\leq b\}.
\end{equation}
Furthermore, in any (not necessarily complete) Heyting algebra, the
supremum on the right hand side of \eqref{eq:haimpchar} exists, and
the identity holds.

The canonical example of a Heyting Algebra is the collection of open
sets in any topological space. 

\begin{expl}The open subsets of $\mathbb{R}$ form a Heyting algebra
  with $\lor$ denoting union, $\land$ the  intersection (and the
  interior of the intersection in the case of infinite joins),
$\leq$ the subset relation, $\bot$ the empty set and $a \harr b :=
\bigcup\{w: w \cap a \subseteq b \}$.  Then $\top$ is evidently
$\mathbb{R}$ and \eqref{hainfdist} holds, but the dual statement 
\begin{equation*}
a \lor \bigland B = \bigland \{ a \lor b : b \in B\}
\end{equation*}
need not:  let $a$ be $\mathbb{R} \setminus \{ 0 \}$ and $B$ be the
set of intervals $\{(-\frac{1}{n},\frac{1}{n}): n \in \mathbb{N}\}$.
Here $a \lor \bigland B = a \lor \emptyset = a$  but $\bigland \{ a
\lor b : b \in B\} = \bigland \{ \mathbb{R} : b \in B\} = \mathbb{R}$.
\end{expl}
The failed dual to $\wedge\biglor$-distributivity
 corresponds, in standard semantics for first-order logic,
 to the truth of the formula $A \lor \forall
x.B(x) \textrm{ iff } \forall x. (A \lor B(x))$, which is not an
intuitionistically valid principle, and 
therefore something we do not want to satisfy in our
models.  This formula corresponds in turn to 
Axiom $6^\alpha$ of Church's classical Theory of
Types \cite{church40}.

  \subsection{Logic Programming Proof procedures}
  \label{subsec:pp}
  
   We remind the reader that  a \emph{uniform proof} in \ICTT is a proof
  satisfying the following conditions:
  \begin{itemize}
  \item only program formulas $D$ may appear as antecedents in a
    sequent, and only goal formulas in the consequent, as defined in
    Definition~\ref{def:hoprograms},
  \item In the $\forall_L$ and $\exists_R$ rules, the witness $t$ must
    be a \emph{positive term} in $\hee$ ($\hhh_2$ in \cite{uniform}),
  \item Every sequent with a nonatomic consequent is the conclusion of
    a right rule.
  \end{itemize}
  In \cite{HHHorn} and \cite{nadathurPP} Miller and Nadathur define
  proof procedures for Higher-Order Horn and First-order Hereditarily
  Harrop Program-Goal pairs, and show them sound and complete with
  respect to Uniform Proofs. In \cite{nadathurPP}, Nadathur 
  extends the procedure to the higher-order (\HOHH) case.
  
  Our two formulations of a proof procedure  are based on theirs but
  have a somewhat different flavor, so we briefly discuss the
  relationship between them. Miller and Nadathur define a \emph{state}
  to be a triple $\ang{G,P,I}$ where $G,P$ are a program-goal pair and
  $I$ is a natural number bounding the labels (our levels) of all
  constants and variables appearing in $G,P$. Instead of our state
  \emph{vectors} they use a tuple $\ang{\cg,\cc,\cv, \cl,\theta}$ where 
  \cg\ is a \emph{set}  of states, \cc,\cv\  a set of constants and
  variables, \cl\ a \emph{labelling function} mapping $\cc\cup\cv$ to
  the set of natural numbers (to indicate their levels). In defining
  proof rules as a set of transitions between such tuples they have no
  need of a selection rule, since all such transitions involve a
  non-deterministic choice of an active member of \cg.

  \subsubsection{A sketch of the equivalence of the uniform proof system
    with the RES(t) system.}
\label{app-subsec:res(t)}

  Now we show that our derived resolution system RES(t) is sound and
  complete with respect to uniform proofs.
  \begin{thm}
    \label{thm:equivalence}
    Let ${\sf \ICTT}^c$ be the sequent calculus in
    Figure~\ref{fig:ictt}, with all constants $c$ drawn from the
    collection $\bigcup_\alpha \univ_i(\alpha)$ defined in
    Section~\ref{hosigs}. 
    Let $P$ be a program and $G$ a goal formula.  Then
    \begin{description}
    \item[(i)] 
    If the levels of $P$ and $G$ are less than or equal to $i \in \nn$,
    and sequent  $P\vdash G$ has a uniform proof there is a
    successful derivation in RES(t) of $\ipresgoal{i}{P}{G}$ with
    computed answer the identity substitution (with respect to the
    free variables in $P,G$).
  \item[(ii)]
     Let $\midres{\mfa}{\theta}{\NULL}$ be a resolution sequence in
  RES(t), and $\ipresgoal{i}{P}{G}$ a state in $\mfa$.
  Then there is a uniform proof
$P\theta\vdash_U G\theta$.
  \end{description}
\end{thm}

\begin{proof}[Proof \underline{of (i)}]
    Without loss of generality, all resolution sequences below will be
    assumed \emph{flat} (Lemma~\ref{lem:weak-lifting}).
    We induct on the \emph{depth} of a given uniform proof.
    Suppose $P\vdash G$ has a uniform proof. 
    If it is a single
    identity sequent, i.e. an occurrence of rule $Ax$ in
    Figure~\ref{fig:ictt}  with conclusion $\sequent{P,A}{A}$
    \gabox{Actually this case does not happen but is needed for inductive hypotehsis.}
    with $A$ atomic, then $\ipresgoal{i}{P,A}{A}$ resolves to
    $\ipresgoal{i}{P,A}{\top}$ in one step, and then to $\NULL$.
    We have a similar result if the proof is a single instance of
    $\top_R$.
    
    If the uniform proof has depth greater than 1, we need to consider
    the last rule used. If it is one of $\wedge_R$, $\vee_R$ or $\hoe_R$ this corresponds
    directly to a resolution step, and the induction hypothesis gives us the rest of the
    sequence to $\NULL$.

    If the last rule used is
    
    \case{$\exists_R$}
    this means there is an uniform proof of $P \vdash G[t/x]$ for some term  $t$ of an undetermined level. Consider $t'$ obtained from $t$ by replacing every variable or constant of level greater than $i$ with a fresh variable or constant of level $i$ or less. Then, by routine arguments we can prove that there is an uniform proof of $P \vdash G[t'/x]$. By inductive hypothesis, we have a successful derivation of $\ipresgoal{i}{P}{G[t'/x]}$. Hence, a single \textbf{instance} step suffices for our proof.
    
    \case{$\forall_R$}
    this means there is an
    uniform proof of $P \vdash G[c/x]$ where $c$ is a fresh constant
    of an undetermined level.  Then, by routine arguments we can prove
    that there is an uniform proof of $P \vdash G[c'/x]$, where $c'$
    is a consant of level $i+1$. By inductive hypothesis, we have a
    successful derivation of $\ipresgoal{i+1}{P}{G[c'/x]}$. Hence, a
    single \textbf{generic} step suffices for our proof.
   
    We consider now the left rules, when
    the consequent is atomic. Only $\wedge_L, \hoe_L,\forall_L$ apply
    in a uniform proof, since $\vee,\exists$ may not occur as
    principal connectives on the left.

    \case{$\wedge_L$}
    Suppose the given uniform proof ends with the following rule
    \[
      \infer[\wedge_L]
      {\sequent{P,D_1\wedge D_2}{G}}
      {\sequent{P,D_1,D_2}{G}}
    \]
    The induction hypothesis gives us a resolution sequence
    \begin{equation}
      \label{eq:commares}
      \resdots{\ipresgoal{i}{P,D_1,D_2}{G}}{id}{\NULL}.
    \end{equation}
    Since $\elab(P \cup \{D_1, D_2\}) = \elab(P \cup \{D_1 \wedge D_2\})$, by the left-weakening lemma (Lemma~\ref{lem:left-weakening}) we have
    \[
      \resdots{\ipresgoal{i}{P,D_1\wedge D_2}{G}}{id}{\NULL}.
    \]
    as we wanted to show.

    \case{$\hoe_L$}
    Suppose the last step in the given uniform proof is
    \[
      \infer[\hoe_L]
      {\sequent{P,G\hoe A}{K}}
      {\sequent{P}{G} & \sequent{P,A}{K}}
    \]
    with $K$ atomic.
    By the induction hypothesis we have RES(t) sequences
    \[
     \resdots{\ipresgoal{i}{P}{G}}{id}{\NULL} \qquad
     \resdots{\ipresgoal{i}{P,A}{K}}{id}{\NULL}.
   \]
   where in the two derivations, the only variables in common
   may be taken to be the ones in common in the two initial states.
   We now consider two cases in the {\bf second} resolution sequence above.
   Suppose {\bf backchain against the clause $A$ never occurs} in this
   sequence, i.e. any backchain step only makes use of
   $\elab(P)$. Then an easy induction on length shows that $A$ can be
   omitted, and by left weakening, we have
    \[
     \resdots{\ipresgoal{i}{P,G\hoe A}{K}}{id}{\NULL}.
   \]

   If  {\bf backchain against the clause $A$ does occur} using a
   substitution $\theta$ then 
  the resolution sequence has the following form where 
the {\bf first} occurrence of such a backchain step  is displayed, and
where weakening left has been used to add $G\hoe A $ to the program
\begin{multline*}
    \ipresgoal{i}{P,G\hoe A,A}{K} \derivation{\theta_1} 
       \mfa\otimes\ipresgoal{j}{P,Q,G\hoe A,A}{K'} 
       \otimes \mfb \derstep{\theta} \\
       \derstep{\theta}  \mfa\theta \otimes \ipresgoal{j}{P,Q,G\hoe A,A}{\top}  \otimes \mfb\theta  
       \derivation{\theta_2}\NULL
       \enspace ,
\end{multline*}
where $j \geq i$, $A\theta = K'\theta$ and $\theta_1\theta\theta_2$ is the
identity when restricted to the free variables of the original state vector.
Note that from $\mfa\theta \otimes \ipresgoal{j}{P,Q,G\hoe A,A}{\top}  \otimes \mfb\theta  \derivation{\theta_2}\NULL$ we may easily extract a derivation $ \mfa\theta \otimes \mfb\theta  \derivation{\theta_2}\NULL$ by removing a \textbf{null} step.


 We now
   \emph{replace the first backchain step against $A$
     with a backchain against the
     clause $G\hoe A$} to obtain
   {\small
     \[
     \cutrest{\ipresgoal{i}{P,G\hoe
         A,A}{K}}{\theta_1}{\mfa\otimes\ipresgoal{j}{P,Q,G\hoe A,A}{K'}
       \otimes \mfb}{\theta} {\mfa\theta\otimes
       \ipresgoal{}{P,Q\theta,G\hoe A,A}{G}
       \otimes \mfb\theta} 
  \]}

Now using the first sequence in the induction hypothesis, we have
$\resdots{\ipresgoal{i}{P}{G}}{id}{\NULL}$. By
  weakening left (Lemma~\ref{lem:left-weakening}) and level increase (Corollary~\ref{cor:level-increase}) we have
  $\resdots{\ipresgoal{j}{P,Q\theta,G\hoe
      A,A}{G}}{id}{\NULL}$.
  Thus we obtain a resolution sequence
  \[
   \resdots{\mfa\theta\otimes
       \ipresgoal{j}{P,Q\theta,G\hoe A,A}{G}
       \otimes \mfb\theta}{id}{\mfa\theta\otimes \mfb\theta} \enspace .
   \]
   Attaching $\resdots{\mfa\theta\otimes\mfb\theta}{\theta_2}{\NULL}$
   from above, we obtain $\resdots{\ipresgoal{i}{P,G\hoe
       A,A}{K}}{id}{\NULL}$.
     Repeating this process for all subsequent
     backchains against $A$ leaves us with a resolution sequence in
     which $A$ is never used (which can be established more formally
     by an easy induction on the length of the resolution sequence),
     and can therefore be removed.
  

     \case{$\forall_L$}
     Suppose the given uniform proof ends with
       \[
         \infer[\forall_L]
         {\sequent{P,\forall x D}{G}}
         {\sequent{P,D[t/x]}{G}}
       \]
       giving us, by the induction hypothesis,
       \[
         \resdots{\ipresgoal{i}{P,D[t/x]}{G}}{id}{\NULL}
       \]
       If no backchain step occurs using clauses in $\elab(D[t/x])$ in
       this proof, then it is straightforward to show, by induction on
       the length of the resolution sequence, that
       \[
         \resdots{\ipresgoal{i}{P,\forall x D}{G}}{id}{\NULL}.
       \]
       Otherwise, suppose the first backchain step using a
       clause in $\elab(D[t/x])$, say $\forall \vec{z}(H\hoe K')$,
       has the form
       \[
         \mfa\otimes\bareres{\ipresgoal{j}{P,D[t/x]}{K}\otimes
           \mfb}{\theta}{\mfa\theta\otimes\ipresgoal{j}{P,D[t/x]}{H'}\otimes\mfb\theta}
         \]
       with $\sigma$ the standardizing apart substitution from
       bound variables into
       fresh variables, $\theta(\sigma(H)) = H'$ and
       $\theta(\sigma(K')) = K$ ($\theta$ the identity on the free
       variables of $P, D[t/x], K$ since we are assuming the derivation to be flat).
       By Lemma~\ref{lem:elab-shift}
       there is a clause $\forall
       x\forall\vec{z}(\hat{H}\hoe\hat{K})\in \elab(\forall x D)$ with 
       $\hat{H}[t/x] = H$ and $\hat{K}[t/x] = K'$.
       For a fresh variable $w$, let the substitutions
       $\theta',\sigma'$ be defined by
       $\sigma'(x) = w$ and $\theta'(w) = t$
       and otherwise identical to $\sigma,\theta$.
       Then the following is a legal backchain step
         \[
         \mfa\otimes\bareres{\ipresgoal{j}{P,\forall x D}{K}\otimes
           \mfb}{\theta'}{\mfa\theta\otimes\ipresgoal{j}{P,\forall x
             D}{\hat{H}\sigma'\theta'}\otimes\mfb\theta} 
       \]
       with $\hat{H}\sigma'\theta' = H'$. Continuing this way with
       every relevant backchain step we obtain a successful resolution
       sequence of the form 
        \[
         \resdots{\ipresgoal{i}{P,\forall x D}{G}}{id}{\NULL}.
        \] \qedhere
     \end{proof}

     \begin{proof}[Proof \underline{of (ii)}]
       Part (ii) of this theorem is established, in a different setting,
       in the proof of Theorem 13 of
\cite{nadathurPP}, which shows that given a closed program formula $P$
and a goal $G$, if there is a derivation (using the proof procedure of
that paper) of $G$ relative to $P$, there is a substitution $\theta$
(the associated answer substitution) such that for any instance $G'$ of
$G\theta$ there is an \emph{intuitionistic} proof of $P\vdash
G'$. By \cite{uniform} the existence of an intuitionistic  proof for a
program-goal pair implies the existence of a uniform one. Our claim is
weaker, since we allow arbitrary unifiers and arbitrary instantiations
of existentially quantified variables. A direct proof is given here to
show that our variant proof system RES(t) is also sound with respect
to uniform proofs.

We induct on the length of the given RES(t) derivation. The base cases are easy and left to the reader.
%
%
%
%
  We also note that if 
   we first assume $\ipresgoal{i}{P}{G}$ in $\mfa$
is \emph{not} the selected goal,  the sequence looks like this: 
\[
\mfa_1\otimes
\ipresgoal{i}{P}{G}\otimes\mfa_2\resub{\theta_0}\mfa_1\theta_0\otimes
\ipresgoal{i}{P\theta_0}{G\theta_0}\otimes\mfa_2\theta_0\residots{\theta_1}\NULL 
\]
possibly with identity substitution $\theta_0$.
By induction hypothesis, letting $\theta = \theta_0\theta_1$ we have
$P\theta \vdash_U G\theta$.

  Suppose the claim holds for sequences of length less than $\ell$ and
  consider all possible first steps where $\ipresgoal{i}{P}{G}$ is the
  selected goal.
  
  \begin{description}
  \item[null] Assume the resolution sequence has the form
  \[
   \mfa \otimes \ipresgoal{i}{P}{\top} \otimes \mfb \derstep \mfa \otimes \mfb \derivation{\theta} \NULL
  \]
  We have $P \vdash_U \top$ by the $\top_R$ rule in \ICTT.
\item[true] Suppose the resolution sequence has the following form:
  \[
    \mfa \otimes \ipresgoal{i}{P}{A} \otimes \mfb \derstep\varphi \mfa\varphi \otimes \ipresgoal{i}{P\varphi}{\top} \otimes \mfb\varphi \derivation{\eta} \NULL  
  \]
  where $A\varphi = \top$ and $\theta$ is $\varphi\eta$ restricted to
  the free variables in the first state vector.
  Inductive hypothesis immediately gives an uniform proof of $P \varphi\eta \vdash \top\eta$, that is $P \theta \vdash_U \top$.

\item[and] is a straightforward case: if the resolution sequence is
   \[
         \mfa_1\otimes \ipresgoal{i}{P}{G_1\wedge G_2}\otimes\mfa_2
         \resub{\wedge}
         \mfa_1\otimes \ipresgoal{i}{P}{G_1}\otimes
         \ipresgoal{i}{P}{ G_2}\otimes\mfa_2 \residots{\varphi}
           \NULL
 \]
         the induction hypothesis gives $P\varphi\vdash_U G_1\varphi$
         and
         $P\varphi\vdash_U G_2\varphi$, so by $\wedge$-right
         $P\varphi\vdash_U (G_1\wedge G_2)\varphi$
\item[backchain] Suppose the resolution sequence has the following form:
  \[
    \mfa \otimes \ipresgoal{i}{P}{A'} \otimes \mfb \derstep{\varphi} \mfa\varphi \otimes \ipresgoal{i}{P\varphi}{G\varphi} \otimes \mfb\varphi\derivation{\eta}  \NULL 
  \]
  where $\forall \vec x(G\hoe A_r) \in \elab(P)$, $\delta$ is the renaming-apart substitution and $A'\varphi = A_r\delta\varphi$.
  By the induction hypothesis and the identity rule we have
  \[
    P\varphi\eta \vdash_U G\varphi\eta \mbox{ and } P\varphi\eta, A_r\delta\varphi\eta
  \vdash_U A'\varphi\eta
\]
so by $\hoe$-left
\[
  P\varphi\eta,G\varphi\eta\hoe A_r\delta\varphi\eta \vdash_U A'\varphi\eta.
\]
By the $U$-contraction lemma, we have $P\varphi\eta\vdash_U
A'\varphi\eta$ hence $P\theta \vdash_U A'\theta$.
\item[instance] Suppose the resolution sequence has the following form
 \[
    \mfa \otimes \ipresgoal{i}{P}{\exists x G} \otimes \mfb \derstep{\exists} \mfa \otimes \ipresgoal{i}{P}{G[t/x]} \otimes \mfb \derivation{\eta} \NULL  
    \]
    with $t$ in ${\cal U}_i$. 
    By the induction hypothesis $P\eta\vdash_U G[t/x] \eta$,
    hence  $P\eta\vdash_U (G\eta)[t\eta/x]$. By $\exists_R$
    we have $P\eta\vdash_U \exists x (G\eta)$ so
    $P\eta\vdash_U (\exists x G)\eta$.
  \item[augment] Easy application of $\hoe_R$, left to the reader.
  \item[generic] Suppose the resolution sequence has the following form:
  \[
    \mfa \otimes \ipresgoal{i}{P}{\forall x G} \otimes \mfb \derstep{\forall} \mfa \otimes \ipresgoal{i+1}{P}{G[c_{i+1}/x]} \otimes \mfb \derivation{\eta} \NULL  
    \]
  By the induction hypothesis, there is a proof of $P \eta \vdash_U G[c_{i+1}/x]\eta$ where $c_{i+1}$ is fresh. Since $\eta$ is level preserving and $G$ is of level $i$ or less, the constant $c_{i+1}$ cannot appear in $G\eta$. Therefore, a single application of the $\forall_R$ rule gives $P \vdash_U \forall x G$. \qedhere
  \end{description}
\end{proof}




{\footnotesize
\bibliography{newtau}}

\end{document}
